\documentclass[runningheads]{llncs}

\usepackage[utf8]{inputenc}
\usepackage{graphicx}
\usepackage[T1]{fontenc} 
\usepackage{url}
\usepackage{ifthen}
\usepackage{enumitem,wrapfig}
\usepackage[cmex10]{amsmath}
\usepackage{amssymb}
\usepackage{cite}

\usepackage{amsthm}
\usepackage{textcomp}
\usepackage{siunitx}
\usepackage{color}
\usepackage[lofdepth,lotdepth]{subfig}
\usepackage[countmax]{subfloat}
\usepackage{cases}
\usepackage[font={small}]{caption}
\usepackage{tcolorbox}
\usepackage{fix-cm}
\usepackage{bbm}
\usepackage{hyperref}
\usepackage{algorithm}
\usepackage[noend]{algpseudocode}

\usepackage{xcolor}
\definecolor{azure}{rgb}{0.54, 0.17, 0.89}
\newcommand{\colorcomment}[1]{\Comment{ {\color{azure} #1}} }
\newcommand{\maincolorcomment}[1]{{\color{azure}// #1 } }
\makeatletter
\def\thm@space@setup{\thm@preskip=2pt
\thm@postskip=2pt \itshape}
\makeatother
\newtheoremstyle{newstyle}      
{} 
{} 
{\mdseries} 
{} 
{\bfseries} 
{.} 
{ } 
{} 

\newtheorem{assumption}{Assumption}

\theoremstyle{definition}

\theoremstyle{remark}

\newcommand{\N}{{\mathcal{N}}}

\renewcommand{\H}{{\mathcal{H}}}
\newcommand{\A}{{\mathcal{A}}}
\newcommand{\C}{{\mathcal{C}}}

\newcommand{\Z}{{\mathcal{Z}}}

\newcommand{\T}{{\mathcal{T}}}
\newcommand{\cF}{{\mathcal{F}}}
\newcommand{\cN}{{\mathcal{N}}}

\renewcommand{\tt}{\mathcal{T}}
\newcommand{\zz}{\mathcal{Z}}

\newcommand{\I}{{\mathcal I}}
\newcommand{\E}{{\mathcal E}}
\newcommand{\W}{{\mathcal W}}
\newcommand{\J}{{\mathcal J}}

\makeatletter
\renewcommand*{\@fnsymbol}[1]{\ensuremath{\ifcase#1\or { }\or \dagger\or \ddagger\or
    \mathsection\or \mathparagraph\or \|\or **\or \dagger\dagger
    \or \ddagger\ddagger \else\@ctrerr\fi}}
\makeatother

\newcommand{\pv}[1] {{\textcolor{purple}{PV: #1}}}

\newcommand{\sk}[1] {{\textcolor{red}{SK: #1}}}

\newcommand{\vb}[1] {{\textcolor{azure}{}}}

\title{Proof-of-Stake Longest Chain Protocols:\\ Security vs Predictability }

\author{
Vivek Bagaria$^\star$,
Amir Dembo$^\star$,
Sreeram Kannan$^\ddagger$,
Sewoong Oh$^\ddagger$,
David Tse$^\star$,
Pramod Viswanath$^\dagger$,
Xuechao Wang$^\dagger$,
Ofer Zeitouni$^+$
\thanks{The authors are listed alphabetically. 
Email: vbagaria@stanford.edu,
amir@math.stanford.edu,  
ksreeram@uw.edu, sewoong@cs.washington.edu, 
dntse@stanford.edu, 
pramodv@illinois.edu,
xuechao2@illinois.edu,
ofer.zeitouni@weizmann.ac.il.
Amir Dembo and Ofer Zeitouni
were partially supported  by a US-Israel BSF grant.}
}

\institute{
$^\dagger$University of Illinois Urbana-Champaign,
$^\star$Stanford University,\\
$^\ddagger$University of Washington,
$^+$Weizmann Institute of Science
}

\begin{document}
\sloppy
\authorrunning{Bagaria et al.}

\maketitle
\begin{abstract}
    The Nakamoto longest chain protocol is remarkably simple and has been proven to provide security against any adversary with less than $50\%$ of the total hashing power. 
    Proof-of-stake (PoS) protocols are an energy efficient alternative; 
    however existing protocols adopting Nakamoto's longest chain design achieve provable security only by allowing long-term predictability, subjecting the system to serious bribery attacks. In this paper, we prove that a natural longest chain PoS protocol with similar predictability as Nakamoto's PoW protocol can achieve security against any adversary with less than $1/(1+e)$ fraction of the total stake. Moreover we propose a new family of longest chain PoS  protocols that achieve security against a $50\%$ adversary, while only requiring short-term predictability. Our proofs present a new approach to analyzing the formal security of blockchains, based on a notion of {\em adversary-proof convergence}.

    


\end{abstract}

\section{Introduction}
\label{sec:intro}

Bitcoin is the original blockchain, invented by Nakamoto. At the core  is the   permissionless consensus problem, which Nakamoto solved 
with a remarkably simple but powerful scheme known as the longest chain protocol. 
It uses only basic cryptographic primitives (hash functions and digital signatures). 
In the seminal paper \cite{bitcoin} that introduced the original Bitcoin protocol, 
Nakamoto also showed that the protocol is secure against one specific attack, a private double-spend attack, if the fraction of adversarial hashing power, $\beta$, is less than half the hashing power of the network. This  attack is mounted by the adversary trying to grow a long chain over a long duration in private to replace the public chain.  Subsequently, the security of Bitcoin against {\em all} possible attacks is proven in \cite{backbone}, and further extended to a more realistic network delay model in \cite{pss16}. 

The permissionless design (robustness to Sybil attacks) of Bitcoin is achieved via a proof-of-work (PoW) mining process, but comes at the cost of large energy consumption. Recently proof-of-stake (PoS)  protocols have emerged as an energy-efficient alternative. 
When running a lottery to win the right to propose the next valid block on the blockchain, each node wins with probability proportional its stake  in the total pool.   
This replaces the resource intense mining process of PoW, while ensuring fair chances to contribute and claim rewards. 

There are broadly two families of PoS protocols: those derived from decades of research in Byzantine Fault Tolerant (BFT) protocols and those  inspired by the Nakamoto longest chain protocol. Attempts at blockchain design via the BFT approach include Algorand \cite{chen2016algorand,gilad2017algorand} and  Hotstuff \cite{yin2018hotstuff}.
The adaptation of these new  protocols into blockchains is an active area of research and engineering \cite{gilad2017algorand,baudet2018state}, with large scale permissionless deployment as yet untested. 

Motivated and inspired by the time-tested  Nakamoto longest chain protocol are the PoS designs of  Snow White \cite{bentov2016snow} and the Ouroboros family  of protocols \cite{kiayias2017ouroboros,david2018ouroboros,badertscher2018ouroboros}. 
The inherent energy efficiency of the PoS setting comes with the cost of  enlarging the space of adversarial actions. 
In particular, the attacker can ``grind'' on the various sources of the randomness, i.e., attempt multiple samples from the sources of randomness to find a favorable one. Since these multiple attempts are without any cost to the attacker this strategy is  also  known as a {\em nothing-at-stake} (NaS) attack.   
One way to prevent an NaS attack is to 
 rely on a source of randomness on which a consensus has been reached.  In Snow White \cite{bentov2016snow} and the Ouroboros family \cite{kiayias2017ouroboros,david2018ouroboros,badertscher2018ouroboros}, 
 this agreed upon randomness  is derived from the stabilized segment of the blockchain from a few epochs before. {Each epoch is a fixed set of consecutive PoS lottery slots that use the same source of agreed upon randomness.} 
 However, this comes at a price of 
 allowing each individual node  to simulate and predict in advance whether it is going to win the PoS lottery at a given slot and add a new block to the chain. Further, as the size of each epoch is proportional to the security parameter $\kappa$ (specifically, a block is confirmed if and
only if it is more than $\kappa$ blocks deep in the blockchain), higher security necessarily implies that the nodes can predict further ahead into the future. 
 This is a serious security concern, as predictability makes a protocol vulnerable against other types of attacks 
 driven by incentives, 
 such  as predictable selfish mining or bribing attacks  \cite{brown2019formal}. 
 



{\bf Nakamoto-PoS}.  A straightforward PoS adoption of Nakamoto protocol, which in contrast to Ouroboros and Snow White can update randomness every block,  runs as follows; we term the protocol as Nakamoto-PoS. 
The protocol proceeds in discrete time units called {\it slots}, during which each node runs the ``PoS lottery'',  a  leader election with a winning probability proportional to the stake owned by the node --   winners get to {\em propose} new blocks. Each node computes 
 $hash = H({\rm time}, \text{secret key}, {\rm parentBk}.hash)$, where  the hash function  $H$ is a {\em verifiable random function} (VRF) (formally defined in \S \ref{sec:vrf}), which enables the nodes to run leader elections with their secret keys (the output $hash$ is verified with the corresponding public key). 
The node $n$ is elected a leader if $hash$ is smaller than a threshold $\rho \times {\sf stake}_n$, that is proportional to its stake ${\sf stake}_n$, then the node $n$  proposes  a new block consisting of time, ${\rm parentBk}.hash$, public key and $hash$, and appends it to the parent block.  
A detailed algorithmic description of this protocol is in \S \ref{app:pseudocode} (with $c=1$). 
Following Nakamoto's protocol, each honest node runs {\em only
one} election, appending to
the last block in the longest chain in its local view.
Having the hash function depend on 
parentBk.$hash$ ensures that every appended block provides a fresh source of randomness, for the following elections. However, there is no consensus on the randomness used and the randomness is block dependent, giving opportunities for the adversary to mount a NaS attack by trying its luck at many different blocks.


The analysis of the security of the Nakamoto-PoS  protocol is first attempted  in \cite{fan2018scalable}. Just like  Nakamoto's original analysis, their analysis is on the security against a specific attack: the private double-spend attack. Due to the NaS  phenomenon, they showed the adversary can grow a private chain faster than just growing at the tip, as though its stake increases by a factor of $e$. This shows that the PoS longest chain protocol is secure against the private double spend against if the adversarial fraction of stake $\beta < 1/(1+e)$. The question of whether the protocol is secure against {\em all} attacks, or there are attacks more serious than the private double spend attack, remains open. This is not only an academic question, as well-known blockchain protocols like GHOST \cite{ghost} had been shown to be secure against the private attack, only to be shown not secure later \cite{natoli2016balance,kiffer2018better}.

{\bf Methodological Contribution}. In this paper, we show that, under a formal security model (\S\ref{sec:model}),  the Nakamoto-PoS  protocol is indeed secure against {\em all} attacks, i.e., it has persistence and liveness whenever $\beta < 1/(1+e)$. One can view our result as analogous to what \cite{backbone} proved for Nakamoto's PoW protocol. However, {\em how} we prove the result is based on an entirely different approach. Specifically, the security proofs of \cite{backbone} are based on {\em counting} the number of blocks that can be mined by the adversary over a long enough duration (see Fig.~\ref{fig:race}), and showing that the longest chain is secure because the number of such adversarial blocks is less than the number of honest blocks whenever $\beta < 0.5$. This proof approach does not give non-trivial security results for the PoS protocol in question, because the number of adversarial blocks is exponentially larger than the number of honest blocks, due to the NaS phenomenon. Rather, our proof takes a dynamic view of the evolution of the blockchain, and shows that, whenever $\beta < 1/(1+e)$,  there are infinite many time instances, which we call {\em adversary-proof convergence times},  in which no chains that the adversary can grow from the past can ever catch up to the longest chain any time in the future   (see Fig.~\ref{fig:race}).\footnote{The notion of adversary-proof convergence is not to be confused with the notion of {\em convergence opportunities} in \cite{pss16}. The latter refers to times in which the honest blocks are not mined at similar time so that honest nodes may have a split view of the blockchain. It says nothing about whether the adversary can launch an attack using the blocks it mined. That is what our notion focuses on.}  
 Whenever such an event occurs, the  current longest chain will remain as a prefix of any future longest chain.

  \begin{figure}
    \begin{center}
    \includegraphics[scale=0.3]{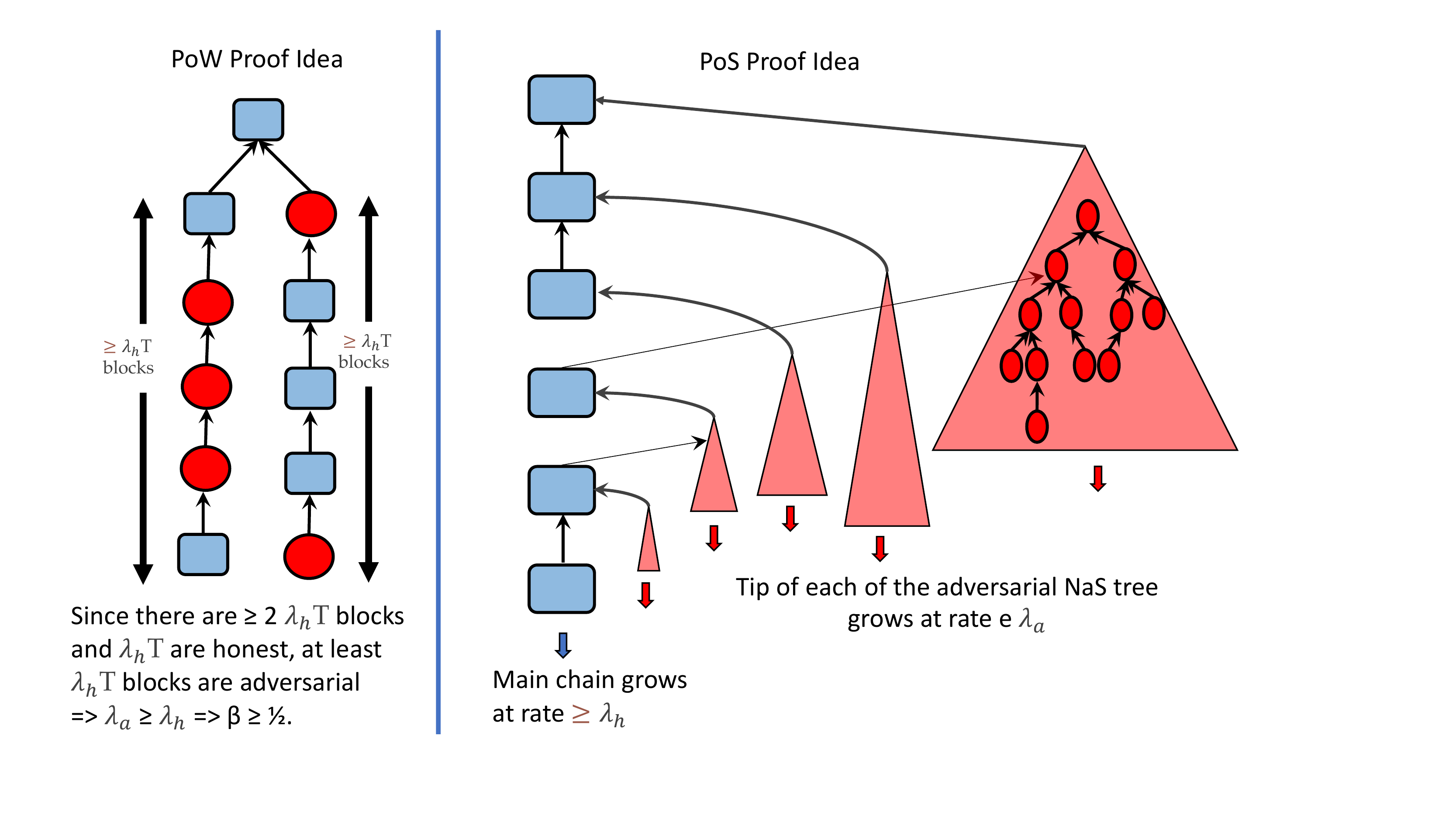}
    \end{center}
    \vspace{-0.1in}
  \caption{Notations: $\lambda_a, \lambda_h$ are the rates at which the adversary and the honest nodes can mine a block on a given block, and $T$ is the total duration.   {\bf Left}: In the PoW case, a counting argument shows that for the adversary to create a chain to match the longest chain, $\beta > 1/2$.  This proof fails to work in the PoS case because there is no conservation of work and the total number of adversarial blocks that can be generated over a time duration $T$ is exponentially larger than $\lambda_aT$. {\bf Right}: Our proof technique. Race between main chain and adversarial trees: adversary-proof convergence happens at a honest block if {\em none} of the previous NaS trees can ever catch up with the main chain downstream of the honest block. Security is proven by showing these events occur at a non-zero frequency.} \label{fig:race}
 \end{figure}

 Although the adversary can propose an exponentially large number of blocks, perhaps surprisingly, the protocol can still tolerate a {\em positive } fraction $\beta$ of adversarial stake. On the other hand, the fraction  that can be tolerated ($\frac{1}{1+e}$) is still less than the fraction  for the longest chain PoW protocol ($\frac{1}{2}$). In \cite{fan2018scalable} and \cite{fanlarge}, modifications of the longest chain protocol (called $g$-greedy and $D$-distance-greedy) are proposed, based on improvements to their security against the private double-spend attack. In  \S \ref{sec:attack}, we showed that, unlike the longest chain protocol, these protocols are subject to worse public-private attacks, and they not only do not exhibit true improvements in security than the longest chain protocol, but in many cases, they do far worse.  

{\bf New PoS Protocol Contribution}. Taking a different direction, we propose a new family of simple longest chain PoS protocols that we call $c$-Nakamoto-PoS (\S\ref{sec:protocols}); the fork choice rule remains the longest chain but the randomness update in the blockchain is controlled by a  parameter $c$, the larger the value of the parameter $c$, the slower the randomness is updated.  The common source of randomness used to elect a leader remains the same for $c$ blocks starting from the genesis and is
updated only when the current block to be generated is at a depth that is a multiple of $c$. When updating the randomness, the hash of that newly appended block is used as the source of randomness. The basic PoS Nakamoto protocol corresponds to $c = 1$, where the NaS attack is most effective. 
We can increase $c$ to gracefully reduce the potency of NaS attacks and increase the security threshold. 
To analyze the formal security of this family of protocols, we combine our  analysis for $c=1$ with results from the theory of branching random walks \cite{shi}; this allows us to characterize the largest adversarial fraction $\beta^*_c$ of stake that can be securely tolerated. As $c \rightarrow \infty$, $\beta^*_c \rightarrow 1/2$.  We should point out that the Ouroboros family of protocols \cite{david2018ouroboros,badertscher2018ouroboros} achieves security also by an infrequent update of the randomness; however, the update is much slower than what we are considering here, at the rate of once every constant multiple of $\kappa$, the security parameter. This is needed because the epoch must be long enough for the blockchain in the previous epoch to stabilize in order to generate the common randomness for the current epoch. Here, we are considering $c$ to be a fixed parameter independent of $\kappa$, and show that this is sufficient to thwart the NaS attack. 
Technically, we show that even if $c$ is small, there is no fundamental barrier to achieving any desired level of security $\kappa$. 
Hence, achieving a high level of security $\kappa$ should not come at the cost of longer  predictable window, and in this paper we introduce a natural adoption of Nakamoto protocol to achieve this. The practical implication of this result is shown in Fig.~\ref{fig:threshold_day}, where we can see that $c$-Nakamoto can achieve comparable security with a much smaller prediction window than a current implementation of Orouboros as part of the Cardano project \cite{stutz2020stake}. 

\begin{figure}
    \centering
    \includegraphics[width=0.7\textwidth]{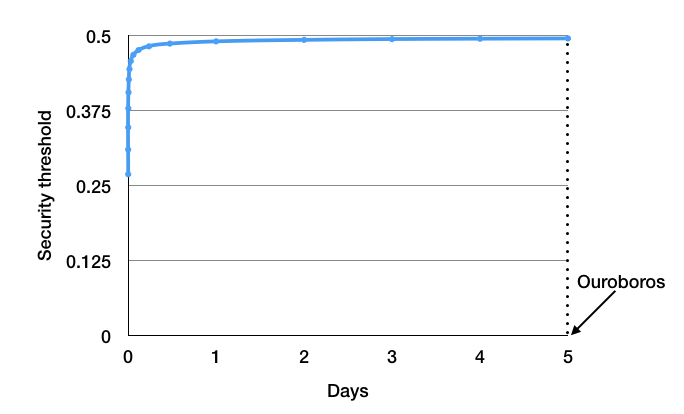}
    \caption{The security threshold $\beta^*_c$ of $c$-Nakamoto-PoS against the prediction window, equaling to $c$ times the inter-block time, which we set to be $20$s, to match the implementation of Orouboros in Cardano.  The Cardano project currently updates the common randomness every 5 days (21600 blocks, or $10 \kappa$), while the security threshold of $c$-Nakamoto-PoS can approach $1/2$ with much higher randomness update frequency.}
    \label{fig:threshold_day}
\end{figure}

 {\bf Organization}. In \S\ref{sec:prediction} we formally define prediction windows for any PoS protocol and discuss a class of bribery attacks enabled by the predictive feature of leader elections of many PoS protocols. The seriousness of these bribery attacks is underscored by the fatal nature of the attacks (double-spends and ledger rewrites) and  that they can be implemented by bribing a fraction of users possessing an arbitrarily small total stake. This sets the stage for the protocols discussed in this paper, Nakamoto-PoS and $c$-Nakamoto-PoS, which have low prediction windows and hence resistant to these bribery attacks. \S\ref{sec:model} discusses the formal security model we use to the analyze the  $c$-Nakamoto-PoS protocols (defined formally in \S \ref{sec:protocols}). The formal security analyses of  Nakamoto-PoS and $c$-Nakamoto-PoS protocols are conducted in \S\ref{sec:analysis}. This analysis, conducted in the static stake setting, is generalized to dynamic stake settings in \S\ref{sec:dynamic}.  

\section{Predictability in PoS Protocols}
\label{sec:prediction}

\begin{table}[ht]
\begin{center}
\begin{tabular}{ |c||c|c|c|c|c|c|c|c|c|c| } 
 \hline
 & 
 \multicolumn{6}{c|}{Longest Chain} & BFT & 
 \multicolumn{2}{c|}{Our results}
 \\\hline
 &  \cite{kiayias2017ouroboros}&
 \cite{david2018ouroboros}&  
 \cite{badertscher2018ouroboros}&  \cite{bentov2016snow} & \cite{fan2018scalable} & \cite{fanlarge} & Algorand \cite{chen2016algorand} & Nakamoto-PoS & $c$-Nakamoto-PoS  \\ \hline
 $W$ & $2\kappa$& $3\kappa$& $6\kappa$& $6\kappa$ & 1 &1& $\Theta(\kappa)$ & 1 & $c$ \\
 \hline
$\beta^*$ & 
\multicolumn{4}{c|}{$\frac{1}{2}$} &
\multicolumn{2}{c|}{$\mbox{unknown}$} & $\frac{1}{3}$  & $\frac{1}{1+e}$ & $\beta^*_c$  \\  \hline
\end{tabular}
\end{center}
\caption{ 
Our results decouple the prediction window $W$ and the security parameter $\kappa$, achieving any combination of $(W,\kappa)$. 
Prediction window $W$  for other PoS protocols are strongly coupled with  the  security parameter $\kappa = \log({1}/{P_{\rm failure}})$.  The maximum  threshold of adversarial stake that can be tolerated by the PoS protocols while being secure is $\beta^*$. Nakamoto-PoS is the most basic way of extending Nakamoto protocol to the PoS setting. This was originally introduced in \cite{fan2018scalable,fanlarge} but with an incomplete security analysis;
In \S\ref{sec:analysis} we show $\beta^*_c \approx {1}/{2} - \Theta(\sqrt{(1/c){\ln c}})$ and numerically tabulate $\beta^*_c$ (example: $\beta^*_c = 39\%$ for $c=10$).  
\cite{kiayias2017ouroboros,david2018ouroboros,badertscher2018ouroboros} are the Ouroboros family, \cite{bentov2016snow} is Snow White,  \cite{fan2018scalable} is $g$-Greedy, and  \cite{fanlarge} is $D$-Distance-Greedy. 
}
\label{tbl:intro}
\end{table} 

In PoW protocols such as  Bitcoin, no miner knows when they will get to propose the block until they solve the puzzle, and once they solve the puzzle, the block is inviolable (because the puzzle solution will become  invalid, otherwise). This causality is reversed in proof-of-stake (PoS) protocols: a node  eligible to propose a block knows {\em a priori} of its eligibility before proposing a block. This makes PoS protocols vulnerable to a new class of serious attacks not possible in the PoW setting. We briefly discuss these attacks here, deferring a detailed discussion to  \S\ref{app:prediction_bribery}.

\begin{definition}[$W$-predictable]
Given a PoS protocol $\Pi_{\rm PoS}$, let $\mathcal{C}$ be a valid blockchain ending with block $B$ with a time stamp $t$. We say a block $B$ enables $w$-length prediction, if
there exists a time $t_1>t$ and a block $B_1$ with a time stamp $t_1$ such that $(i)$ $B_1$ can be mined by miner (using its private state and the common public state) at time $t$; and 
$(ii)$ $B_1$ can be appended to $\mathcal{C'}$ to form a valid blockchain for any valid chain
$\mathcal{C'}$ that extends  $\mathcal{C}$ by appending  $w-1$ 
 valid blocks with time stamps  within the interval $(t,t_1)$. By taking the maximum over the prediction length over all blocks in $\Pi_{\rm PoS}$, we say $\Pi_{\rm PoS}$ is $W$-predictable.  $W$ is the size of the prediction window measured in units of number of blocks. 
\label{def:pred}
\end{definition}

 Informally, a longest-chain PoS protocol has $W$ length prediction window if it is possible for a miner to know that it is allowed to propose a block $W$ blocks downstream of the present blockchain. We note that our definition is similar to the definition of $W$-locally predictable protocols in \cite{brown2019formal}, where it has been pointed out that PoS protocols trade off between predictability and nothing-at-stake incentive attacks. In Table~\ref{tbl:intro}, we compare the prediction windows of various protocols. The longest-chain family of protocols of Ouroboros update randomness every epoch  have prediction window equal to the epoch length. Furthermore, they require the epoch length to be proportional to the security parameter $\kappa$, since one error event is that a majority of block producers in an epoch are not honest (and in that case, they can bias the randomness of all future slots).   

We note that in any PoS protocol with confirmation-depth $\kappa$ (the number of downstream blocks required to confirm a given block), a simple bribing attack is possible, where a briber requests the previous block producers to sign an alternate block for each of their previous certificates. However, such attacks are overt and easily detectable, and can be penalized with slashing penalties. If the prediction window $W$ is greater than the confirmation-depth  $\kappa$, then the following {\em covert} (undetectable) attack becomes possible. An adversary who wants to issue a double-spend can create a website where nodes that have future proposer slots post their leadership certificates for a bribe. If the adversary gets more than $\kappa+1$ miners to respond to this request, then the adversary can launch the following attack: (1) collect the $\kappa+1$ leadership certificates (2) issue a transaction that gets included in the upcoming honest block, (3) let the honest chain grow for $\kappa$ blocks to confirm the transaction and receive any goods in return, then (4) create a double spend against the previous transaction and (5) create a longer chain downstream of a block including the double-spend using the $\kappa+1$ certificates.  We note that this attack does not require participation from miners having a majority of the stake. Far from it, it only requires $\kappa+1$ out of the next $2 \kappa$ miners each holding a potentially infinitesimal fraction of stake.  Furthermore, this attack does not require miners to double-sign blocks, making it indistinguishable from unexpected network latency and providing plausible deniability for the miners who take the bribe. Thus, it is a serious covert attack on security requiring only participants with a net infinitesimal stake to participate in it. We note that this attack is not covered in the popular adaptive adversary model \cite{david2018ouroboros, gilad2017algorand}, since in that model, nodes are assumed not to have any agency and remain honest till the adversary corrupts them (based only on public state).

We note that it is {\em not possible} to mitigate the prediction issue by increasing the confirmation depth beyond the prediction window (which is equal to the epoch length). This is because the guarantees of existing protocols rely on the randomness of each epoch being unbiased and this guarantee fails  to hold when a majority of nodes in an epoch are bribed through the aforementioned mechanism in order to bias the randomness. 

While our discussion so far focused on longest-chain PoS protocols, we note that the prediction issue is even more serious in BFT based PoS protocols. PoS-based BFT protocols that work with the same committee or same proposer for many time-slots give raise to similar prediction based attacks.  Even in protocols such as Algorand \cite{chen2016algorand,gilad2017algorand}, which require a new committee for each step of the BFT protocol, the entire set of committees for all steps of the BFT protocol for a given block is known once the previous block is finalized. This leads to a similar type of bribing attack where once a $2/3$ majority of a BFT-step committee coordinate through a centralized website, they can sign a different block than the one the honest nodes agreed on. We note that since Algorand elects a small constant size committee (proportional to $\kappa$) for each round, a $2/3$ majority of the committee can comprise a negligible total stake.  Thus, in Algorand, even though the prediction window appears negligible, the confirmation delay is also small -- thus leading to the same type of attack (detailed discussion deferred to  \S\ref{app:prediction_bribery}).  A formal definition of  prediction window for BFT-based PoS protocols is in \S\ref{app:bftprediction} where we evaluate the prediction window $W$ for a canonical BFT based PoS protocol: Algorand \cite{chen2016algorand}.  There is a strong coupling between the security parameter and prediction window for Algorand, and is tabulated in Table~\ref{tbl:intro}. 

{\bf Summary}. We have demonstrated that both longest-chain and BFT based protocols are highly vulnerable to prediction-based security attacks when coordinating through an external bribing mechanism, thus compromising the persistence and liveness of the system. These attacks are covert, i.e., the deviant behavior is not detectable and punishable on the blockchain, and require only an infinitesimal fraction of the stake to collude, thus significantly weakening the security of the protocol. This motivates the study of Nakamoto-PoS (with a very small prediction window) and the design of a new PoS protocol that  has a prediction window much shorter than the confirmation depth and can be secure against adversaries with up to 50\% of the stake. This state of affairs, together with the security results we prove in this paper are summarized in Table~\ref{tbl:intro}. 

\section{Security Model}
\label{sec:model}

%
%
%
%
%
%
%
%
%


A blockchain protocol $\Pi$ is directed by an
\emph{environment} $\Z(1^\kappa)$, where $\kappa$ is the security parameter.
This environment
$(i)$ initiates a set of participating nodes $\N$;
$(ii)$ manages nodes through an adversary $\A$ which \emph{corrupts} a dynamically changing subset of nodes;
$(iii)$ manages all accesses of each node from/to the environment including broadcasting and receiving messages of blocks and transactions.

The protocol $\Pi$ proceeds in discrete time units called {\em slots}, each consisting of $\delta$ milliseconds (also called the slot duration), i.e.~the {\em time} argument in the input to the hash function should be in $\delta$ millisecond increments.
Each slot $sl_r$ is indexed by an integer $r \in \{1,2,\cdots\}$.
A {\em ledger} associates at most one block to each slot among those generated (or proposed) by participating nodes, each running a distributed protocol.
Collectively, at most one block per slot is selected to be included in the ledger according to a rule prescribed in the protocol $\Pi$. 
Similar to \cite{kiayias2017ouroboros}, we assume that the real time window for each slot satisfies that: (1)The time window of a slot is determined by a publicly-known and monotonically increasing function of current time; (2) every user has access to the current global time and any discrepancies between nodes’ local time are insignificant in comparison with the slot duration. 

We follow the security model of \cite{backbone,pss16,david2018ouroboros}
with an ideal functionality $\cF$.
This includes {\em diffuse} functionality and {\em key and transaction} functionality as described below.
With a  protocol $\Pi$,  adversary $\A$, environment $Z$, and security parameter $\kappa$, we denote by ${\sf VIEW}_{\Pi,\A,\Z}^{n,\cF}(\kappa)$ the view of a node $n\in \N$ who has access to an ideal functionality $\cF$.

We consider a semi-synchronous {\em network model} with bounded network delay similar to that of \cite{pss16,david2018ouroboros}
that accounts for adversarially controlled message delivery and immediate node corruption.
All broadcast messages are delivered by the adversary, with a bounded network delay $\Delta$ millisecond.
Let $\tau = \Delta/\delta$ be an integer.
We model this bounded network delay by allowing the adversary to selectively delay messages sent by honest nodes, with the following restrictions:
$(i)$ the messages broadcast in slot $sl_r$ must be delivered by the beginning of slot $sl_{r+\tau}$; and
$(ii)$ the adversary cannot forge or alter any message sent by an honest node.
This is the so called {\em delayed diffuse functionality} (denoted by {\sf DDiffuse}$_\tau$ in \cite{david2018ouroboros}).

The dynamically changing set of \emph{honest} (or uncorrupted) nodes $\H   \subseteq \N$  strictly follows the blockchain protocol $\Pi$.
The {\em key registration} functionality (from \cite{kiayias2017ouroboros}) is initialized
with the nodes $\cN$ and their respective stakes
$({\sf stake}_1,\ldots,{\sf stake}_{|\N|})$ such that the fraction of the initial stake owned by node $n$ is
${\sf stake}_n/\sum_{m\in \N} {\sf  stake}_m$.
At the beginning of each round, the adversary can dynamically
corrupt or uncorrupt any node $n\in\N$ , with a permission from
the environment $\Z$ in the form of a message
$({\sf Corrupt},n)$ or $({\sf Uncorrupt},n)$.
Even the corrupted nodes form a dynamically changing set, the total proportion of the adversarial stake is upper bounded by $\beta$ all the time.
For the honest nodes, the functionality can sample
a new public/secret key pair for each node and record them. For the corrupted nodes, if it is missing a public key, the adversary can set the node's public key, and
the public keys of corrupt nodes will be marked as such.\vb{The last two sentences are unclear: (i) when are these new key pairs sampled?, (ii) how does a (corrupted) node have a missing public key? when does the above actions take place? at the beginning of every round? }
When the adversary releases the control of a corrupted node, the node retrieves the current view of the  honest nodes at the beginning of the following round.

Any of the following actions are allowed to take place.
$(i)$ A node can retrieve its public/secret key pair from the functionality.
$(ii)$ A node can retrieve the whole database of public keys from the functionality.
$(iii)$ The environment can send a message ({\sf Create}) to spawn a new node\vb{do we require dynamically changing set $\mathcal{N}$? }, whose local view only contains the genesis block, and the functionality samples its public/secret key pair. 
$(iv)$ The environment can request a transaction, specifying its payer and recipient.
The functionality adjusts the stakes according to the transactions that make into the current ledger, as prescribed by the protocol $\Pi$.
The adversary has access to the state of a corrupt node $n$,
and will be activated in place of node $n$
with restrictions imposed by $\cF$.

\noindent{\bf Verifiable Random Function (VRF).}
Verifiable Random Functions (VRF), first introduced in \cite{vrf},
generates a pseudorandom number with a proof of its correctness.
A node with a secret key $sk$ can call {\sc VRFprove}$(\cdot,sk)$ to generates a pseudorandom {\em output}
$F_{sk}(\cdot)$ along with a {\em proof} $\pi_{sk}(\cdot)$. Other nodes that have the proof and the corresponding public key $pk$ can check that
the output has been generated by VRF, by calling {\sc VRFverify}$(\cdot,{\rm output},\pi_{sk}(\cdot),pk)$.
An efficient implementation of VRF was introduced in \cite{vrf2}, which formally satisfy Definition~\ref{def:vrf} in \S\ref{sec:vrf}.
This ensures that the output of a VRF is computationally indistinguishable from a random number even if the public key $pk$ and the function {\sc VRFprove} is revealed.

\noindent{\bf Key Evolving Signature schemes (KES).}
We propose using forward secure signature schemes \cite{bellare1999forward} to sign the transactions to be included in a generated block.
This prevents the adversary from
altering the transactions in the blocks mined in the past.
Efficient Key Evolving Signature (KES) schemes have been proposed in
\cite{itkis2001forward,david2018ouroboros} where keys are periodically erased and generated, while the new key is linked to the previous one.
This is assumed to be available to the nodes via the ideal functionality $\cF$. This ensures immutability of the contents of the blocks.
\section{Protocol description}
\label{sec:protocols}

We explain our protocol following   terminologies from \cite{david2018ouroboros}
and emphasize the differences as appropriate.  
The ideal functionality $\cF$ captures the resources available to the nodes  in order to securely execute the protocol.
When a PoS system is launched, a collection $\N$ of nodes are initialized.
Each node $n\in\N$ is initialized with a coin possessing stake ${\sf stake}_n$,
 a verification/signing key pair $({\sf KES}.vk_n,{\sf KES}.sk_n)$,
 and a public/secret key pair $({\sf VRF}.pk_n,{\sf VRF}.sk_n)$.
 The Key Evolving Signature key pair $({\sf KES})$ is used to sign and verify the content of a block,
 while the Verifiable Random Function key pair $({\sf VRF})$ is used to verify and elect leader nodes who generate new blocks.
All the nodes and the adversary know all public keys
$\{{\sf pk}_n=({\sf KES}.vk_n, {\sf VRF}.pk_n)\}_{n\in \N}$.
The {\em genesis} block contains all
public keys and initial stakes of all nodes, $ \{({\sf pk}_n,{\sf stake}_n)\}_{n\in \N}$, and also contains
a nonce in genesis.content.RandSource.
This nonce is  used as a seed for the randomness. 
The depth of a block in a chain is counted from the genesis (which is at depth zero).
We denote the time at the inception of the genesis block as zero (milliseconds),
such that the $i$-th slot starts at
 the time $\delta\cdot i$ milliseconds (since the inception of the genesis block).
 Nakamoto-PoS protocol is executed by the nodes and is assumed to run indefinitely. 
 At each slot a node starts with a local chain ${\mathcal C}$, which it tries to append new blocks on.

\medskip\noindent{\bf Proposer selection.}
At each slot,
a fresh subset of nodes are randomly elected to be the leaders, who
 have the  right to generate new blocks.
To be elected one of the leaders,
each node first decides on where to append the next block, in its local view of the blocktree.
This choice of a {\em parent block} is governed by the fork choice rule prescribed in the protocol.
For example, in {\sc BitCoin}, an honest node appends a new block to the deepest node in the local view of the blocktree.
This is known as Nakamoto protocol. 
We propose {\em $s$-truncated longest chain rule}
that includes the Nakamoto protocol as a special case,
which we define later in this section.

A random number of  leaders are elected in a single slot,
and the collective average block generation rate is controlled by a global parameter $\rho$ that is adaptively set by the ideal functionality $\cF$.
The individual block generation rate is proportional to the node's stake.
The stakes are updated continuously as the ledger is updated,  but only a coin $s$ blocks deep in the ledger can be used in the election
(the same parameter $s$ as used in the truncated longest chain rule),
and is formally defined  later in this section.

Concretely, at each slot, a node $n\in\N$  draws a number
distributed uniformly at random in a predefined range.
If this is less than the product of  its stake and a parameter $\rho$ (Algorithm \ref{alg:PoS} line \ref{algo:bias}),
the node is elected one of the leaders of the slot and
gains the right to generate a new block.
Ideally, we want to simulate such a random trial while ensuring that the outcome
$(i)$ is {\em verifiable} by any node after the block generation;
$(ii)$ is {\em unpredictable} by any node other than node $n$ before the generated block has been broadcast;  and
$(iii)$ is {\em independent} of any other events.
Verifiability in $(i)$ is critical
in ensuring consistency among untrusted pool of nodes.
Without unpredictability in $(ii)$, the adversary can easily take over the blockchain by
adaptively corrupting the future leaders.
Without independence in $(iii)$, a corrupted node might be able to grind on
the events that the simulator (and hence the outcome of the election) depends on,
until it finds one that favors its chances of generating future blocks.
Properties $(ii)$ and $(iii)$ are
challenges unique to PoS systems, as {\em predicting} and {\em grinding attacks} are  computationally costly in PoW systems.

To implement such a simulator
in a distributed manner among mutually untrusting nodes,
\cite{gilad2017algorand,david2018ouroboros} proposed using Verifiable Random Functions (VRFs), formally defined  in \S\ref{sec:vrf}.
In our proposed protocol, a node $n$ uses its secret key ${\sf VRF}.sk$ to generate a
pseudorandom {\em hash} and a {\em proof} of correctness (Algorithm~\ref{alg:PoS} line~\ref{algo:vrfproof}).
If node $n$ is elected a leader and broadcasts a new block, other nodes can verify the correctness with the corresponding public key {\sf VRF}.$pk$ and the proof (which is included in the block content).
This ensures unpredictability, as only node $n$ has access to its secret key, and
verifiability, as any node can access all public keys and verify that the correctness of the random leader election.

The pseudorandom hash generated by {\sc VRFprove}$(x,{\sf VRF}.sk)$,
depends on the external source of randomness, $(x,{\sf VRF}.sk)$, that is fed into the function.
Along with the secret key ${\sf VRF}.sk$, which we refer to as the {\em private} source of randomness,
we prescribe constructing a header $x$ that  contains the time (in a multiple of $\delta$ milliseconds) and a dynamically changing {\em common} source of randomness.
Including the time ensures that
the  hash is drawn
exactly once every slot.
Including the common source of randomness ensures that the random elections
cannot be predicted in advance, even by the owner of the secret key.
Such {\em private} predictability by the owner of the secret key leads to other security concerns that we discuss in~\S\ref{app:prediction_bribery}.

A vanilla implementation of such a protocol might
$(a)$ update stakes immediately and
$(b)$ use the hash of the previous block (i.e.~the parent of the newly generated block in the main chain as defined by the fork chain rule)
as the common source of randomness.
Each of these choices creates a distinct opportunity for an adversary to grind on, that could result in serious security breaches.
We explain the potential threats in the following
and propose how to update the randomness and the stake, respectively, to prevent each of the grinding attacks.
A formal analysis of the resulting protocol is provided in~\S\ref{sec:analysis}.

\medskip\noindent{\bf Updating the common source of randomness.}
One way to ensure unpredictability by even the owner of the secret key is to
draw randomness from the dynamically evolving blocktree.
For example, we could use the hash of the
parent block (i.e.~the block that a newly generated block will be appended to).
This hash depends only on the parent block proposer's
secret key, the time, and the source of randomness included in the header of the parent block.
In particular, this hash does not depend on the content of the parent block, to prevent an additional source of grinding attack.
However, such a frequent update of the source creates an opportunity for the adversary to grind on.
At every round, a corrupted node
can run as many leader elections as the number of blocks in the blocktree,
each appending to a different block as its parent.
To mitigate such grinding attacks,
we propose a new update rule for the source of randomness which we call {\em $c$-correlation}.

A parameter $c\in {\mathbb Z}$  determines how frequently we update.
The common source of randomness remains the same for $c$ blocks, and is
updated only when the current block to be generated is at a depth that is a multiple of $c$ (Algorithm~\ref{alg:PoS} line~\ref{algo:c-correlation}).
When updating, the hash of that newly appended block is used as the source of randomness.
When $c=1$, this recovers the vanilla update rule, where a grinding attack is most effective.
We can increase $c$ to gracefully increase the security threshold.
A formal analysis  is provided in~\S\ref{sec:analysis}.
When $c=\infty$, every block uses the nonce at the genesis block as the common source of randomness.
This makes the entire future leader elections predictable in private, by the owners of the secret keys.

\medskip\noindent{\bf Dynamic stake.}
The stake of a node $n$ (or equivalently that of the coin the node possesses) is
not only changing over {\em time}  as transactions are added to the blocktree,
but also over which {\em chain} we are referring to in the blocktree.
Different chains in the tree contain different sequences of transactions, leading to different stake allocations.
One needs to specify which chain we are referring to, when we access the stake of a node.
Such accesses are managed by the ideal functionality $\cF$ (Algorithm~\ref{alg:PoS} line~\ref{algo:stake}).

When running a random election to append a block to a parent block $b$ at depth $\ell-1$ in the blocktree,
a coin can be used for this election of creating a block with depth $\ell$ if and only if the coin is in the stake at the block with depth $\ell-s$ on the chain leading to block $b$.
Accordingly, a node $n$ has a winning probability proportional to ${\sf stake}_{n}(b)$ when mining on block $b$,
where ${\sf stake}_{n}(b)$
denotes the stake belonging to node $n$ as in the $(s-1)$-th block before $b$.
Starting from an initial stake distribution ${\sf stake}_n(b_{\rm genesis})$,
we add to or subtract from the stake according to all transactions
that $(i)$ involve node $n$ (or the coin that belongs to node $n$);
$(ii)$ are included in the chain of blocks from the genesis to the reference block $b$; and
$(iii)$ is included in the blockchain at least $s-1$ blocks before $b$.
Here, $s\in{\mathbb Z}$ is a global parameter. 

When $s$=1, the adversary can grind on  (the secret key ${\sf VRF}.sk$ of) the coin.
For example, once a corrupted node is elected as a leader at some time slot and proposed a new block,
it can include transactions in that block to transfer all stake to a coin that has a higher chance of winning the election at later time slots.
To prevent such a grinding on the coin, a natural attempt
is to use the stake in the block with depth $\ell-s$ when trying to create a block at depth $\ell$ on the main chain.  
However, there remains  a vulnerability, if we use the Nakamoto protocol  from  BitCoin as the fork choice rule.

Consider a corrupted node growing its own {\em private chain} from the genesis block (or any block in the blocktree).
A private chain is a blockchain that the corrupted node grows privately without broadcasting it to the network until it is certain that it can take over the
public blocktree.
Under the Nakamoto protocol,
this happens when the private chain is longer (in the number of blocks) than the longest chain in the public blocktree.
Note that the public blocktree grows at a rate proportional to $\rho$ and
the total stake of the nodes that append to the public blocktree.
With a  grinding attack, the private chain, which is entirely composed of the blocks generated by the corrupted node,
can eventually take over the public blocktree.

Initially, the private chain grows at a rate proportional to $\rho$ and the stake controlled by the corrupted node.
However, after $s$ blocks from the launch of the private chain,
the corrupted node can start grinding on the private key of the coin;
once a favorable coin is found,
it can transfer the stake to the favored coin by including transactions in the first ancestor block in the private chain.
This is possible as all blocks in the private chain belong to the corrupted node.
It can alter any content of the private chain and sign all blocks again.
With such a grinding attack (which we refer to as {\em coin grinding}), the corrupted node can potentially be elected a leader every slot in the private chain, eventually overtaking the public blocktree.
To prevent this private grinding attack, we propose using an $s$-truncation as the fork choice rule.

\medskip\noindent{\bf Fork choice rule.}
An honest node follows a fork choice rule prescribed in the protocol.
The purpose is to reach a consensus
on which chain of blocks to maintain, in a distributed manner.
Eventually, such chosen chain of blocks produces a final ledger of transactions.
Under the Nakamoto protocol, a node appends the next generated block to the longest chain in its local view of the blocktree.
Unlike PoW systems, Nakamoto protocol can lead to serious security issues for PoS systems as discussed above.
We propose using the following $s$-truncated longest chain rule, introduced in \cite{badertscher2018ouroboros,fan2018scalable}.

At any given time slot, an honest node keeps track of one main chain that it appends its next generated block to.
Upon receiving a new chain of blocks, it needs to decide which chain to keep.
Instead of comparing the length of those two chains, as in Nakamoto protocol,
we compare the creation time of the first $s$ blocks after the fork in truncated versions of those two chains (Algorithm~\ref{alg:PoS} line~\ref{algo:lc}).
Let $b_{\rm fork}$ be the block where those two chains fork.
The honest node counts how long it takes in each chain to create, up to
$s$ blocks after the fork.
The chain with shorter time for those $s$ blocks is chosen, and the next generated block will be appended to the newest block in that selected chain.
When $s=\infty$, the stake is fixed since the genesis block, which leads to a system that is secure but not adaptive. This is undesirable, as even
a coin with no current stake  can participate in block generation.
We propose using an appropriate
global choice of $s<\infty$, that scales linearly with the security parameter $\kappa$.
This ensures that the protocol meets the desired level of security,
while adapting to dynamic stake updates.
One caveat is that we only apply this $s$-truncation when comparing two chains that  both have at least $s$ blocks after those two chains forked. If one of the chain has less than $s$ blocks after forking, we use the longest chain rule to determine which chain to mine on. This is necessary in order to ensure that $s$-truncation is only applied to chains with enough blocks, such that our probabilistic analysis results hold.

\medskip\noindent{\bf Content of the block.}
Once a node is elected a leader,
all unconfirmed transactions in its buffer are added to the content (Algorithm~\ref{alg:PoS} line \ref{algo:sign}).
Along with the transactions, the content of the block also includes
the identity of the coin that won the election, and the $hash$ and $proof$ from {\sc VRFprove$(\cdot)$}.
This allows other nodes to verify the accuracy of the leader election.
A common source of randomness $RandSource$ is also included,
to be used in the next leader election.
The $state$ variable in the content contains the hash of parent block, which ensures that the content of the parent block cannot be altered.
Finally, the header and the content is signed with the forward secure signature ${\sf KES}.sk_n$.

Note that the content of the block is added after the leader election, in order to avoid any grinding on the content.
However, this allows the adversary to create multiple blocks with the same header but different content.
In particular,
after one leader election,
the adversary can create multiple blocks  appending to different parent blocks, as long as those parent blocks share the same common source of randomness.
Such copies of a block with the same header but different contents are
known as a ``forkable string'' in \cite{kiayias2017ouroboros} or ``non-core blocks'' in \cite{fan2018scalable}.
We show in the next section that the Nakamoto-PoS protocol is secure against all such variations of attacks.

\begin{table}[h]
\centering
\begin{tabular}{ |c l| }
\hline
    $\beta$ & total proportion of the adversarial stake\\
    $\delta$ & slot duration\\
    $\Delta$ & network delay \\
    $\kappa$ & security parameter  \\
    $c$ & correlation parameter  \\
    $s$ & parameter in the fork choice rule \\
    $\phi_c$ & maximum growth rate of a private tree \\
 \hline
\end{tabular}
\caption{The parameters used in our analysis.}
\label{tbl:notation}
\end{table}

\section{Security analysis}
\label{sec:analysis}

In this section, we provide formal security analysis for the longest chain PoS protocol. We focus on  $c=1$, and will discuss the general $c$ case in the last subsection and further in  \S\ref{app:growth_c}.  To simplify the expressions, we will consider the regime when the time slot duration $\delta$ is very small, so that the block generation processes can be modeled as Poisson. We will also assume the stake distribution is static in this section; the case of dynamic stake is discussed in \S \ref{sec:dynamic}.

We prove liveness and persistence of the protocol through  understanding when the longest chain {\em converges} as time passes, regardless of the adversarial strategy. We first analyze this convergence in the setting when the network delay $\Delta = 0$. This setting contains the core ideas of the proof, and allows us to explain it with the simplest notations. Then we extend it to the case of positive network delay. Finally, we use these results to prove high probability guarantees on the liveness and persistence of the protocol. 

\subsection{Warmup: $\Delta = 0$}
\label{sec: problem}

\subsubsection{Random processes}

In this setting, all the honest nodes have the same view of the blockchain, which can be modeled as  a random process $\{(\T(t), \C(t)): t \ge 0\}$. $\T(t)$ is a tree  and $\C(t)$ is the public longest chain at time $t$.  The tree $\T(t)$ is interpreted as consisting of all the blocks that are generated by both the adversary and the honest nodes up until time $t$, including blocks that are kept in private by the adversary. Note that $\T(t)$ consists of the honest blocks, mined at the tip of the longest chain, and all the blocks that the adversary {\em can} generate, by trying and winning the election lotteries at all possible locations of the blocktree. As such $\T(t)$ captures all the resources the adversary has at its disposal to attack at time $t$.

The longest chain protocol in \S \ref{sec:protocols} results in a process described as follows.

\begin{enumerate}
    \item $\T(0)= \C(0)$ is a single root block (the genesis block).

\item $\T(t)$ evolves as follows: there are independent Poisson processes of rate $\lambda_a$ at each block of $\T(t)$ (we call them the {\em adversary} processes), plus an additional independent Poisson process of rate $\lambda_h$ (we call it the {\em honest} process) arriving at the last block of the chain $\C(t)$, i.e. an aggregate Poisson process of rate $\lambda_a + \lambda_h$ at that block (the tip of the longest chain), and rate $\lambda_a$ at every other block of $\T(t)$. A new block is added to the tree at a certain block when a block is generated.  An arrival from the honest process is called an honest block. An arrival from the adversary process is called an adversarial block.

\item The chain $\C(t)$ is updated in two possible ways : 1) an additional honest block is added to $\C(t)$ if an arrival from the honest process occurs; 2) an adversary can replace $\C(t^-)$ by another chain $\C(t)$ from $\T(t)$ which is equal or longer in length than $\C(t^-)$.\footnote{All jump processes are assumed to be right-continuous with left limits, so that $\C(t),\T(t)$, etc include the new arrival if there is a new arrival at time $t$.}  The adversary's decision has to be based on the current state of the process. 

\end{enumerate}

The longest chain protocol means that the honest nodes always propose on the tip of the current public longest chain $\C(t)$ (at rate $\lambda_h$, proportional to their stake). The adversary can propose on any block (at rate $\lambda_a$, again proportional to its stake). The adversary can change where the honest nodes act by broadcasting an equal or longer length chain using the blocks it has succeeded in proposing. Since the adversary can change where the honest nodes can propose even with an equal length new chain, that means the adversary is given the ability to choose where the honest nodes propose when there are more than one longest public chain.

Proving the liveness and persistence of the protocol boil down to  providing a guarantee that the chain $\C(t)$ converges as $t \rightarrow \infty$ regardless of the adversary's strategy. We will show that this happens provided that $\lambda_a < \lambda_h/e$, i.e. 
$$\beta := \frac{\lambda_a}{\lambda_a + \lambda_h} < \frac{1}{1+e}$$
Our key contribution here is defining an appropriate notion of {\em adversary-proof convergence} and analyzing how frequently it occurs.

\subsubsection{Adversary-proof convergence }

We first define several basic random variables and random processes which are constituents of the processes $\T(\cdot)$ and $\C(\cdot)$. Then we will use them to define the notion of adversary-proof convergence event, and prove that indeed once the event occurs, convergence of the longest chain will occur regardless of what the adversarial strategy is.  

\begin{enumerate}
    \item $\tau_i$ = generation  time of the $i$-th honest block; $\tau_0 = 0$ is the generation time of the genesis block, $\tau_{i+1} - \tau_i$ is exponentially distributed with mean $1/\lambda_h$, i.i.d. across all $i$'s.  
    \item $A_h(t)$ = number of honest blocks generated from time $0$ to $t$. $A_h(t)$ increases by $1$ at each time $\tau_i$. $A_h(\cdot)$ is a Poisson process of rate $\lambda_h$.
    \item $L(t)$ is the length of $\C(t)$. $L(0) = 0$. Note that since the chain $\C(t)$ increments by $1$ for every honest block generation, it follows that for all $i$ and for all $t > \tau_i$,  
    \begin{equation} L(t)- L(\tau_i) \ge A_h(t) - A_h(\tau_i).
    \label{eq:lc}
    \end{equation}
    \item $\tt_i = \{\T_i(s) : s \ge 0\}$ is the random tree process generated by the adversary starting from the $i$-th honest block. $\T_i(0)$ consists of the $i$-th honest block and $\T_i(s)$ consists of all adversarial blocks grown on the $i$-th honest block from time $\tau_i$ to $\tau_i + s$. Note that the $\tt_i$'s are i.i.d. copies of the pure adversarial tree $\tt^a$ (i.e. the tree $\T(t)$ when $\lambda_h = 0$).
    \item $D_i(s)$ is the depth of the adversarial tree $\T_i(s)$. 
\end{enumerate}

Note that the overall tree $\T(t)$ is the composition of the adversarial trees $\T_0(t), \T_1(t-\tau_1), \ldots \T_i(t - \tau_i)$ where the $i$-th honest block is the last honest block that was generated before time $t$. How these trees are composed to form $\T(t)$ depends on the adversarial action on when to release the private chains.

Let us define the events:
\begin{equation}
    E_{ij} = \mbox{event that $D_i(t - \tau_i)  < A_h(t) - A_h(\tau_i)$ for all $t > \tau_j$}
\end{equation}
and 
\begin{equation}
F_j = \bigcap_{i = 0}^{j-1} E_{ij}.
\end{equation}

These events can be interpreted as about a fictitious system where there is a growing chain consisting of only honest blocks. The event $E_{ij}$ is the event that the adversarial tree rooted at the $i$-th honest block does not catch up with the honest chain any time after the generation of the $j$-th honest block. Such a tree can be interpreted as providing resource for a possible attack at the honest chain. If $E_{ij}$ occurs, then there is not enough resource for the $i$-th tree to attack after the $j$-th block. If $F_j$ occurs, there is not enough resource for {\em any} of the previous trees to attack the honest chain. 

Even though the events are about a fictitious system with a purely honest chain and the longest chain in the actual system may consist of a mixture of adversarial and honest blocks, intuitively the actual chain can only grow faster than the fictitious honest chain, and so we have the following key lemma. This lemma justifies us giving the name {\em adversary-proof convergence event} to $F_j$.

\begin{lemma}
\label{lem:regen}
If $F_j$ occurs, then the chain $\C(\tau_j)$ is a prefix of any future chain $\C(t)$, $t > \tau_j$. Equivalently, the $j$-th honest block will be in $\C(t)$ for all $t > \tau_j$. 
\end{lemma}

\begin{proof}
We will argue by contradiction. Suppose $F_j$ occurs and let $t^* > \tau_j$ be the smallest $t$ such that $\C(\tau_j)$ is not a prefix of $\C(t)$. Let $b_h$ be the last honest block on $\C(t^*)$ (which must exist, because the genesis block is by definition honest.) If $b_h$ is generated at some time $t_1> \tau_j$, then $\C(t_1^-)$ is the prefix of $\C(t^*)$ before block $b_h$, and does not contain $\C(\tau_j)$ as a prefix, contradicting the minimality of $t^*$. So $b_h$ must be generated before $\tau_j$, and hence $b_h$ is the $i$-th honest block for some $i < j$. The part of $\C(t^*)$ after block $b_h$ must lie entirely in the adversarial tree $T_i(t^*-\tau_i)$ rooted at $b_h$. Hence,
\begin{equation}
    D_i(t^*-\tau_i) < A_h(t^*)-A_h(\tau_i) \le L(t^*) - L(\tau_i),
\end{equation}
where the first inequality follows from the fact that $F_j$ holds, and the second inequality follows from the longest chain policy (eqn. (\ref{eq:lc})). From this we obtain that 
\begin{equation}
    L(\tau_i) + D_i(t^*-\tau_i) < L(t^*)
\end{equation}
which is a contradiction  since $L(t^*) \le L(\tau_i) + D_i(t^*-\tau_i)$.
\end{proof}

We will show that the adversary-proof convergence event occurs infinitely number of times if $\beta < 1/(1+e)$, and will also give an estimate on how frequently that happens. This will imply persistence and liveness of the protocol with high probability guarantees.

Since the occurrence of the event $F_j$ depends on whether the adversarial trees from the previous honest blocks can catch up with the (fictitious) honest chain, we next turns to an analysis of the growth rate of an adversarial tree.\footnote{The mean growth rate of this tree was analysed in \cite{fan2018scalable} using difference equations. Here, we are using the machinery of branching random walks, which not only gives us tail probabilities but also allow the extension to the $c$-correlated protocol, $c > 1$, easily.}


\subsubsection{The adversarial tree via branching random walks}

\label{sec-BRW}

The adversarial tree $\tt^a(t)$ is the tree $\T(t)$ when $\lambda_h = 0$, i.e. honest nodes not acting. Let the depth of the tree $\tt^a(t)$ be denoted by $D(t)$ and defined as the maximum depth of its blocks. The genesis block is always at depth $0$ and hence $\tt^a(0)$ has depth zero.

We give a description of the (dual of the)
adversarial tree in terms of a 
Branching Random Walk (BRW). Such a representation appears already in \cite{Drmota}, but we use here 
the standard language from, e.g., \cite{aidekon,shi}.

Consider the collection of $k$ tuples of positive integers, 
$\I_k =\{(i_1,\ldots,i_k)\}$, and  set $\I=\cup_{k>0} 
\I_k$. We consider elements of $\I$ as labelling the vertices of a rooted 
infinite tree, with $\I_k$ labelling the vertices at generation $k$ as follows:
the vertex $v = (i_1,\ldots,i_k)\in \I_k$ 
is the $i_k$-th child of vertex $(i_1,\ldots,i_{k-1})$
at level $k-1$.
An example of labelling is given in Fig.~\ref{fig:label_tree}.
For such $v$ we also let $v^j=(i_1,\ldots,i_j)$, 
$j=1,\ldots,k$, denote the 
ancestor of $v$ at level $j$, with $v^k=v$. 
For notation convenience, we set $v^0=0$ as the root of the tree.

\begin{figure}[h]
    \centering
\includegraphics[width=0.6\textwidth]{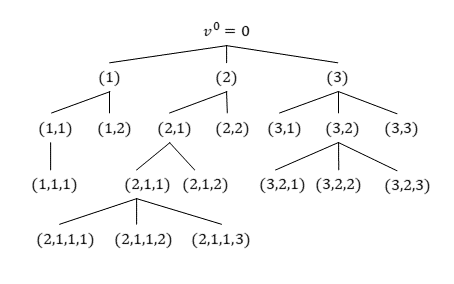}
\caption{Labelling the vertices of a rooted infinite tree.} \label{fig:label_tree}
\end{figure}

Next, let $\{\E_v\}_{v\in \I}$ be an i.i.d. family of exponential 
random variables
of parameter $\lambda_a$. For $v=(i_1,\ldots,i_k)\in \I_k$,
let $\W_v=\sum_{j\leq i_k} \E_{(i_1,\ldots,i_{k-1},j)}$ and
let $S_v=\sum_{j\leq k} \W_{v^j}$. This creates a labelled tree, with the following
interpretation: for $v=(i_1,\ldots,i_j)$,  the $W_{v^j}$ are the waiting
for $v^j$ to appear, measured from the appearance of $v^{j-1}$, and 
$S_v$ is the appearance time of $v$. A moments thought ought to convince the 
reader that the tree $S_v$ is a description of the adversarial tree,
sorted by depth. 

Let $S^*_k=\min_{v\in \I_k} S_v$. Note that $S^*_k$ is the time of appearance
of a block at level $k$ and therefore we have
\begin{equation}
  \label{eq-Ofer1}
  \{D(t)\leq k\}=\{S^*_k\geq  t\}.
\end{equation}

$S^*_k$ is the minimum of a standard BRW. Introduce, for $\theta<0$,
the moment generating 
function
\begin{eqnarray*}
\Lambda(\theta)&=&\log \sum_{v\in \I_1} E(e^{\theta S_v})=
\log \sum_{j=1}^\infty E (e^{\sum_{i=1}^j \theta \E_i})=
\log \sum_{j=1}^\infty (E(e^{\theta \E_1}))^j\\&=&
\log \frac{E(e^{\theta \E_1})}{1-E(e^{\theta\E_1})}.
\end{eqnarray*}
Due to the exponential law of $\E_1$,
$E(e^{\theta \E_1})= \frac{\lambda_a}{\lambda_a-\theta}$ and therefore
$\Lambda(\theta)=\log(-\lambda_a/\theta)$. 

An important role is played by $\theta^*=-e\lambda_a$, for which 
$\Lambda(\theta^*)=-1$ and 
$$\sup_{\theta<0} \left(\frac{\Lambda(\theta)}{\theta}\right)= \frac{\Lambda(\theta^*)}{\theta^*}=\frac{1}{\lambda_a e}=\frac{1}{|\theta^*|}.$$
Indeed, 
see 
e.g \cite[Theorem 1.3]{shi}, we have the following.
\begin{lemma}
  \label{prop-1}
$$\lim_{k\to\infty} \frac{S^*_k}{k}= \sup_{\theta<0} \left(\frac{\Lambda(\theta)}{\theta}\right)=\frac{1}{|\theta^*|}, \quad a.s.$$
\end{lemma}
In fact, much more is known, see e.g. \cite{hushi}. 
\begin{lemma}
  \label{prop-2}
  There exist explicit constants $c_1>c_2>0$ 
so that the sequence
$ S^*_k-k/\lambda_a e-c_1\log k$ is tight,
and
$$\liminf_{k\to \infty} S^*_k-k/\lambda_a e-c_2\log k=\infty, a.s.$$
\end{lemma}

Note that Lemmas \ref{prop-1}, \ref{prop-2}
and \eqref{eq-Ofer1} imply in particular
that  $D(t)\leq e\lambda_a t$ for all large $t$, a.s., and also that
\begin{equation}
  \label{eq-Ofer2}
\mbox{\rm  if
$e\lambda_a >\lambda_h$ then } D(t)> \lambda_ht\; \mbox{\rm for all large $t$, a.s.}.
\end{equation}

We will need also tail estimates for the event
$D(t)>e\lambda_a t+x$. While such estimates can be read from 
\cite{shi}, we bring instead a quantitative statement suited 
for our
needs.
\begin{lemma}
   \label{theo-tail}
  For $x>0$ so that $e\lambda_a t+x$ is
  an integer,
  \begin{equation}
    \label{eq:Ofertheo}
    P(D(t)\geq e\lambda_a t+x)\leq  e^{-x}.
  \end{equation}
\end{lemma}
\begin{proof}
  We use a simple upper bound. Write $m=e\lambda_a t+x$.
  Note that by \eqref{eq-Ofer1},
  \begin{equation}
    \label{eq-Ofer22}
    P(D(t)\geq m)=P(S^*_{ m}\leq t)
  \leq \sum_{v\in \I_{m}} P(S_v\leq t).
\end{equation}
  For $v=(i_1,\ldots,i_k)$, set $|v|=i_1+\cdots+i_k$. Then,
  we have that $S_v$ has the same law as 
  $\sum_{j=1}^{|v|} \E_j$. Thus, by Chebycheff's inequality,
  for $v\in \I_{m}$,
  \begin{equation}
    \label{eq-Ofer3}
    P(S_v\leq  t)\leq Ee^{\theta^* S_v} e^{-\theta^*t}=\left(\frac{\lambda_a}
  {\lambda_a-\theta^*} \right)^{|v|} e^{-\theta^* t}=\left(\frac{1}{1+e}\right)^{|v|} e^{e\lambda_a t}.
\end{equation}
  But 
  \begin{equation}
    \label{eq-Ofer4}\sum_{v\in \I_{m}}\left(\frac{1}{1+e}\right)^{|v|}
  =\sum_{i_1\geq 1,\ldots,i_{m}\geq 1} 
  \left(\frac{1}{1+e}\right)^{\sum_{j=1}^m i_j}=\left(\sum_{i\geq 1}\left(\frac{1}{1+e}\right)^i\right)^m= e^{-{m}}.
\end{equation}
Combining \eqref{eq-Ofer3}, \eqref{eq-Ofer4} and \eqref{eq-Ofer22}
yields \eqref{eq:Ofertheo}.
\end{proof}

\subsubsection{Occurrence of adversary-proof convergence }

If the growth rate of the adversarial tree is greater than $\lambda_h$, then the adversary can always attack the honest chain by growing a side chain at a rate faster than the honest chain's growth rate and replace it at will. (\ref{eq-Ofer2}) immediately shows that if  $\lambda_a  > \lambda_h/e$, i.e. when the adversarial fraction $\beta > 1/(1+e)$, the growth rate of the adversarial tree is at least $1$, and hence the private attack is successful. This is what \cite{fan2018scalable} showed. The question we want to answer is what happens when $\beta < 1/(1+e)$? Will another attack work? 
We show below that in this regime,  the adversary-proof convergence event $F_j$ has a non-zero probability of occurrence, and this implies that no attack works. 


\begin{lemma}
\label{lem:infinite_E}
If $\lambda_a < \lambda_h/e$, i.e. $\beta < 1/(1+e)$, then there exists a strictly positive constant $p> 0 $ such that $P(F_j) \ge p$ for all $j$. Also, with probability $1$, the event $F_j$ occurs for infinitely many $j$'s.
\end{lemma}

The proof of this result can be found in \S\ref{app:proof_0_delay}.

\subsubsection{Waiting time for convergence}
 
 In the previous section, we established the fact that the event $F_j$ has  $P(F_j) > p > 0$ for all $j$. This implies that the event $F_j$ occurs infinitely often. But how long do we need to wait for such an event to occur? We answer this question in the present section.

More specifically, we would like to get a bound on the probability that in a time interval $[s,s+t]$, there are no adversary-proof convergence events, i.e.  a bound on:
$$q[s,s+t] :=P(\bigcap_{j: \tau_j \in [s,s+t]} F_j^c).$$

\begin{lemma}
\label{lem:time-zero-delay}
If $\lambda_a< \lambda_h/e$, i.e. $\beta < 1/(1+e)$ then there exist constants $a_2,A_2$ so that, for any $s,t\geq 0$, 
\begin{equation}
\label{eq-satofer1}
q[s,s+t] \;\; \leq \;\; A_2\exp(-a_2\sqrt{t}).
\end{equation}
\end{lemma}
\noindent The bound in \eqref{eq-satofer1} is not optimal, see Remark \ref{rem-1} below.
\begin{proof}
Define $R_j = \tau_{j+1}- \tau_j$, and let 
\begin{equation}
    B_{ik} = \mbox{event that $D_i(\sum_{m = i}^{k-1} R_m) \; \ge \;  (k - i-1)
    $}.
\end{equation}
(Notation as in \eqref{eq-Bik}.)
Note that from Lemma~\ref{theo-tail} we have 
\begin{eqnarray}
\label{eq-PBik}
  P(B_{ik})
    &\le & P\left (B_{ik}|\sum_{m = i}^{k-1} R_m < (k-i-1) \frac{\lambda_h+\lambda_a e}{2\lambda_a e}\frac{1}{\lambda_h} \right)\nonumber\\
    &&\qquad \qquad + P\left(\sum_{m = i}^{k-1} R_m \geq (k-i-1) \frac{\lambda_h+\lambda_a e}{2\lambda_a e}\frac{1}{\lambda_h} \right) \nonumber\\
    & \le &  e^{-\frac{\lambda_h-\lambda_a e}{2\lambda_h}(k-i-1)} + A_1 e^{-\alpha_1 (k-i-1)}
\end{eqnarray}
for some positive  constants $A_1, \alpha_1$ independent of $k,i$. The first term in the last inequality follows from (\ref{eq:Ofertheo}), and the second term follows from the fact that $(\lambda_h+\lambda_a e)/(2\lambda_a e) > 1$ and the $R_i$'s are iid exponential random variables  of mean $1/\lambda_h$. 
Then 
\begin{equation}
    F_j^c = \bigcup_{(i,k): i < j, k >j} B_{ik}.
\end{equation}
Divide $[s,s+t]$ into $\sqrt{t}$ sub-intervals of length $\sqrt{t}$, so that the $r$ th sub-interval is:
$$\J_r : = [s+  (r-1) \sqrt{t}, s+ r\sqrt{t}].$$

Now look at the first, fourth, seventh, etc sub-intervals, i.e. all the $r = 1 \mod 3$ sub-intervals. Introduce the event that in the $\ell$-th $1 \mod 3$th sub-interval,
an adversarial tree that is rooted at a honest block arriving in that sub-interval or in the previous ($0 \mod 3$) sub-interval catches up with a honest block in that sub-interval or in the next ($2 \mod 3$) sub-interval. 
Formally,
$$C_{\ell}=\bigcap_{j: \tau_j \in \J_{3\ell+1}}
\bigcup_{(i,k): \tau_j - \sqrt{t} < \tau_i < \tau_j, \tau_j < \tau_k < \tau_j +\sqrt{t} } B_{ik}.$$
Note that for distinct $\ell$, the events $C_\ell$'s  are independent. Also, we have
\begin{equation}
    \label{eq-qq2}
    P(C_{\ell})\leq P(\mbox{no arrival in $\J_{3\ell+1}$}) + 1-p < 1
    \end{equation}
for large enough $t$. 

Introduce the atypical events:
\begin{eqnarray}
    B &=& \bigcup_{(i,k): \tau_i \in [s,s+t] \mbox{~or~} \tau_k \in [s,s+t], i <k, \tau_k - \tau_i >  \sqrt{t}} B_{ik} \;, \text{ and }\\
    \tilde{B} &=& 
    \bigcup_{(i,k):\tau_i<s,s+t<\tau_k} B_{ik}\;.
\end{eqnarray}
The events $B$ and $\tilde{B}$ are the events  that an adversarial tree catches up with an honest block far ahead. 
Consider also the events
\begin{eqnarray}
D_1&=&\{\# \{i: \tau_i\in (s-\sqrt{t},s+t+\sqrt{t})\} >2\lambda_h t\} \label{D1}\\
D_2&=& \{ \exists i,k: \tau_i  \in (s,s+t), (k-i)<\sqrt{t}/2\lambda_h,
\tau_k-\tau_i>\sqrt{t}\} \label{D2}\\
D_3&=& \{ \exists i,k: \tau_k \in (s,s+t), (k-i)<\sqrt{t}/2\lambda_h,
\tau_k-\tau_i>\sqrt{t}\} \label{D3}
\end{eqnarray}
In words, $D_1$ is the event of atypically many honest arrivals in $(s-\sqrt{t},s+t+\sqrt{t})$ while $D_2$ and $D_3$ are the events that there exists an interval 
of length $\sqrt{t}$ with at least one endpoint 
inside $(s,s+t)$ with atypically small number of 
arrivals. Since the number of honest arrivals in $(s,s+t)$ is Poisson with parameter $\lambda_h t$, we have from the memoryless property of the Poisson process that  $P(D_1)\leq e^{-c_0t}$ for some constant $c_0=c_0(\lambda_a,\lambda_h)>0$.  On the other hand, using the memoryless property and a union bound, and decreasing $c_0$ if needed, we have
that $P(D_2)\leq e^{-c_0 \sqrt{t}}$. Similarly, using time reversal, $P(D_3)\leq e^{-c_0\sqrt{t}}$. Therefore,
again using the memoryless property of the Poisson process,
\begin{eqnarray}
P(B)&\leq & P(D_1\cup D_2\cup D_3)+ P(B\cap D_1^c\cap D_2^c\cap D_3^c)\nonumber\\
&\leq & e^{-c_0 t} + 2e^{-c_0\sqrt{t}}+\sum_{i=1}^{2\lambda_h t} \sum_{k: k-i>\sqrt{t}/2\lambda_h}
P(B_{ik})\leq c_1e^{-c_2\sqrt{t}},
\label{eq-thursday8}
\end{eqnarray}
where $c_1,c_2>0$ are constants that may depend on
$\lambda_a,\lambda_h$ and the last inequality is due to \eqref{eq-PBik}.
We next claim that there exists a constant $\alpha>0$ so that,
for all $t$ large,
\begin{equation}
    \label{eq-thursday1}
    P(\tilde B)\leq e^{- \alpha t}.
    \end{equation}
    Indeed, we have that
    \begin{eqnarray}
    P(\tilde B)&=&
    \sum_{i<k} \int_0^s P(\tau_i\in d\theta)
    P(B_{ik}, \tau_k-\tau_i>s+t-\theta)\nonumber \\
    &\leq & \sum_i \int_0^s 
    P(\tau_i\in d\theta) \sum_{k:k>i} P(B_{i,k})^{1/2}
    P(\tau_k-\tau_i>s+t-\theta)^{1/2}.
    \label{eq-thursday2}
    \end{eqnarray}
    By \eqref{eq-PBik}, there exists $c_3>0$ so that
    \begin{equation}
        \label{eq-thursday1bis}
    P(B_{i,k})\leq e^{-c_3(k-i-1)},
    \end{equation}
    while the tails of the
    Poisson distribution yield the existence of constants
    $c,c'>0$ so that
    \begin{equation}
        \label{eq-thursday3}
        P(\tau_k-\tau_i>s+t-\theta)=
        P(\tau_{k-i}>s+t-\theta)
        \leq \left\{
        \begin{array}{ll}
        1,& (k-i)>c(s+t-\theta)\\
        e^{-c'(s+t-\theta)},& (k-i)\leq c(s+t-\theta).
        \end{array}\right.
    \end{equation}
    Combining \eqref{eq-thursday1bis} with \eqref{eq-thursday3}
   yields that there exists a constant $\alpha>0$ so that
    \begin{equation}
        \label{eq-thursday4}
        \sum_{k: k>i} P(B_{i,k})^{1/2}P(\tau_k-\tau_i>s+t-\theta)^{1/2}
        \leq e^{-2\alpha(t+s-\theta)}.
    \end{equation}
    Substituting this bound in \eqref{eq-thursday2} and using that $\sum_i P(\tau_i\in d\theta)=d\theta$ gives
    \begin{eqnarray}
    \label{eq-thursday5}
    P(\tilde B)&\leq &
    \sum_{i} \int_0^s 
    P(\tau_i\in d\theta) e^{-2\alpha (t+s-\theta)}\nonumber\\
    &\leq& \int_0^s e^{-2\alpha (t+s-\theta)} d\theta
    \leq e^{-\alpha t},
    \end{eqnarray}
    for $t$ large, proving \eqref{eq-thursday1}.
    
    Continuing with the proof of the lemma, we have: 
\begin{eqnarray}
\label{eq-qqq2}q[s,s+t] &\leq& P(B)+P(\tilde B)+P(\bigcap_{\ell=0}^{\sqrt{t}/3} C_{\ell})
=
P(B)+ P(\tilde B)+(P(C_{\ell}))^{\sqrt{t}/3}\nonumber \\
&\leq& c_1 e^{-c_2\sqrt{t}}+ e^{-\alpha t} + (P(C_\ell))^{\frac{\sqrt{t}}{3}}
\end{eqnarray}
where the equality is due to independence, and in
the last inequality we used \eqref{eq-thursday8} and \eqref{eq-thursday1}. The lemma follows from
\eqref{eq-qq2}.
\end{proof}
\begin{remark}
\label{rem-1}
Iterating the proof above  (taking  longer 
blocks and using the bound of Lemma \ref{lem:time-zero-delay} to improve on $P(C_\ell)$ in \eqref{eq-qq2} by replacing $p$ with the bound from 
\eqref{eq-satofer1})
shows that \eqref{eq-satofer1} can be improved to the statement
that for any $\theta>1$ there exist constants $a_\theta,A_\theta$ so that, for any $s,t>0$,
\begin{equation}
\label{eq-satofer2}
q[s,s+t] \;\; \leq \;\; A_\theta\exp(-a_\theta t^{1/\theta}).
\end{equation}
\end{remark}

\subsection{Nonzero network delay: $\Delta > 0$}

We will now extend the analysis in the above subsection to the case of non-zero delay.

In the case of zero network delay, the power of the adversary is in the adversarial blocks that it can generate by winning lotteries. We show that if $\beta < 1/(1+e)$, regardless of the adversarial strategy, there will be adversary-proof convergence events happening in the system once in a while. When the network delay is non-zero, the adversary has the additional power to delay delivery of honest blocks to create split view among the honest nodes. In the context of the security analysis of Nakamoto's PoW protocol, the limit of this power is quantified by the notion of {\em uniquely successful} round in \cite{backbone} in the lock-step synchronous round-by-round model, and extended to the notion of {\em convergence opportunity} in \cite{pss16} in the semi-synchronous model. (This notion is further used in \cite{ren} to provide a simpler security proof for Nakamoto's protocol.)  They show that during these convergence opportunities, the adversary cannot create split view between honest nodes, because only one honest block is generated during a sufficiently long time interval. We combine our notion of adversary-proof convergence event with the notion of convergence opportunity to define a stronger notion of adversary-proof convergence event for the non-zero delay case.

\subsubsection{Random processes}

We consider the network model in \S \ref{sec:model} with bounded communication delay, where all broadcast blocks are delivered by the adversary with maximum delay $\Delta$.
With this network model, the evolution of the blockchain can be modeled as  a random process $\{(\T(t), \C(t), \T^{(p)}(t), \C^{(p)}(t): t \ge 0, 1 \leq p \leq n\}$, where $n$ is the number of honest nodes, $\T(t)$ is a tree, $\T^{(p)}(t)$ is an induced sub-tree of $\T(t)$ in the view of the $p$-th honest node at time $t$, and $\C^{(p)}(t)$ is the longest chain in the $p$-th tree.  Then let $\C(t)$ be the common prefix of all the local honest chains $\C^{(p)}(t)$ for $1 \leq p \leq n$.  The tree $\T(t)$ is interpreted as consisting of all the blocks that are generated by both the adversary and the honest nodes up until time $t$, including blocks that are kept in private by the adversary. The chain $\C^{(p)}(t)$ is interpreted as the longest chain in the local view of the $p$-th honest node at time $t$. The process is described as follows.

\begin{enumerate}
    \item $\T(0)= \T^{(p)}(0) = \C(0) = \C^{(p)}(0), 1\leq p\leq n$ is a single root block (the genesis block).

    \item $\T(t)$ evolves as follows: there are independent Poisson processes of rate $\lambda_a$ at each block of $\T(t)$ (we call them the {\em adversary} processes), plus an additional independent Poisson process of rate $\lambda^{(p)}_h$ (we call it the {\em honest} process) arriving at the last block of the chain $\C^{(p)}(t)$ (the tip of the local longest chain) for each $1 \leq p \leq n$, with $\sum_{p=1}^n \lambda^{(p)}_h = \lambda_h$. A new block is added to the tree at a certain block when an arrival event occurs at that node. An arrival from the honest process is called an honest block. An arrival from the adversary process is called an adversarial block.

    \item The sub-tree $\T^{(p)}(t)$ for each $1 \leq p \leq n$ is updated in three possible ways : 1) an additional honest block can be added to $\T^{(p)}(t)$ by the adversary if an arrival event of the honest process with the $p$-th honest node occurs; 2) a block (whether is honest or adversarial) must be added to $\C^{(p)}(t)$ if it appears in $\T^{(q)}$ for some $q \neq p$ at time $t-\Delta$; 3) the adversary can replace $\T^{(p)}(t^-)$ by another sub-tree $\T^{(p)}(t)$ from $\T(t)$ as long as $\T^{(p)}(t^-)$ is an induced subgraph of the new tree $\T^{(p)}(t)$. The adversary's decision has to be based on the current state of the process. 
    
    \item $\C^{(p)}(t)$ is updated as follows for each $1\leq p \leq n$: $\C^{(p)}(t)$ is the longest chain in the tree $\T^{(p)}(t)$ starting from the root block at time $t$. If there are more than one longest chain, tie breaking is in favor of the adversary.
    
    \item $\C(t)$ is updated as follows: $\C(t)$ is the common prefix of all the local honest chains $\C^{(p)}(t)$ for $1 \leq p \leq n$ at time $t$.

\end{enumerate}

The adversary can change where the honest nodes act by broadcasting an equal or longer length chain using the blocks it has succeeded in proposing. Since the adversary can change where the honest nodes can propose even with an equal length new chain, that means the adversary is given the ability to choose where the honest nodes propose when there are more than one longest public chain. Also the adversary has the ability to have one message delivered to honest nodes at different time (but all within $\Delta$ time).

\subsubsection{Adversary-proof convergence}

We first define several basic random variables and random processes which are constituents of the processes $\T(\cdot)$ and $\C(\cdot)$. We make use of the terminology in \cite{ren}.

\begin{enumerate}
    \item $\tau_i$ = generation time of the $i$-th honest block; $\tau_0 = 0$ is the mining time of the genesis block, $\tau_{i+1} - \tau_i$ is exponentially distributed with mean $1/\lambda_h$, i.i.d. across all $i$'s. 
    Suppose an honest block $B$ is generated at time $\tau_j$.  If $\tau_j - \tau_{j-1} > \Delta$, then we call $B$ is a {\it non-tailgater} (otherwise, $B$ is a {\it tailgater}). If $\tau_j - \tau_{j-1} > \Delta$ and $\tau_{j+1} - \tau_j > \Delta$, then we call $B$ is a {\it loner}. Note that non-tailgaters have different depths and a loner is the only honest block at its depth.
    \item $H_h(t)$ = number of non-tailgaters generated from time $0$ to $t$. 
    \item $L^{(p)}(t)$ is the length of $\C^{(p)}(t)$ for each $1 \leq p \leq n$. $L^{(p)}(0) = 0$. Note that since every non-tailgater appears at different depth in the block tree, it follows that for all $t > s+\Delta$,  
    \begin{equation} L^{(p)}(t)- L^{(p)}(s) \ge H_h(t-\Delta) - H_h(s).
    \label{eq:lc_growth}
    \end{equation}
    \item $\tt_i = \{\tt_i(s) : s \ge 0\}$ is the random tree process generated by the adversary starting from the $i$-th honest block. $\tt_i(0)$ consists of the $i$-th honest block and $\tt_i(s)$ consists of all adversarial blocks grown on the $i$-th honest block from time $\tau_i$ to $\tau_i + s$. Note that the $\tt_i$'s are i.i.d. copies of the adversarial tree $\T^a$.
    \item $D_i(s)$ is the depth of the adversarial tree $\tt_i(s)$. 
\end{enumerate}

Let us define the events:
\begin{equation}
    \hat{E}_{ij} = \mbox{event that $D_i(t - \tau_i)  < H_h(t-\Delta) - H_h(\tau_i)$ for all $t > \tau_j+\Delta$},
\end{equation}

\begin{equation}
    \hat{F}_j = \bigcap_{0 \leq i < j} \hat{E}_{ij},
\end{equation}

\begin{equation}
    U_j = \mbox{event that $j$-th honest block is a loner} = \{\tau_j - \tau_{j-1} > \Delta, \tau_{j+1} - \tau_j > \Delta\},
\end{equation}
and
\begin{equation}
    \hat{U}_j = \hat{F}_j \cap U_j.
\end{equation}

And we have the following lemma, which justifies calling the event $\hat{U}_j$ adversary-proof convergence event for the non-zero delay case. . 

\begin{lemma}
\label{lem:loner}
If $\hat{U}_j$ occurs, then the $j$-th honest block is contained in any future chain $\C(t)$ (i.e. in all local chains $\C^{(p)}(t), 1\leq p\leq n$), $t > \tau_j+\Delta$.
\end{lemma}

\begin{proof}
We will argue by contradiction. Suppose $\hat{U}_j$ occurs and let $t^* > \tau_j+\Delta$ be the smallest $t$ such that the $j$-th honest block is not contained in $\C(t)$. Let $b_h$ be the last honest block on $\C^{(p)}(t^*)$ (which must exist, because the genesis block is by definition honest.) If $b_h$ is mined at some time $t_1> \tau_j+\Delta$, then $\C^{(p)}(t_1^-)$ is the prefix of $\C^{(p)}(t^*)$ before block $b_h$, and does not contain the $j$-th honest block, contradicting the minimality of $t^*$. So $b_h$ must be mined before time $\tau_j+\Delta$. And since the $j$-th honest block is a loner, we further know that $b_h$ must be mined before time $\tau_j$, hence $b_h$ is the $i$-th honest block for some $i < j$. The part of $\C^{(p)}(t^*)$ after block $b_h$ must lie entirely in the adversarial tree $\tt_i(t^*-\tau_i)$ rooted at $b_h$. Hence, we have 
\begin{equation}
    D_i(t^*-\tau_i) < H_h(t^*-\Delta)-H_h(\tau_i) \le L^{(p)}(t^*) - L^{(p)}(\tau_i),
\end{equation}
where the first inequality follows from the fact that $\hat{F}_j$ holds, and the second inequality follows from the longest chain policy (eqn.~(\ref{eq:lc_growth})). From this we obtain that 
\begin{equation}
    L^{(p)}(\tau_i) + D_i(t^*-\tau_i) < L^{(p)}(t^*)
\end{equation}
which is a contradiction  since $L^{(p)}(t^*) \le L^{(p)}(\tau_i) + D_i(t^*-\tau_i)$.
\end{proof}

Note that, Lemma~\ref{lem:loner} implies that if $\hat{U}_j$ occurs, then the entire chain leading to the $j$-th honest block from the genesis is stabilized after the $j$-th honest block is seen by all the honest nodes.

\subsubsection{Occurrence of adversary-proof convergence}

\begin{lemma}
\label{lem:infinite_many_F}
If $\lambda_a < g/e \cdot \lambda_h $, i.e. $\beta < g/(g+e)$ with $g = e^{- \lambda_h \Delta}$, then there is a $p > 0$ such that  $P(\hat{U}_j) \ge p$ for all $j$. Also, with probability $1$, the event $\hat{U}_j$ occurs for infinitely many $j$'s.
\end{lemma}

The proof of this result can be found in \S\ref{app:proof_delay}.

\subsubsection{Waiting time for adversary-proof convergence}

We have established the fact that the event $\hat{U}_j$ has  $P(\hat{U}_j) > p > 0$ for all $j$.
In analogy to the zero-delay case, we would like to get a bound on the probability that in a time interval $[s,s+t]$, there are no adversary-proof convergence events, i.e.  a bound on:
$$\tilde{q}[s,s+t] :=P(\bigcap_{j: \tau_j \in [s,s+t]} \hat{U}_j^c).$$

The following lemma is analogous to Lemma \ref{lem:time-zero-delay} in the zero-delay case. 
Its proof is almost verbatim  identical, and will not be repeated here.
\begin{lemma}
\label{lem:time}
If $\lambda_a <g/e \cdot \lambda_h$, i.e. $\beta < g/(g+e)$, then
there exist constants $\bar a_2,\bar A_2$ so that for all $s,t\geq 0$,
\begin{equation}
\tilde{q}[s,s+t] \leq \bar A_2 \exp(-\bar a_2\sqrt{t}).
\end{equation}
\end{lemma}
\noindent The improvement mentioned in 
Remark \ref{rem-1} applies here as well.
\subsection{Persistence and Liveness}



We will now use Lemma \ref{lem:time} to establish the liveness and persistence of the basic longest chain PoS protocol.

In the absence of any adversary, each node will
contribute to the final ledger as many blocks as
their proportion of the stake. In the presence of an adversary, the chain quality property ensures that the contribution of the adversary is bounded.
When $\beta < g/(g+e)$, we show that these properties hold with high probability, as stated in the following theorem.

Our goal is to generate a transaction ledger that satisfies
{\em persistence}  and {\em liveness} as defined in  \cite{backbone}.
Together, persistence and liveness guarantees robust transaction ledger; honest transactions will be adopted to the ledger and be immutable.  

\begin{definition}[from \cite{backbone}]
    \label{def:public_ledger}
    A protocol $\Pi$ maintains a robust public transaction ledger if it organizes the ledger as a blockchain of transactions and it satisfies the following two properties:
    \begin{itemize}
        \item (Persistence) Parameterized by $\tau \in \mathbb{R}$, if at a certain time a transaction {\sf tx} appears in a block which is mined more than $\tau$ time away from the mining time of the tip of the main chain of an honest node (such transaction will be called confirmed), then {\sf tx} will be confirmed by
        all honest nodes in the same position in the ledger.
        \item (Liveness) Parameterized by $u \in \mathbb{R}$, if a transaction {\sf tx} is received by all honest nodes for more than time $u$, then all honest nodes will contain {\sf tx} in the same place in the ledger forever.
    \end{itemize}
\end{definition}

The main result is that common prefix, chain quality, and chain growth imply that the transaction ledger satisfies persistence and liveness.
\begin{theorem}
    \label{thm:public_ledger}
    Distributed nodes running Nakamoto-PoS protocol generates a transaction ledger
    satisfying {\em persistence} (parameterized by $\tau=\sigma$) and {\em liveness} (parameterized by $u=\sigma$) in Definition~\ref{def:public_ledger} with probability at least $1-e^{-\Omega(\sqrt{\sigma})}$.
\end{theorem}

\begin{proof}
We first prove persistence and then liveness.
\begin{lemma}[Persistence]
The public transaction ledger maintained by Nakamoto-PoS satisfies {\rm Persistence} parameterized by $\tau=\sigma$ with probability at least $1-e^{-\Omega(\sqrt{\sigma})}$.
\end{lemma}
\begin{proof}
For a chain $\mathcal{C}_t$ with the last block generated at time $t$, 
let $\mathcal{C}_t^{\lceil\sigma}$ be the
chain resulting from pruning a chain  $\C_t$ up to $\sigma$, by removing the last blocks at the end of the chain that were generated after time $t-\sigma$.
Note that $\mathcal{C}^{\lceil \sigma}$ is a prefix of $\mathcal{C}$, which we denote by $\mathcal{C}^{\lceil \sigma}\preceq \mathcal{C}$.

Let $\mathcal{C}_1$ be the main chain of an honest node $P_1$ at time $t_1$.
Suppose a transaction {\sf tx} is contained in $\mathcal{C}_1^{\lceil \sigma}$ at round $t_1$, i.e., it is confirmed by $P_1$.
Consider a main chain $\mathcal{C}_2$ of an honest node $P_2$ at some time $t_2 \geq t_1$.
The {\em $\sigma$-common prefix property} ensures that after pruning a longest chain, it is a prefix of all future longest chains in the local view of any honest node.
Formally, it follows that  $\mathcal{C}_1^{\lceil \sigma} \preceq  \mathcal{C}_2 $, which completes the proof.

We are left to show that the $\sigma$-common prefix defined below holds with a probability at least $1-e^{-\Omega{(\sqrt{\sigma}})}$. 
This is a variation of a similar property first introduced in \cite{backbone} for PoW systems. 
Ours is closer to a local definition of $k$-common prefix introduced in \cite{prism}, which works for a system running for an  unbounded time.

\begin{definition}[$\sigma$-common prefix]
We say a protocol and a corresponding confirmation rule have a {\em $\sigma$-common prefix property} at time $t$, if
in the view ${\sf VIEW}_{\Pi,\A,\Z}^{n,\cF}(\kappa)$ of a honest node $n$ at time $t$, $n$ adopts a longest chain $\mathcal{C}$, then any longest chain $\mathcal{C'}$ adopted by some honest node $n'$ at time $t'>t$ satisfies
$\mathcal{C}^{\lceil \sigma} \preceq  \mathcal{C'}$.
    \label{def:prefix}
\end{definition}

Let $\mathcal{C}_t$ denote the longest chain adopted by an honest node with the last node generated at time $t$.  
There are a number of honest nodes generated in the interval $[t-\sigma,t]$, each of which can be in $\mathcal{C}_t$, $\mathcal{C}_{t'}$, or neither. 
We partition the set of honest blocks  generated in that interval with three sets:
${\cal H}_t\triangleq\{H_j \in {\cal C}_t:\tau_j\in[t-\sigma,t]\}, {\cal H}_{t'}\triangleq\{H_j \in {\cal C}_{t'}:\tau_j\in[t-\sigma,t]\}$, and ${\cal H}_{\rm rest}\triangleq\{H_j \notin {\cal C}_t \cup {\cal C}_{t'} :\tau_j\in[t-\sigma,t]\}$, depending on which chain they belong to.

Suppose ${\cal C}_t^{\lceil \sigma} \not\preceq {\cal C}_{t'}$, and we will show that this event is unlikely.  
Under this assumption, we claim that none of the honest blocks generated in the interval $[t-\sigma,t]$ 
are {\em stable}, i.e.~for each honest block, there exists a time in the future (since the generation of that block) in which the block does not belong to the longest chain. 

Precisely, we claim that ${\cal C}_t^{\lceil \sigma} \not\preceq {\cal C}_{t'}$ implies that $F_j^c$ holds for all $j$ such that $\tau_j\in[t-\sigma,t]$. 
This in turn implies that $P({\cal C}_t^{\lceil \sigma} \not\preceq {\cal C}_{t'}) \leq P(\cap_{j:\tau_j\in[t-\sigma,t]} F_j^c)$. 
By Lemma~\ref{lem:time}, we know that 
the probability of this happening is low: $e^{-\Omega(\sqrt{\sigma})}$. 

This follows from the following facts. 
$(i)$ the honest blocks in ${\cal C}_t$ does not make it to the longest chain at time $t'$: $H_j \notin {\cal C}_{t'}$ for all  $H_j\in{\cal H}_t$. This follows from ${\cal C}_t^{\lceil \sigma} \not\preceq {\cal C}_{t'}$.
$(ii)$ the honest blocks in ${\cal C}_{t'}$ does not make it to the longest chain ${\cal C}_t$ at time $t$:  
$H_j \notin {\cal C}_{t}$ for all  $H_j\in{\cal H}_{t'}$. This also follows from ${\cal C}_t^{\lceil \sigma} \not\preceq {\cal C}_{t'}$.
$(iii)$ the rest of the honest blocks did not make it to either of the above: $H_j \notin {\cal C}_t \cup {\cal C}_{t'} $  for all $H_j\in {\cal H}_{\rm rest}$.



\end{proof}

We next prove liveness.
\begin{lemma}[Liveness]
The public transaction ledger maintained by Nakamoto-PoS satisfies {\rm Livenesss} parameterized by $u = \sigma$ with probability at least $1-e^{-\Omega(\sqrt{\sigma})}$.
\end{lemma}
\begin{proof}
Assume a transaction {\sf tx} is received by all honest nodes at time $t$, then by Lemma~\ref{lem:time}, we know that with probability at least $1-e^{-\Omega(\sqrt{\sigma})}$, there exists one honest block $B_j$ mined at time $\tau_j$ with $\tau_j \in [t,t+u]$ and event $F_j$ occurs, i.e., the block $B_j$ and its ancestor blocks will be contained in any future longest chain.
Therefore, {\sf tx} must be contained in block $B_j$ or one ancestor block of $B_j$ since $tx$ is seen by all honest nodes at time $t < \tau_j$.
In either way, $\sf tx$ is stabilized forever.
Thus, the lemma follows.
\end{proof}
\end{proof}

\subsection{Extending the analysis to general $c$}

We can repeat the analysis in the previous section for the $c$-Nakamoto-PoS for $c > 1$. Like for $c=1$, the analysis basically boils down to the problem of analyzing the race between the adversarial tree and the fictitious purely honest chain. In \S \ref{app:growth_c}, we confine our analysis to a study of that race. It turns out that under $c$-correlation randomness, the adversarial tree is developed by another branching random walking process, but with a slowing growth amplification factor $\phi_c$, such that $\phi_1 = e$ and $\phi_\infty = 1$. So the security threshold is
$$\beta_c = \frac{e^{-\lambda_h \Delta}}{e^{-\lambda_h \Delta} + \phi_c}$$
and it 
goes from 
$$ \frac{e^{-\lambda_h\Delta}}{e^{-\lambda_h\Delta} + e}$$
for $c = 1$ to 
$$ \frac{e^{-\lambda_h\Delta}}{e^{-\lambda_h \Delta} + 1}$$
for $c = \infty$. 
Under the $c$-correlation protocol, we have 
\begin{equation*}
    \phi_c = -\frac{c \theta_c^*}{\log(-\theta_c^*)+(c-1)\log(1-\theta_c^*)},
\end{equation*}
where $\theta_c^*$ is the unique negative solution of Eq.~(\ref{eqn:relation})
\begin{equation}
     -\log(-\theta_c) - (c-1)\log(1-\theta_c) = -1 +(c-1)\frac{\theta_c}{1-\theta_c}.
     \label{eqn:relation}
\end{equation}


We numerically compute the value of $\phi_c$ and $\beta^*_c$ with $\Delta = 0$ in the table below.

\begin{table}[h]
\begin{center}
\begin{tabular}{ |c|c|c|c|c|c|c|c|c|c|c| }
 \hline
 $c$      & 1 & 2 & 3 & 4 & 5 & 6 & 7 & 8 & 9 & 10\\ \hline
 $\phi_c$   & e & 2.22547  & 2.01030 & 1.88255  & 1.79545  & 1.73110 & 1.68103 & 1.64060 &1.60705 & 1.57860 \\
 \hline
 $\beta^*_c$ & $\frac{1}{1+e}$& 0.31003 & 0.33219 &0.34691 &0.35772 & 0.36615 & 0.37299 &
 0.37870 & 0.38358 & 0.38780
 \\ \hline
\end{tabular}
\end{center}
\caption{Numerically computed growth rate $\phi_c$ and stake threshold $\beta^*_c$ with $\Delta = 0$.}
\vspace{-1cm}
\end{table}

In the deployment of the Ouroboros protocol in the Cardano project\footnote{https://www.cardano.org}, each slot takes 20 seconds and each epoch is chosen to be 5 days~\cite{stutz2020stake}, that is a common randomness will be shared by 21600 blocks.
In Fig.~\ref{fig:threshold_day}, we plot the security threshold $\beta^*_c$ against $c$ up to $c=21600$ for our $c$-Nakamoto-PoS. It turns out the security threshold of $c$-Nakamoto-PoS can approach $1/2$ very closely even when $c$ is much less than $21600$, which is the current randomness update frequency in the Cardano project.

\section{Security analysis of the dynamic stake protocol} \label{sec:dynamic}

One major advantage of the Nakamoto protocol in the proof of work setting is its permissionless setting: anyone can join (leave) the system by simply contributing (extricating) computing power for the mining process.  Under the PoS setting  the stakes are used in lieu of the computing power. To support a protocol that is close to a permissionless system, we need to handle the case when the stake is  dynamically varying. Unlike the case of compute power, the entry/removal of stake has to be more carefully orchestrated:
if a change in the stake takes effect immediately after the transaction has been included in the blockchain, then this gives an opportunity for the adversary to grind on the secret key of the header (time, secret key, common source of randomness).
Concretely, once an adversary has generated a block, it can add a transaction
that moves all its stake to a new coin with a new pair of public and secret keys.
The adversary can keep drawing a new coin and simulating the next leader election, until it finds one that wins.
This is a serious concern as the adversary can potentially win all elections.

To prevent such a grinding on the coin,
we use $s$-truncation introduced in \cite{badertscher2018ouroboros,fan2018scalable}.
$s$-truncation has two components:
using stakes from ancestor blocks and a fork choice rule.
The winning probability of a leader election uses the stake computed at
an ancestor block which is $s$ blocks above
the parent block that is currently being mined on.
However, this allows an adversary to launch a long range attack, for any value of $s<\infty$. To launch a long range attack, an adversary grows a private block tree. Once it has grown for longer than $s$ blocks, then it can grind on the coin to win the election for the next block, and add the favorable transaction to the first ancestor block. As all blocks in this private tree are adversarial, the consistency of the transactions can be maintained by re-signing all intermediate blocks. This allows the adversary to win all elections after $s$ private blocks, and eventually take over the honest block chain.

The $s$-truncation longest chain rule works as follows.
When presented with a chain that forks from current longest chain, a node compares the two chains according to the following rule.
Both chains are truncated up to $s$ blocks after the forking. Whichever truncated chain was created in a shorter time (and hence denser) is chosen to be mined on. This ensures that honest block chains will be chosen over privately grown adversarial chains with long range attacks. A detailed description of the $c$-correlation, $s$-truncation, Nakamoto-PoS protocol is in \S\ref{sec:protocols}.

We extend Theorem \ref{thm:public_ledger} to show that the above $s$-truncation scheme is secure.

\begin{theorem}
   \label{thm:dynamic}
   Under the dynamic stake setting, 
   distributed nodes running Nakamoto-PoS protocol with a choice of $s=\Theta(\sigma)$ generates a transaction ledger satisfying {\em persistence} and {\em liveness} in  Definition~\ref{def:public_ledger} with probability at least $1-e^{-\Omega( \sqrt{\sigma})}$.
\end{theorem}

\noindent{\bf Proof sketch.} We prove it in three steps. First, we show that with static stake setting, common prefix, chain growth, and chain quality properties still hold for $s$-truncation protocol with $s=\Theta(\sigma)$. Second, we show that with dynamic stake setting and $s$ stake update rule, the adversary can mine a private chain with consecutive $s$ blocks that is denser than the public main chain only with a negligible probability. 

We show in \S\ref{sec:analysis} that the choice of $c$ only determines how many fraction of adversary the system can tolerate and not the security parameter $\kappa$. On the other hand, the choice of $s$ is critically related to the target security parameter $\kappa$ 
as we show in 
\S\ref{sec:dynamic}. 
By decoupling these two parameters $c$ and $s$ in the protocol, we can achieve any level of predictability (with an appropriate choice of $c$), while managing to satisfy any target security parameter $\kappa$ 
(with an appropriate choice of $s$).


\section{Conclusion}
\label{sec:conclusion}
  We proposed a new family of PoS protocols and proved that our proposed PoS protocols  have low  predictability while guaranteeing security against adversarial nothing-at-stake attacks. We did not discuss the design of incentive systems that encourage users to follow the honest protocol. We point out that there are existing ideas that could be naturally adapted for our problem. For example, one way to minimize NaS attacks is to require users to deposit stake that can be slashed if the node has a provable deviation (for example, double-signing blocks) \cite{bentov2016snow,buterin2017casper}. Another important idea is that fruitchain-type incentive mechanisms \cite{fruitchains} which protect against selfish mining in PoW can be ported to PoS protocols \cite{kiayias2017ouroboros}.
  However, as pointed out in \cite{brown2019formal}, the full problem of designing PoS protocols that strongly (instead of weakly) disincentivze NaS and selfish-mining attacks remains an important direction of future research. Finally, a detailed  mathematical modeling of bribing attacks that consider interaction between the blockchain and external coordination mechanisms also remains open. 
  






\bibliographystyle{acm}
\bibliography{posPrism}

\newpage
\appendix
\section*{Appendix}



\section{Attacks on greedy protocols of~\cite{fan2018scalable,fanlarge}
}
\label{sec:attack}

PoS protocols based directly on the Nakamoto protocol, have been proposed recently
($g$-greedy Protocol \cite{fan2018scalable}
and $D$-distance-greedy Protocol \cite{fanlarge}), in
a series of attempts to 
achieve a short predictability 
and a provable security. 
However, the security analyses are subtle and  we show more malignant attacks in this section on the protocols of \cite{fan2018scalable,fanlarge}; thus more 
rigorous analysis is needed to study the security of those protocols.
Both ($g$-greedy Protocols \cite{fan2018scalable} and $D$-distance-greedy Protocols \cite{fanlarge}) are extensions of Nakamoto-PoS discussed in~\ref{sec:intro}, where honest nodes are encouraged to mine on multiple blocks in a single slot as well.

\subsection{$g$-greedy Protocol\cite{fan2018scalable} ($F$-height-greedy Protocol\cite{fanlarge})}
\label{sec:g-greedy}

\noindent{\bf $g$-greedy Protocol}.
In an attempt to increase the amount of adversarial stake that can be tolerated by the basic protocol closer to $1/2$,
 \cite{fan2018scalable} proposes 
a family of protocols called $g$-greedy,
parameterized by
a non-negative integer $g$, which is called as $F$-height-greedy in~\cite{fanlarge}.
Under the $g$-greedy protocol, the honest nodes are prescribed to run multiple leader elections in an attempt to outgrow the private NaS tree.
In an honest node's view, if the longest chain  is at height $\ell$ blocks, then
the node mines on any block that is higher than $\ell-g$ blocks (illustrated by the blocks inside the dotted box in Fig.~\ref{fig:g-greedy-protocol}).
When $g = \infty$, each honest node works on every block,
and when $g=0$, the node works on all blocks at the same height as the tip of the longest chain in the current view.
Note that $0$-greedy differs from the basic Nakamoto's protocol above, which chooses one block to mine on when there are multiple longest chains. However, the $g$-greedy protocol also has prediction window $W=1$. 


\begin{figure}[h]
    \centering
\includegraphics[width=\textwidth]{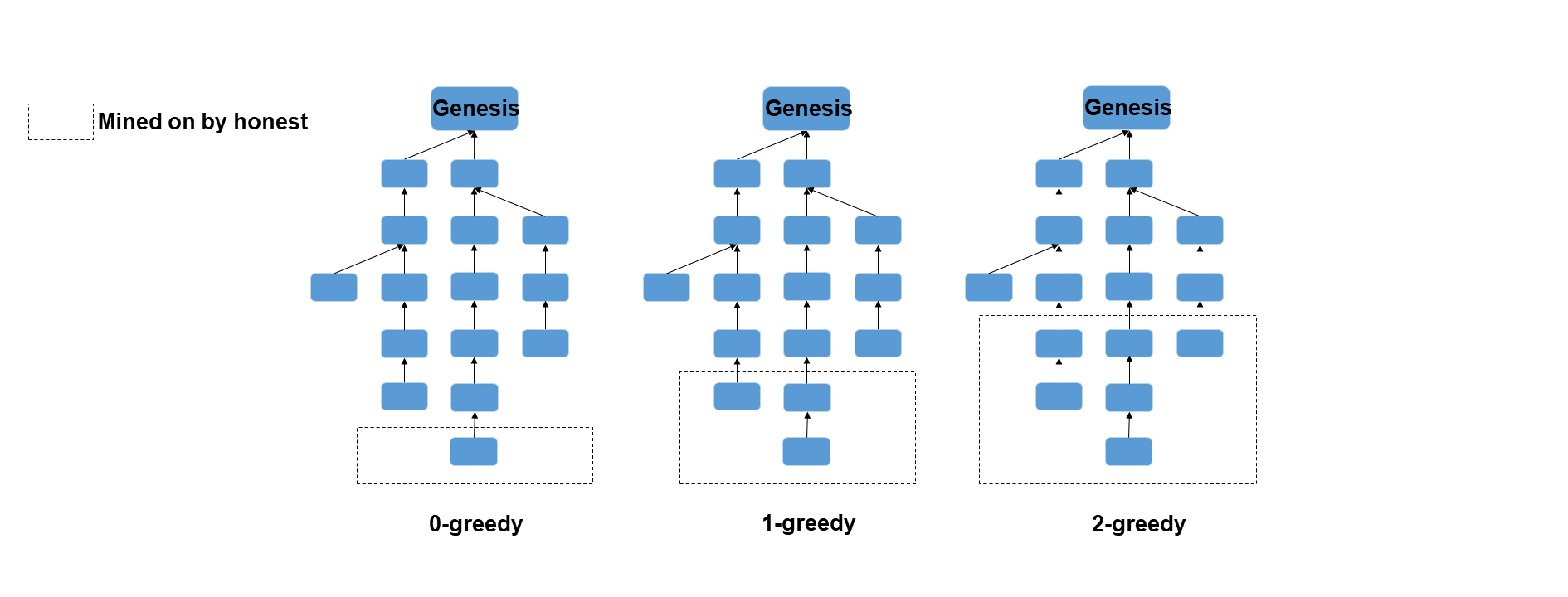}
\caption{$g$-greedy protocol with $g = 0,1,2$}
\label{fig:g-greedy-protocol}
\end{figure}



\noindent {\bf Private Attack on g-greedy Protocol}. A notorious attack in longest chain blockchain protocols is the {\em private attack}:  the adversary privately mines a chain that is longer than the longest chain in the honest view.
Roughly speaking, by the law of large numbers, adversary will fail with the private attack eventually if the {\em growth rate} of the adversarial chain is strictly lower than the honest chain.
The security guarantee for $g$-greedy protocol in \cite{fan2018scalable} is  based on this argument.

\begin{figure}[h]
    \centering
\includegraphics[width=\textwidth]{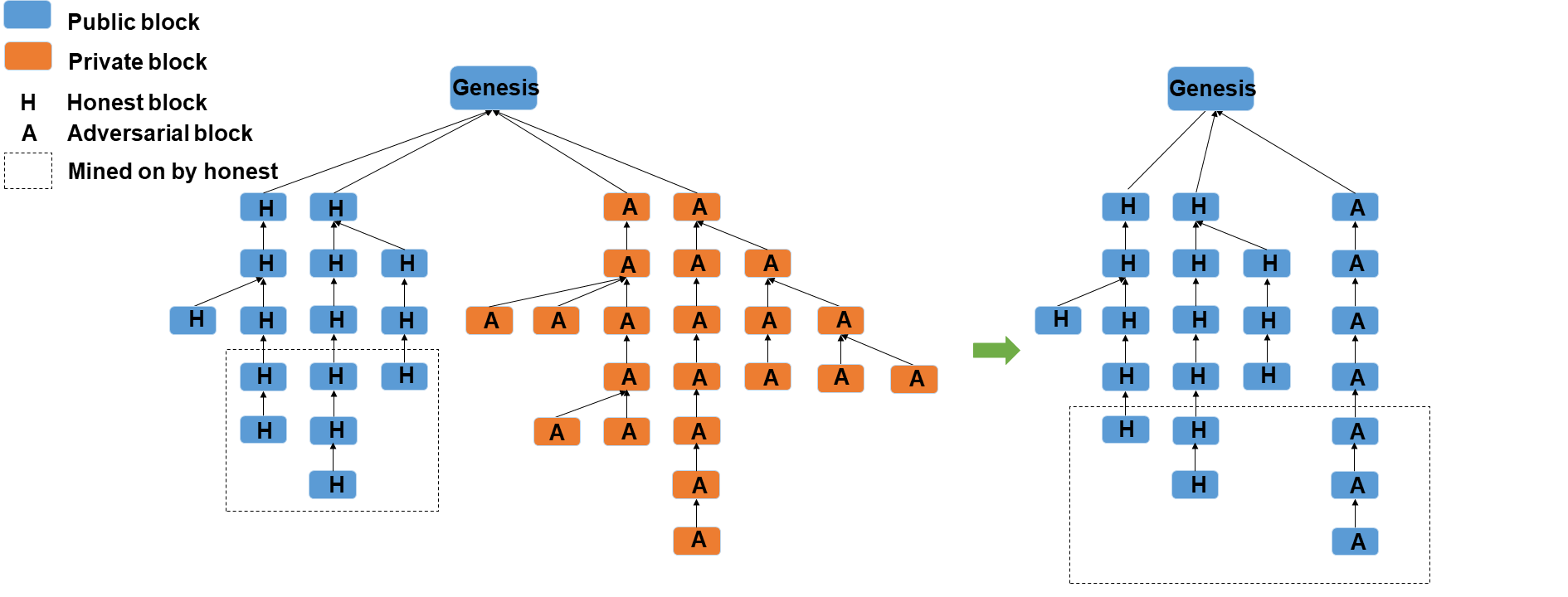}
\caption{Private attack on $g$-greedy protocol with $g = 2$}
\label{fig:g-private-attack}
\end{figure}

\begin{table}
\begin{center}
\begin{tabular}{ |c|c|c|c|c|c|c|c|c|c| }
 \hline
 $g$ & 0 & 1 & 2 & 3 & 4 & 5 & 6 & 7 & 8 \\ \hline
 ${R_g}$ \cite{fan2018scalable} & 1 &  1.7071 & 2.1072 & 2.3428 & 2.4905 & 2.5883 & 2.6562 & 2.7051 & {\color{red} 2.7414}  \\ \hline
 ${R_g}$ this paper & 1 & 1.6531 & 2.0447 & 2.2708 & 2.4084 & 2.4952 & 2.5472 & 2.5805 & 2.6048 \\ \hline
 $\beta_g=\frac{R_g}{e+R_g}$ \cite{fan2018scalable} & 0.2689 & 0.3858 & 0.4367 & 0.4629 & 0.4781 & 0.4876 & 0.4942 & 0.4988 & {\color{red} 0.5021}
 \\  \hline
  $\beta_g=\frac{R_g}{e+R_g}$ this paper & 0.2689 & 0.3782 &  0.4293 & 0.4552 & 0.4698 & 0.4786 & 0.4838 & 0.4870 & 0.4893 \\ \hline
\end{tabular}
\caption{Expected growth rate of honest tree for $g$-greedy protocol under private adversarial behavior. The largest adversarial stake that can be tolerated under the private attack is denoted by $\beta_g$. The claims of \cite{fan2018scalable} are corrected and compared with our calculations. }
\label{tab:growthrates}
\end{center}
\end{table}

How about the expected growth rate $R_g(1-\beta)$ of the honest tree under the $g$-greedy protocol and private behavior by the adversary? It is clear that $R_g$ is monotonically increasing in $g$. The limiting largest expected growth rate  is achieved at $g=\infty$, where the protocol is the same as the NaS attack. Thus the expected growth rate $R_g(1-\beta) = (1-\beta)e$ for $g=\infty$. One suspects that the protocol for $g=0$ is similar to simply growing only one of the longest chains. If the mining rate is slow enough (mining is considered to occur in discrete rounds which are spaced enough apart relative to the network broadcast propagation delay; detailed model in \S\ref{sec:model}), then the expected growth rate of the honest tree (essentially a chain) is simply $1-\beta$.  This is the expected growth rate of the honest tree with $g=0$.

The expected growth rate for a general $g$ is  significantly involved. In Lemma 4.6 of \cite{fan2018scalable}, a heuristic argument is conducted to estimate the expected growth rate for $g=2$ to be $R_2(1-\beta) = 2.1(1-\beta)$. This heuristic argument can be extended to other values of $g$ (see \S\ref{app:fzthreshold}) as summarized in Table~\ref{tab:amplification}. We note that $R_8 = 2.7414$, which is at odds with the observation that $R_g$ increases monotonically and $R_{\infty} = e$. This shows that the heuristic argument in Lemma 4.6 of \cite{fan2018scalable} is flawed. It is clear that the expected growth rate is related to the solution to the following set of recursive differential equations: $\frac{dx_{\ell}(t)}{dt} = x_{\ell-1}(t)$ for $m-g \leq \ell \leq m $ and $\frac{dx_{\ell}(t)}{dt} = 0$ for $0 \leq \ell < m-g$ for some fixed $m$.
Now $R_g$ is the largest value of $\frac{m}{t}$ such that $x_{m}(t)$ decreases exponentially in $t$.
A closed form solution to $R_g$ is challenging; although a numerical solution is readily achieved and tabulated in Table~\ref{tab:growthrates}.

We can make the following conclusion from Table~\ref{tab:growthrates}: the $g$-greedy protocol of \cite{fan2018scalable} is robust to the purely private NaS attack as long as the adversarial stake is such that  $R_g(1-\beta) > R_{\rm NaS}\beta = e\beta$. As $g$ grows large, this threshold on the adversarial stake approaches $\frac{1}{2}$; for instance, for $g = 5$, robustness to private NaS attack is achieved as long as the adversarial stake is less than $47.86\%$.  The caveat is that this security statement is misleading since only  one specific attack (the private NaS attack) has been studied. A different attack strategy could prove more malignant, as we show next.


\noindent {\bf Balance Attack on $g$-greedy Protocol}. We describe an  adversarial action that is a combination of private and public behaviors below;  we term this as a {\em balance attack} for reasons that will become obvious shortly; this attack has some commonalities with the balance attack on the {\texttt{GHOST}} protocol in \cite{kiffer2018better}.
The key idea of the balance attack is to reveal some privately mined blocks at an appropriate time to {\em balance the length of two longest chains} each sharing a common ancestor in the distant past.

The balance attack aims to have two longest chains of equally long length, forking all the way from the genesis block. Fig.~\ref{fig:g-balance-attack} depicts the balance attack in action, as it progresses over time.
At any time, the adversary will try to mine on every block including public blocks and private blocks, while honest nodes follow $g$-greedy protocol.
Once the adversary succeeds in creating a new block, it will first keep the block in private as shown in step 1 and step 3 of Fig.~\ref{fig:g-balance-attack}.
Whenever the adversary owns a private chain that has the same length as the longest public chain, it will reveal the private chain so that there will be two public chain with the same length as shown in step 2 and step 4 of Fig. \ref{fig:g-balance-attack}.
If the adversary can keep up this kind of balance attack, then it has successfully prevented any block from irreversible no matter how deep the block is buried in the longest chain.
To succeed with the balance attack, the adversary only needs to mine a few blocks as most blocks in the two main chain may be mined by honest nodes under $g$-greedy protocol.
A formal pseudocode describing the balance attack algorithm is available in \S\ref{app:balancepseudocode}.

\begin{figure}[h]
    \centering
\includegraphics[width=\textwidth]{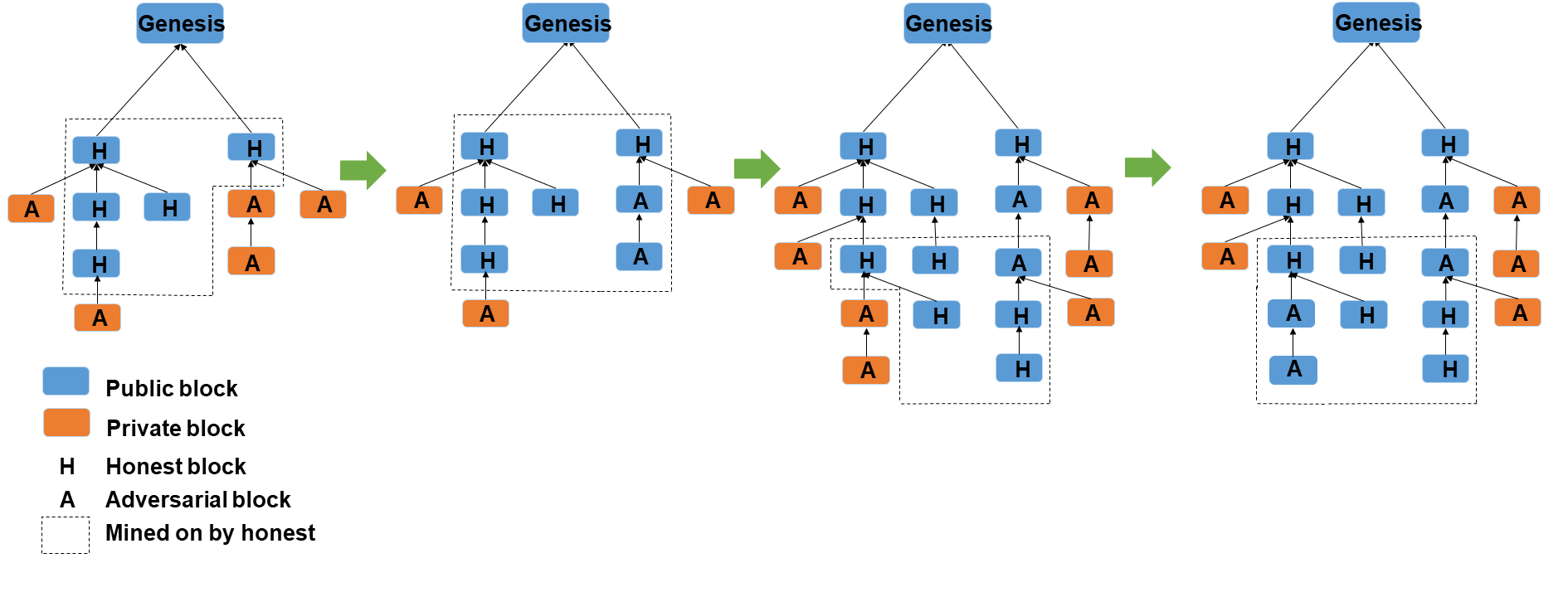}
\caption{Snapshot of balance attack in progress on $g$-greedy protocol with $g = 2$. The goal of the adversary is to balance the lengths of two  chains each tracing the genesis block as their ancestor, consequently creating a deep fork.  }
\label{fig:g-balance-attack}
\end{figure}

The balance attack is hard to analyze theoretically, but its efficacy can be evaluated via simulations.
The simulation starts with a fork with length $1$ (i.e., two independent public trees).
The honest nodes follow $g$-greedy protocol with stake proportion $1-\beta$ while adversary will always mine on every block in his private view.
Whenever the two public trees have different heights, the adversary will immediately reveal the private blocks appending to the shorter public tree to re-balance their heights if he can.
We set $f\Delta = 0.1$ in each round and simulated various pairs of $(g,\beta)$ for 2000 rounds.
Note that the simulation can be made more accurate with a finer discretization and a longer run-time.
The depth of the longest fork and its associated cumulative probability is plotted in Fig.~\ref{fig:balance_attack} for fixed parameters $g,\beta$. We see that in each instance, the probability of a fork of any length is larger for the balance attack than the private NaS attack; this implies that the balance attack {\em stochastically dominates} the private attack.
For example, when $g=2,\beta=0.38$, the adversary is able to cause a fork from genesis, which is longer than 100 blocks, with a probability around 20\%, while the private attack can only achieve it with probability less than 0.1\%. In this instance, the balance attack is significantly more powerful than the private attack.

\begin{figure}[h]
    \centering
\includegraphics[width=1\textwidth]{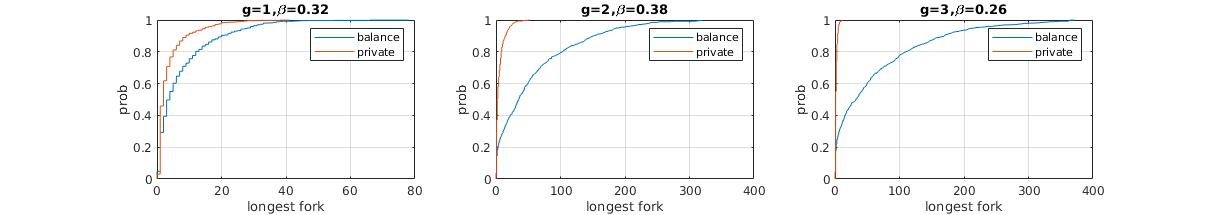}
    \caption{
    Comparison of CDF of the longest fork from genesis caused by private attack and balance attack for various pairs of $(g,\beta)$. The balance attack is  a stochastically dominating strategy for $g \geq 1$.
    }
    \label{fig:balance_attack}
\end{figure}

We see that virulence of the balance attack increases, dramatically, with $g$. In Fig.~\ref{fig:threshold}, the blue lines plot the largest value of $\beta$ the balance attacking adversary can have while allowing a fork of a large fixed length (50,100,200 blocks) with significantly high probability (25\%, 25\%, 10\% respectively). The orange line plots the corresponding largest value of $\beta$ using the private NaS attack. The experiments conclusively demonstrate the fatal nature of the balance attack as $g$ increases: the private NaS attack gets weaker while the balance attack gets dramatically stronger. For $g=6$, the balance attack is successful in creating a very long fork (200 blocks deep) with a high probability (10\%) using only 3\% of the stake.

\begin{figure}[h]
    \centering
\includegraphics[width=\textwidth]{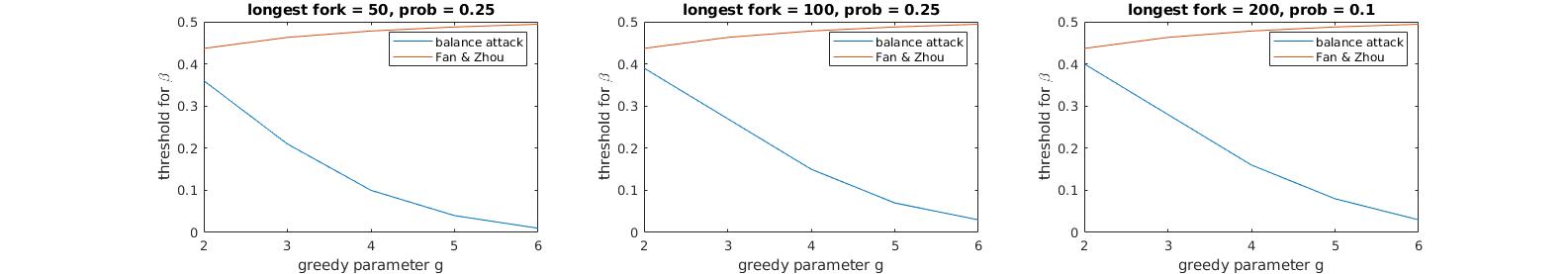}
    \caption{
    The threshold of $\beta$ to break $k$-common-prefix property with a certain probability actually declines as greedy parameter $g$ increases.
    }
    \label{fig:threshold}
\end{figure}

\subsection{D-distance-greedy Protocol\cite{fanlarge}}
\label{sec:D-greedy}

\noindent {\bf D-distance-greedy Protocol}.
In an attempt to make
the  $g$-greedy protocol of \cite{fan2018scalable}
robust against the balance attack,
recent work~\cite{fanlarge} has updated the protocol
to a new variant called the  $D$-distance-greedy protocol  described as below.
The reason $g$-greedy protocol is so vulnerable against
balance attack is that honest nodes are mining on both sides of two, for example, chains that forked many blocks deep.
This makes it easy for the adversary to keep those chains balanced, especially when $g$ is large. The protocol in~\cite{fanlarge} is aimed to prevent honest nodes from mining on both chains  as follows.

We first define the distance between two chains in a tree. For two chains $\mathcal{C}_a$ with length $\ell_a$ and $\mathcal{C}_b$ with length $\ell_b$, let $\mathcal{C}$ be the common prefix of $\mathcal{C}_a$ and $\mathcal{C}_b$ with length $\ell$, then the distance between $\mathcal{C}_a$ and $\mathcal{C}_b$ is defined as $\max(\ell_a-\ell,\ell_b-\ell)$.
$D$-distance-greedy protocol prescribes every honest node to first select one of the longest chains (the tie breaking could be random)  and  attempt to extend a set of chains in which all chains have distance no more than $D$ from the longest chain (illustrated by the blocks with green diamonds in Fig.~\ref{fig:D-greedy-protocol}).
If there are two chains that forked at more than $D$ blocks deep
from the tip of the chain, then each
honest node will mine on one side of those two chains. Again, the prediction window remains the same: $W=1$.

\begin{figure}[h]
    \centering
\includegraphics[width=\textwidth]{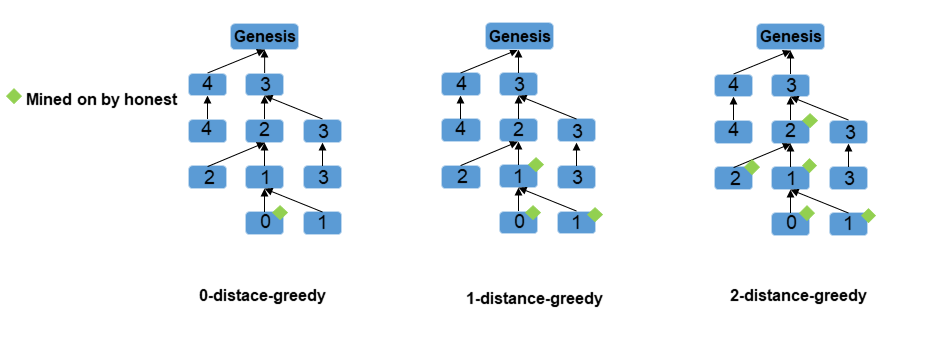}
\caption{$D$-distance-greedy protocol with $D = 0,1,2$}
\label{fig:D-greedy-protocol}
\end{figure}

The intuition behind this protocol  is the following:
\begin{enumerate}
    \item Although we are limiting the mining strategy of honest nodes to
    small distance blocks,
    the growth of the longest chain is still high. Using a heuristic argument, \cite{fanlarge} claims that the expected growth rate $A_D(1-\beta)$ of the longest chain  following the $D$-distance greedy protocol, is  $\frac{(1-\beta) (D+1)}{{\sqrt[D+1]{(D+1)!}}}$. Note that as $D$ increases, the expected growth rate of the honest tree,  according to this calculation, approaches $e(1-\beta)$. Based on this calculation,  the anticipated threshold of largest adversarial stake that can be tolerated  under a private attack, $\beta_D =
    \frac{A_D}{A_D + e}$, approaches the ideal threshold of $\beta=0.5$.

    \item Second, limiting the honest strategy to only grow on nodes that are near the tip of the longest chain might, at the outset, seems to render the balance attack of \S\ref{sec:g-greedy} ineffective. Indeed, an argument is made in \cite{fanlarge} to imply the security of the $D$-distance-greedy protocol against all possible adversarial strategies as long as the adversarial stake is less than $\beta_D$.

\end{enumerate}

In this section, we  show that both these results of \cite{fanlarge} are inaccurate.
We  do this by focusing on the case of $D=1$ first. Specifically:
\begin{enumerate}

    \item For $D=1$, we show that the precise expected growth of the longest chain of the honest block tree is $A_1(1-\beta) = \frac{1-\beta}{e-2}\approx 1.3922(1-\beta)$ and not $\sqrt{2}(1-\beta)$.
    This implies that, even if the adversary is controlling a
    smaller fraction of stake, say $\beta= 0.339$,
    than the threshold $\beta_{D=1} = 0.3422$ claimed in \cite{fanlarge}, the simple private attack can succeed with probability close to one.
    A formal statement and proof is in  \S\ref{app:1distance}; a precise calculation of the growth rate for general $D$ appears to be a challenging mathematical problem.

    \item More critically, we show that a  balance attack similar to the one from \S\ref{sec:g-greedy} can still be effective on the $D$-distance-greedy protocol -- and the efficacy is primarily achieved by {\em slowing down} the growth rate of the honest strategy.
    The D-distance-greedy protocol is not secure against balance attack, even if
    the adversary controls less than
     $\beta_D =  \frac{A_D}{A_D + e}$ fraction of the stake, with the correct $A_D$ we compute in \S\ref{app:1distance}.
     The actual security threshold on $\beta$ is strictly smaller;
      the private attack   considered in \cite{fanlarge} is
     {\em strictly} weaker than a standard adversary studied in the literature which can perform the balance attack.
    This is described in detail next.

\end{enumerate}

For general $D$ larger than one,
a balance attack can successfully slow down the growth rate of honest nodes,
but a precise computation of the slowed down growth rate is a challenging problem. We discuss security implications for large $D$ at the end of this section.
\cite{fanlarge} expects that the security of the
protocol increases as  $D$ increases, eventually achieving the ideal threshold of $\beta=0.5$.
Contrarily, we find that for a large enough $D$,
    the $D$-distance-greedy protocol becomes insecure;
    none of the blocks can be confirmed, regardless of how deep the block is in the current block chain.
    This is also true for the $g$-greedy protocol with a large enough $g$.
    This follows from the fact that when every node is mining on every block,
    the entire block tree becomes {\em unstable}: the prefix of the longest chain keeps changing indefinitely.
    Hence, even if the conjectured growth rate of
    $A_D=\frac{(D+1)}{\sqrt[D+1]{(D+1)!}}$ were true for large $D$,
    the desired threshold of $\beta=0.5$ cannot be achieved.

\noindent {\bf Balance Attack on $D$-distance-greedy Protocol}.
We consider the same balance attack, now on $1$-distance-greedy protocol.
Here we assume that, when two chains with equal length are broadcast in the network, honest nodes will split into two groups with equal stake proportion, where each group picks one chain as the longest chain and mines on the corresponding tree under $1$-distance-greedy protocol. This balance attack is described formally in \S\ref{app:Dbalancepseudocode}.

Perhaps surprisingly, this balance attack can
{\em slow down} the growth rate of the honest block tree:
 we show the growth rate of the longest chain  to be  $\tilde{A}_1(1-\beta) = \frac{4(1-\beta)}{3e^2-19} \approx 1.26 (1-\beta) < A_1(1-\beta)$.
A formal statement and proof of this result is in \S\ref{app:balancedrate}.
We note that this result is against the grain of the Nakamoto longest chain protocol, where no adversary strategy can slow down the growth rate of the block tree, even though  the honest nodes can be potentially split into working on two different chains of equally long length.  Thus when the proportion of adversarial stake $\beta$ falls in the interval $(\frac{\tilde{A}_1}{\tilde{A}_1+e},\frac{A_1}{A_1+e})$, 
the adversarial private tree will be ahead of the public longest chain, and the balance attack in conjunction with the private attack is fatal (i.e.,  any confirmed block can be reversed no matter how deep the block is buried in the longest chain).

We illustrate the high level idea of the slow-down effect of $1$-distance-greedy protocol with Fig.~\ref{fig:D_slow_down}.
The scenario at the top of Fig.~\ref{fig:D_slow_down} is a pure honest tree with one block at height $\ell-1$ and two blocks at height $\ell$, and honest nodes will mine on all these three blocks with rate $1-\beta$; while the scenario at the bottom is two trees balanced by the adversary, where child-blocks of every block are generated as a Poisson point process with rate $0.5(1-\beta)$ since honest nodes are split.
Note that the number of blocks at height $\ell$ and $\ell+1$ increases with the same rates in these two scenarios ($1-\beta)$ and $2(1-\beta)$ respectively).
However, when the number of blocks at height $\ell$ increases by one in both scenarios, the growth rate of blocks at height $\ell+1$ becomes $3(1-\beta)$ in the top scenario but only $2.5(1-\beta)$ in the bottom.
Since only half of honest nodes will benefit from, or mine on the new block due to splitting, the growth rate of the public longest chain is slowed down.
Formal analysis of this slow-down effect will be provided in \S\ref{app:balancedrate}.

\begin{figure}[h]
    \centering
\includegraphics[width=\textwidth]{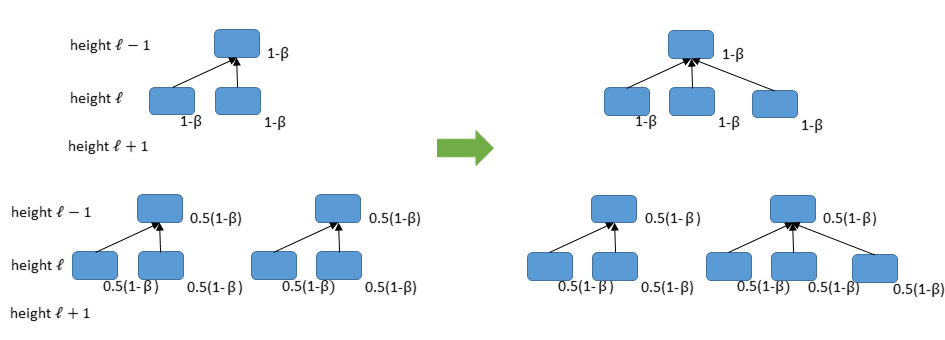}
\caption{The slow-down effect of $1$-distance-greedy protocol}
\label{fig:D_slow_down}
\end{figure}

We also use simulation to verify the analysis above. The simulation starts with a fork with length $2$ (i.e., two independent public trees).
The honest nodes follow $1$-distance-greedy protocol with stake proportion $1-\beta$ while adversary will always mine on every block in his private view.
Whenever the two public trees have different heights, the adversary will immediately reveal a private chain appending to the shorter public tree to re-balance their heights if he can.
For simplicity, when two chains with equal length are revealed in the public view, honest nodes will
split into two groups with equal stake proportion.
We set $f\Delta = 0.1$ in each round and simulated the case $\beta = 0.32$ and $\beta = 0.33$ for 10000 rounds and 20000 rounds respectively.
Note that the simulation can be made more accurate with a finer discretization and a longer run-time.
The depth of the longest fork and its associated cumulative probability is plotted in Fig.~\ref{fig:D_balance_attack} for fixed parameters $\beta = 0.32$ and $\beta = 0.33$.
One can see that the success probability (1 minus the probability in the plots) of private attack will converge to $0$ as for longer fork, while the success probability of balance attack will saturate to some constant, which indicates that, using the balance attack, the adversary can succeed in creating a fork of any length with some non-negligible probability.
This simulation results immediately verify the slow-down effect of $1$-distance-greedy.
When the growth rate of the public longest chain is slowed down due to splitting for some time, the adversarial private tree will be ahead of the public longest chain, and the balance attack can continue to succeed easily in the future and create forks with any length.

\begin{figure}[h]
    \centering
\includegraphics[width=1\textwidth]{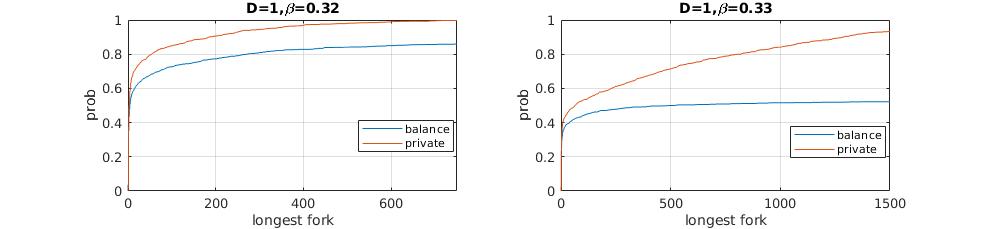}
    \caption{
    Comparison of CDF of the longest fork from genesis caused by private attack and balance attack for various $\beta$. Due to the slow-down effect, the success probability of balance attack saturates to some constant number.
    }
    \label{fig:D_balance_attack}
\end{figure}

For general $D$-distance-greedy,
the growth rate of the honest chain can also be slowed down from $A_D(1-\beta)$ by splitting the honest nodes.
However, this balance attack of keeping {\em two} chains of balanced lengths is
 a sub-optimal strategy (for the adversary).
 The adversary can maintain
 multiple balanced chains,
 splitting honest nodes as much as possible.
Under such an attack, we conjecture that the growth rate of the public longest chain
can be slowed down all the way to $1-\beta$, eliminating amplification entirely.
However, how large $\beta$ needs to be to make this attack successful
remains an open question, and we defer it as a future research direction.  

In order to fix the slow-down effect for $1$-distance greedy protocol,
we  introduce the following  refined tie breaking rule.
Recall that $1$-distance greedy protocol prescribes
an honest node to mine on one of the longest chains (chosen randomly when there are multiple),
and also simultaneously mine on its parent block and all its sibling blocks.

\begin{definition}[no-slow-down tie breaking rule for $1$-distance greedy protocol]
 When there are multiple longest chains with equal lengths,
mine on one node that has  the largest number of sibling blocks,
and simultaneously mine on its parent block and all its sibling blocks.
Further ties can be broken arbitrarily.
\end{definition}

For example, let $h$ denote the height of the longest chain.
If there are two nodes at height $h-1$, one with two children blocks and another with three children blocks,
then honest nodes will mine on the three sibling blocks at height $h$ and its parent.
If there are three nodes at height $h-1$, one with two children blocks and the rest with three children blocks each,
then each honest node will mine on  one set of three sibling blocks and its parent (chosen arbitrarily).
If this finer tie breaking rule is adopted, then we can prove the following ‘‘no-slow-down lemma'' for $1$-distance greedy protocol.

\begin{lemma}[no-slow-down lemma for $1$-distance greedy]
Under the synchronous network setting with bounded message delay from~\cite{fanlarge}, an adversary cannot slow down the honest chain growth. That is, no matter what the adversary does, the honest nodes will grow a chain that is at least as long as if the adversary does not generate any blocks.
\end{lemma}

\begin{proof}
By the analysis in \S\ref{app:1distance}, for $1$-distance greedy protocol, the growth rate of the longest chain in the honest tree depends only on the number of sibling blocks on the tip, then the adversary cannot slow down the honest chain under this tie breaking rule since honest nodes are always choosing the longest chain with the largest potential growth rate.
To be more specific, whenever the adversary publishes a block $B$ and some honest nodes start to mine on $B$ under this no slow-down tie breaking rule, which means that either $B$ is on a deeper level or the number of sibling blocks of $B$ is no less than the number of sibling blocks of the main chain which is chosen before $B$ is revealed, in both case the growth of the honest tree will not be slowed down.
\end{proof}

Thus the chain growth property for $1$-distance greedy protocol can be guaranteed under the tie breaking rule with the no slow-down lemma proved above.
This is not true for the random tie breaking rule, as we shows with balance attack in \S\ref{app:balancedrate}.
Beyond the no-slow-down lemma, a rigorous analysis  of common prefix property for $1$-distance greedy protocol still remains a challenging research problem.

When there are multiple longest chains to be mined on,
no-slow-down lemma holds as long as the honest nodes can choose the one that will ensure fastest expected growth.
However, it is not straightforward how to selecting the fastest growing chain under  $D$-distance-greedy, especially for large $D$.
For large $D$, it becomes computationally intractable to break ties in order to eliminate the slow-down effect completely.

\section{Notes and Proofs for \S\ref{sec:attack}}

\subsection{$R_g$ calculation according to \cite{fan2018scalable}}
\label{app:fzthreshold}


We recall that $R_g(1-\beta)$ is the expected growth rate of the longest chain of the tree grown by honest nodes (adversaries are acting purely privately) using the $g$-greedy protocol. We also recall that $R_g$ is increasing monotonically in $g$ and that $R_\infty = e$, where $e$ is the base of the natural logarithm.
In \cite{fan2018scalable}, the authors claim that $R_2 = 2.1$ by assuming that {\em all} chains grow with at the {\em same}  steady rate as the longest chain and then investigating the average case.
One can generalize their methodology for arbitrary $g$ as follows: 
suppose  the length of the longest chain at round $r$ is $\ell$, let $x_i$ with $0 \leq i \leq g$ be the number of chains with length $\ell - i$.
For simplicity, $x_0$ is set to be $1$.
Then using the  arguments in Lemma 4.6 in \cite{fan2018scalable}, we  obtain the following $g+1$ equations with $g+1$ unknowns $\{x_i:1 \leq i \leq g \}$ and $R_g$.
\begin{equation}
\begin{cases}
     \frac{x_0+x_1}{2} = R_g x_0  \\
     \frac{x_i+x_{i+1}}{2} = R_g (x_i - x_{i-1})  \qquad\text{for }i\in\{1,2,\cdots, g-1\} \\
     x_g = R_g(x_g - x_{g-1}) \nonumber 
\end{cases}
\end{equation}

By solving these equations numerically, we obtain the value of $R_g$ as shown in Table~\ref{tab:amplification}.
It turns out that $R_g > e$ for $g \geq 8$ which is in contradiction to the fact that $R_g$ is monotonically increasing and the limiting value is $R_\infty = e$. This shows that the method  in \cite{fan2018scalable} to estimate the growth rate of the longest chain  is flawed. 

\begin{table}
\begin{center}
\begin{tabular}{ |c|c|c|c|c|c|c|c|c|c|c|c|c| } 
 \hline
 $g$ & 0 & 1 & 2 & 3 & 4 & 5 & 6 & 7 & 8 & 9 & 10 & $\infty$ \\ 
 \hline
 $R_g$ & 1 & 1.7071 & 2.1072 & 2.3428 & 2.4905 & 2.5883 & 2.6562 & 2.7051 & {\color{red} 2.7414} & {\color{red} 2.7690} & {\color{red} 2.7906} & $e=2.7183$ \\
 \hline
\end{tabular}
\caption{Chain growth rate for $g$-greedy protocol using the method of \cite{fan2018scalable}. There is an inconsistency in the calculation for $g \geq 8$.}
\label{tab:amplification}
\end{center}
\end{table}


\subsection{Pseudo Code for the Balance Attack on g-greedy Protocol}
\label{app:balancepseudocode}

In the main text we have described the balance attack informally along with accompanying examples depicted in Fig.~\ref{fig:g-balance-attack}. Here we provide a formal pseudocode describing the balance attack algorithm for completeness. 

\begin{algorithm}
{\fontsize{10pt}{10pt}\selectfont
\caption*{Balance attack $(g)$}
\begin{algorithmic}[1]
\Procedure{Initialize}{ } 
\State ${\rm PrivateTree1} \gets {\rm genesis}$ 
\State ${\rm PrivateTree2} \gets {\rm genesis}$ 
\State ${\rm PublicTree1} \gets {\rm genesis}$ 
\State ${\rm PublicTree2} \gets {\rm genesis}$ 
\EndProcedure
\Procedure{Attack}{ } 
\For{$r = 1:r_{\rm max}$}
\State HonestMining($g$,PublicTree1,PublicTree2)              \colorcomment{honest nodes work on blocks in the public view under $g$-greedy protocol}
\State AdversaryMining(PrivateTree1,PrivateTree2)           \colorcomment{adversary works privately on all blocks in the private view}
\If{height(PublicTree1) $==$ height(PublicTree2)}
		\State adversary reveals nothing 
    	\ElsIf{height(PublicTree1) $>$ height(PublicTree2)} 
		    \State PublicTree2(1:height(PublicTree1)) $\gets$ PrivateTree2(1:height(PublicTree1)) 
        	\ElsIf{height(PublicTree1) $<$ height(PublicTree2)}
		        \State PublicTree1(1:height(PublicTree2)) $\gets$ PrivateTree1(1:height(PublicTree2)) 
\EndIf
\State height($r$) $\gets$ $\max$\{height(PublicTree1),height(PublicTree2)\}
\State diff($r$) $\gets$ height(PublicTree1) $-$ height(PublicTree2)
\EndFor
\State \Return height(last(diff $== 0$))  \colorcomment{return the length of the  longest fork from genesis}

\EndProcedure
\vspace{1mm}

\end{algorithmic}
}
\end{algorithm}

\subsection{Growth rate of $1$-distance-greedy Protocol}
\label{app:1distance}
In this section we calculate the expected growth rate of honest tree under $1$-distance-greedy protocol with a continuous time Markov chain. 

The honest nodes grow a tree starting with the genesis block as root at time $t=0$.
At any time $t$, the honest nodes independently mines on a set of blocks according to $1$-distance-greedy protocol.
From the point of view of a fixed block, its child-blocks are generated as a Poisson point process with
rate $1-\beta$ i.e.,  the time interval between the child-blocks is an exponential random variable with rate $1-\beta$. 

Let $W(t)$ be the number of blocks on the tip of the tree. For very small values of $h$, we have
\begin{align*}
&\mathbb{P}\big(W(t+h)=k+1 \mid W(t)=k\big)  =(1-\beta)h + o(h)\\
&\mathbb{P}\big(W(t+h)=1 \mid W(t)=k\big)  =k(1-\beta)h + o(h)
\end{align*}

It follows that the infinitesimal generator of $W(t)$ is
\begin{equation*}
\mathrm{A}=(1-\beta)\left(\begin{array}{ccccc}{-1} & {1} & {0}  & {0} & {\cdots} \\ 
2 & -3 & 1  & {~0~} & {\cdots} \\ 
~~3~~ & ~~0~~ & {~~-4~~} & ~~1~~  & {\cdots} \\ 
{\vdots} & {\vdots} & {\vdots}  & {\vdots} & {\ddots}\end{array}\right)
\end{equation*}
let $\pi = (\pi_1,\pi_2,\cdots)$ be the stationary distribution, then it satisfies $\pi \mathrm{A}=0$. By solving the equations, we have $\pi_1 = \frac{1}{2(e-2)}$ and $\pi_n = \frac{2\pi_1}{(n+1)!}$ for $n\geq1$. Then the expected number of blocks on the tip of the honest tree can be computed as
\begin{equation*}
    \mathbb{E}[W(t)] = \sum_{n=1}^{\infty} n \pi_n = 2\pi_1 \sum_{n=1}^{\infty}\frac{n}{(n+1)!} = 
    2\pi_1 \sum_{n=1}^{\infty}\big(\frac{1}{n!}-\frac{1}{(n+1)!}\big) = 2\pi_1 = \frac{1}{e-2}.
\end{equation*}

Let $T$ be the time for the height of the honest tree to grow by one, then we have
\begin{align}
    &\mathbb{P}(T \in (t,t+dt) \mid W(t)=k, T \geq t ) = k(1-\beta)dt  \nonumber\\
\Longrightarrow &\mathbb{P}(T \in (t,t+dt) \mid T \geq t) = \sum_{k=1}^{\infty} k(1-\beta)dt \mathbb{P}(W(t)=k) = \mathbb{E}[W(t)](1-\beta)dt  \nonumber\\
\Longrightarrow &\frac{\mathbb{P}(T \in (t,t+dt))}{\mathbb{P}(T \geq t)} = \mathbb{E}[W(t)](1-\beta)dt.
\label{eqn:pdf_cdf_eq}
\end{align}

Let $f_T(t)$ be the probability density function of random variable $T$, $F_T(t)$ be the cumulative distribution function of $T$, and $F_T^c(t) = 1-F_T(t)$, then we know $f_T(t) =-\dot{F_T^c(t)}$. Equation \eqref{eqn:pdf_cdf_eq} can be represented as
\begin{equation*}
    \frac{\dot{F_T^c(t)}}{F_T^c(t)} = -\mathbb{E}[W(t)](1-\beta)dt,
\end{equation*}
with $F_T^c(0) = 1$, and the solution of this differential equation is $F_T^c(t) = e^{\mathbb{E}[W(t)](1-\beta)}$, then we have
\begin{equation*}
    \mathbb{E}[T] = \int_0^\infty F_T^c(t)dt = \frac{1}{\mathbb{E}[W(t)](1-\beta)}.
\end{equation*}
Therefore we can conclude that the growth rate of the honest tree $A_1(1-\beta) = 1/\mathbb{E}[T] = \mathbb{E}[W(t)](1-\beta) = \frac{1-\beta}{e-2} \approx 1.39(1-\beta)$.

\subsection{Growth rate of $1$-distance-greedy Protocol under Balance Attack}
\label{app:balancedrate}

In this section we calculate the expected growth rate of honest tree in $1$-distance-greedy protocol under the balance attack in \S\ref{sec:D-greedy} (see pseudocode in \S\ref{app:Dbalancepseudocode}). 

We start with two chains of equal length and the distance between them is greater than $1$. In this situation,  honest nodes will split into two groups with equal stake proportion to mine on different chains under $1$-distance-greedy protocol.
We first consider the {\em steady state} behavior where   the adversary has ``sufficient" number of private chains under each public chain, that is, whenever the length of one public chain grows by one, the adversary can immediately reveal one block for the other chain to re-balance these two chains and keep the honest nodes split.
We analyze the {\em transient} stage, i.e., getting to steady state from genesis after the steady state analysis.

Let $U(t)$ and $V(t)$ be the number of blocks on the tip of the two trees. For very small values of $h$, we have
\begin{align*}
&\mathbb{P}\big((U(t+h),V(t+h))=(u+1,v) \mid (U(t),V(t))=(u,v)\big)  = 0.5(1-\beta)h + o(h)\\
&\mathbb{P}\big((U(t+h),V(t+h))=(u,v+1) \mid (U(t),V(t))=(u,v)\big)  = 0.5(1-\beta)h + o(h)\\
&\mathbb{P}\big((U(t+h),V(t+h))=(1,1) \mid (U(t),V(t))=(u,v)\big)  = 0.5(u+v)(1-\beta)h + o(h)\\
\end{align*}

Let $Z(t) = U(t)+V(t)$, then we have
\begin{align*}
&\mathbb{P}\big(Z(t+h)=k+1 \mid Z(t)=k\big)  =(1-\beta)h + o(h)\\
&\mathbb{P}\big(Z(t+h)=2 \mid Z(t)=k\big)  =0.5k(1-\beta)h + o(h)
\end{align*}

It follows that the infinitesimal generator of $Z(t)$ is
\begin{equation*}
\mathrm{A}=(1-\beta)\left(\begin{array}{ccccc}{-1} & {1} & {0}  & {0} & {\cdots} \\ 
1.5 & -2.5 & 1  & {~0~} & {\cdots} \\ 
~~2~~ & ~~0~~ & {~~-3~~} & ~~1~~  & {\cdots} \\ 
{\vdots} & {\vdots} & {\vdots}  & {\vdots} & {\ddots}\end{array}\right)
\end{equation*}
let $\pi = (\pi_2,\pi_3,\cdots)$ be the stationary distribution, then it satisfies $\pi \mathrm{A}=0$. By solving the equations, we have $\pi_2 = \frac{2}{3e^2-19}$ and $\pi_n = \frac{3\pi_1 2^{n+1}}{(n+2)!}$ for $n\geq1$. Then the expected number of blocks on the tip of the honest tree can be computed as
\begin{align*}
    \mathbb{E}[Z(t)] &= \sum_{n=2}^{\infty} n \pi_n = 3\pi_2 \sum_{n=2}^{\infty}\frac{2^{n+1}n}{(n+2)!} = 
    3\pi_2 \sum_{n=2}^{\infty}\frac{2^{n+1}(n+2-2)}{(n+2)!}  \\
    &= 3\pi_2 \sum_{n=2}^{\infty}\big(\frac{2^{n+1}}{(n+1)!}-\frac{2^{n+2}}{(n+2)!}\big) = 4\pi_2.
\end{align*}

Let $T$ be the time for the height of the honest tree to grow by one.  Then, similar to the analysis in \S\ref{app:1distance}, we have that $\mathbb{E}[T] = 2/(\mathbb{E}[Z(t)](1-\beta))$
Thus the growth rate of the honest tree 
$\tilde{A}_1 (1-\beta)= 1/\mathbb{E}[T] = \mathbb{E}[Z(t)](1-\beta)/2 = \frac{4}{e^2-19}(1-\beta) \approx 1.26(1-\beta) < A_1(1-\beta)$.

This implies that the growth rate of the public longest chain will be slowed down due to the balance attack. Therefore, when the adversary with $\beta$ smaller but close to $A_1/(A_1+e)$ succeeds in balancing the two chains for some time, the private chains will grow faster than the public chain, which gives the adversary a lot of future blocks in the ``bank'' to continue the balance attack.

\subsection{Pseudo Code for the Balance Attack on $1$-distance-greedy Protocol}
\label{app:Dbalancepseudocode}

Here we provide a formal pseudocode describing the balance attack algorithm on $1$-distance-greedy protocol for completeness. 

\begin{algorithm}
{\fontsize{10pt}{10pt}\selectfont
\caption*{Balance attack $(D=1, \beta)$}
\begin{algorithmic}[1]
\Procedure{Initialize}{ } 
\State ${\rm PrivateTree1} \gets {\rm genesis}$ 
\State ${\rm PrivateTree2} \gets {\rm genesis}$ 
\State ${\rm PublicTree1} \gets {\rm genesis}$ 
\State ${\rm PublicTree2} \gets {\rm genesis}$ 
\EndProcedure
\Procedure{Attack}{ } 
\For{$r = 1:r_{\rm max}$}
\If{height(PublicTree1) $==$ height(PublicTree2)}
    \State HonestMining(Poisson\_rate = $0.5(1-\beta)$,$D=1$,PublicTree1) 
    \State HonestMining(Poisson\_rate = $0.5(1-\beta)$,$D=1$,PublicTree2)
    \ElsIf{height(PublicTree1) $>$ height(PublicTree2)}
        \State HonestMining(Poisson\_rate = $1-\beta$, $D=1$,PublicTree1) 
        \ElsIf{height(PublicTree1) $<$ height(PublicTree2)}
            \State HonestMining(Poisson\_rate = $1-\beta$, $D=1$,PublicTree2)
\EndIf
\colorcomment{honest nodes work on blocks in the public view under $D$-distance-greedy protocol}
\State AdversaryMining(Poisson\_rate = $\beta$,PrivateTree1,PrivateTree2)           \colorcomment{adversary works privately on all blocks in the private view}
\If{height(PublicTree1) $==$ height(PublicTree2)}
		\State adversary reveals nothing 
    	\ElsIf{height(PublicTree1) $>$ height(PublicTree2)} 
		    \State PublicTree2(height(PublicTree2)+1:min(height(PrivateTree2),height(PublicTree1)) $\gets$ 1
		    \colorcomment{adversary reveals a chain from PrivateTree2 to match the height of the public trees if he can}
        	\ElsIf{height(PublicTree1) $<$ height(PublicTree2)}
		        \State PublicTree1(height(PublicTree1)+1:min(height(PrivateTree1),height(PublicTree2)) $\gets$ 1
		    \colorcomment{adversary reveals a chain from PrivateTree1 to match the height of the public trees if he can} 
\EndIf
\State height($r$) $\gets$ $\max$\{height(PublicTree1),height(PublicTree2)\}
\State diff($r$) $\gets$ height(PublicTree1) $-$ height(PublicTree2)
\EndFor
\State \Return height(last(diff $== 0$))  \colorcomment{return the length of the  longest fork from genesis}

\EndProcedure
\vspace{1mm}

\end{algorithmic}
}
\end{algorithm}

F

\section{Verifiable Random Functions}
\label{sec:vrf}

\begin{definition}[from \cite{vrf2}]
    \label{def:vrf}
    A function family $F_{(\cdot)}(\cdot): \{0,1\}^{a(\kappa)} \to \{0,1\}^{b(\kappa)}$ is a family of VRFs is there exists a
    probabilistic polynomial-time algorithm {\sc Gen} and deterministic algorithms {\sc VRFprove} and {\sc VRFverify} such that {\sc Gen}$(1^\kappa)$ outputs a pair of keys $(pk,sk)$; {\sc VRFprove}$(x,sk)$ computes $(F_{sk}(x),\pi_{sk}(x))$, where $\pi_{sk}(x)$ is the proof of correctness; and
    {\sc VRFverify}$(x,y,\pi,pk)$ verifies that $y=F_{sk}(x)$ using the proof $\pi$. Formally, we require
    \begin{enumerate}
        \item {\bf Uniqueness:} no values $(pk,x,y_1,y_2,\pi_1,\pi_2)$ can satisfy {\sc VRFverify}$ (x,y_1,\pi_1,pk)=${\sc VRFverify}$(x,y_2,\pi_2,pk)$ when $y_1 \neq y_2$.
        \item {\bf Provability:} if $(y,\pi)=${\sc VRFprove}$(x,sk)$, then
        {\sc VRFverify}$(x,y,\pi,sk)=1$.
        \item {\bf Pseudorandomness:} for any
        probabilistic polynomial-time algorithm $\mathcal{A}=(A_1,A_2)$, who does not query its oracle on $x$,
        \begin{eqnarray*}
            \Pr\left[z=z' \; \left|\; \begin{array}{l}
            (pk,sk) \gets \text{\sc Gen}(1^\kappa);\\
            (x,st) \gets A_1^{\text{\sc VRFprove}(\cdot)}(pk);\\
            y_0=F_{sk}(x);\\
            y_1\gets \{0,1\}^{b(k)}; \\
            z\gets \{0,1\};\\
            z'\gets A_2^{\text{\sc VRFprove}(\cdot)}(y_z,st)
            \end{array} \right.
            \right] \; \leq \; \frac12 + negl(\kappa)
        \end{eqnarray*}
    \end{enumerate}

\end{definition}

This ensures that the output of a VRF is computationally indistinguishable from a random number even if the public key $pk$ and the function {\sc VRFprove} is revealed.

\section{Prediction inspired bribing attacks} 
\label{app:prediction_bribery}
Proof-of-work (PoW) protocols such as Nakamoto's protocol for Bitcoin achieve high security while maintaining a high unpredictability as to which miners can propose  future blocks. A very attractive feature of this PoW protocol is that nodes that mine a valid block have no further ability to update the block after they have solved the mining puzzle, since the nonce seals the block making it tamper-proof. Thus no node knows whether they have the power to propose the block till they solve the puzzle, and once they solve the puzzle, they have no future rights to alter the content. 

This causality is reversed in proof-of-stake (PoS) protocols: usually, the node that is eligible to propose a block knows {\em a priori} of its eligibility before proposing a block. This makes PoS protocols vulnerable to a new class of serious attacks not found in the PoW setting. We will show that a set of miners controlling an infinitesimal fraction of the stake can potentially completely undermine the security of the protocol. We demonstrate that the longer the prediction window, the more serious the attack space is. This raises an important questions as to whether it is possible to design a secure proof-of-stake protocol which has minimal prediction window. 

We point out that an existing work \cite{brown2019formal} has already raised this issue that PoS protocols are forced make a tradeoff between predictability and NaS attacks. While that work mainly concerned itself with incentive attacks, our concern here is adversarial attacks that compromise consensus. We point out that even in the adversarial setting, all provably secure PoS protocols have a long prediction window; the main contribution of this paper is to  design such a protocol and show its security up to 50\% of adversarial stake.

\subsection{Adaptive Adversaries and the VRF Attack}

\begin{figure}[ht]
    \centering
\includegraphics[width=0.8\textwidth]{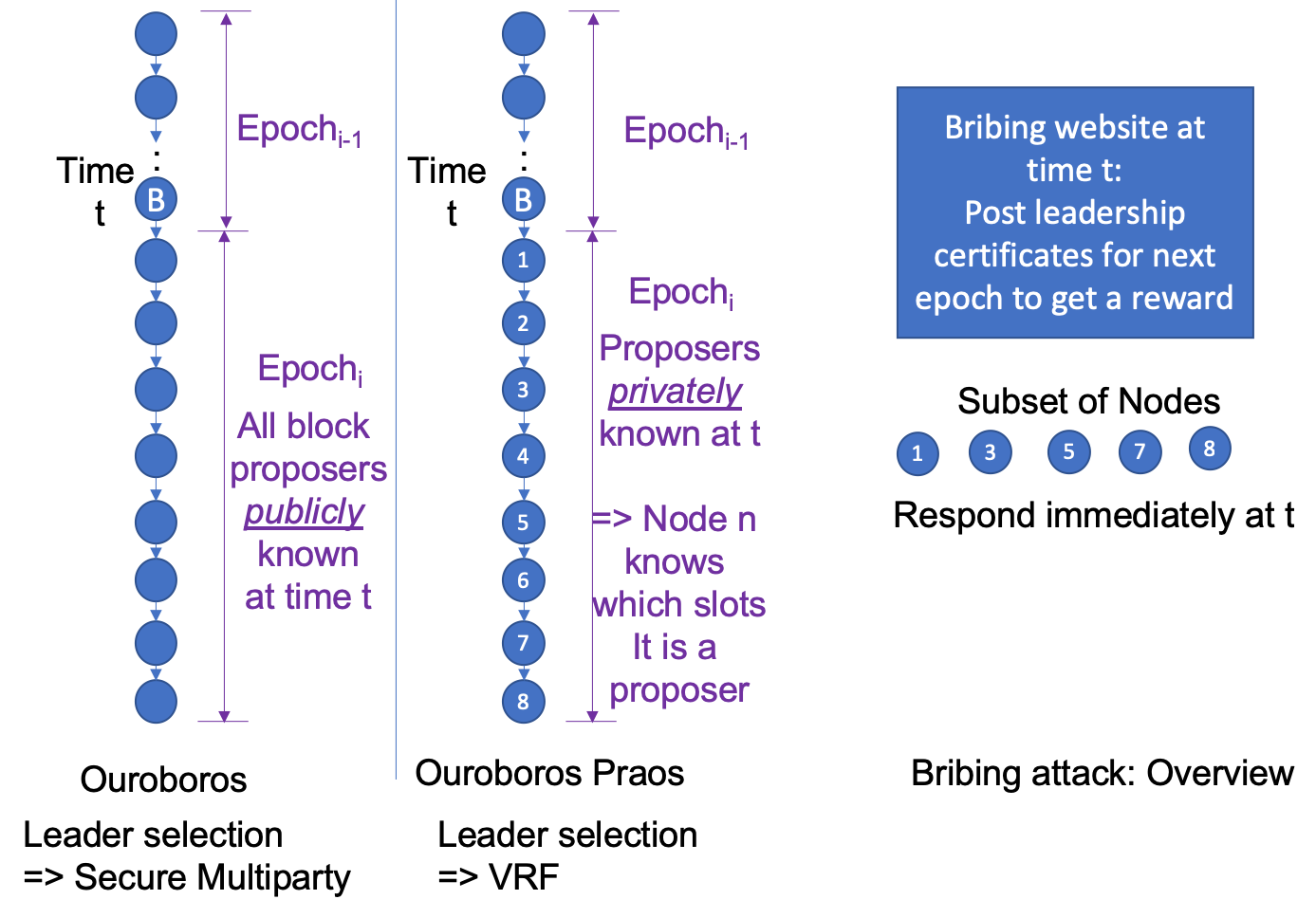}
\caption{Structure of bribing attack} \label{fig:bribing_structure}
\end{figure}

A popular model that has been proposed to capture the effect of future prediction in the PoS setting is the so-called adaptive adversary model \cite{david2018ouroboros, gilad2017algorand}. In the adaptive adversary model, a node remains honest unless corrupted by an adversary (who can change who they are corrupting based on the public state). The adversary has a bound on how many nodes it can corrupt at any given time. To defend against an adaptive adversary model, many protocols have moved from global predictability of future block proposers (i.e., everyone knows who the future block leaders are) \cite{kiayias2017ouroboros} to local predictability  (i.e., each miner knows when in the future they will propose a block) \cite{david2018ouroboros, gilad2017algorand}. The local predictability is achieved using a Verifiable Random Function (VRF) \cite{micali1999verifiable} based leader-election. 

However, the adaptive adversary model assumes that miners do not have any independent agency but rather only get corrupted based on an adversary's instructions. An adversary can easily circumvent this assumption by establishing a website where it can offer a bribe to anyone who posts their credentials for proposing blocks in an upcoming epoch of time. Thus even when the node's future proposer status is not public knowledge, this bribing website can solicit such information and help launch serious attacks (see Fig.~\ref{fig:bribing_structure}).

\subsection{Longest Chain protocols}
   We first consider longest-chain PoS protocols, in order to demonstrate our prediction attacks. We begin with a definition of prediction window of a protocol.
   
\begin{definition}[$W$-predictable]
Given a PoS protocol $\Pi_{\rm PoS}$, let $\mathcal{C}$ be a valid blockchain ending with block $B$ with a time stamp $t$. We say a block $B$ is $w(B)$-predictable, if
there exists a time $t_1>t$ and a block $B_1$ with a time stamp $t_1$ such that $(i)$ $B_1$ can be mined by miner (using its private state and the common public state) at time $t$; and 
$(ii)$ $B_1$ can be appended to $\mathcal{C'}$ to form a valid blockchain for any valid chain
$\mathcal{C'}$ that extends  $\mathcal{C}$ by appending  $w-1$ 
 valid blocks with time stamps  within the interval $(t,t_1)$. By taking the maximum over the prediction parameter over all blocks in $\Pi_{\rm PoS}$, i.e., let $W = \max_{B} w(B)$, we say $\Pi_{\rm PoS}$ is $W$-predictable.  $W$ is the size of the prediction window measured in units of number of blocks. 
\label{def:predapp}
\end{definition}

We note that our definition is similar to the definition of $W$-locally predictable protocols in \cite{brown2019formal}. We note furthermore that longest chain protocols also have a $\kappa$-deep confirmation policy, where a block embedded deep enough is deemed to be confirmed. 

\subsection{Prediction Attack on $W$-predictable protocols}
  Let us consider a $W$-predictable protocol, where the prediction window $W$ is longer than the confirmation window $\kappa$. We note that the prediction window of many existing protocols are quite large, as demonstrated in Table~\ref{tbl:intro} and therefore this is a reasonable assumption. We will consider the alternative case ($W << \kappa$ in the upcoming subsection).
  
\begin{figure}[ht]
    \centering
\includegraphics[width=0.8\textwidth]{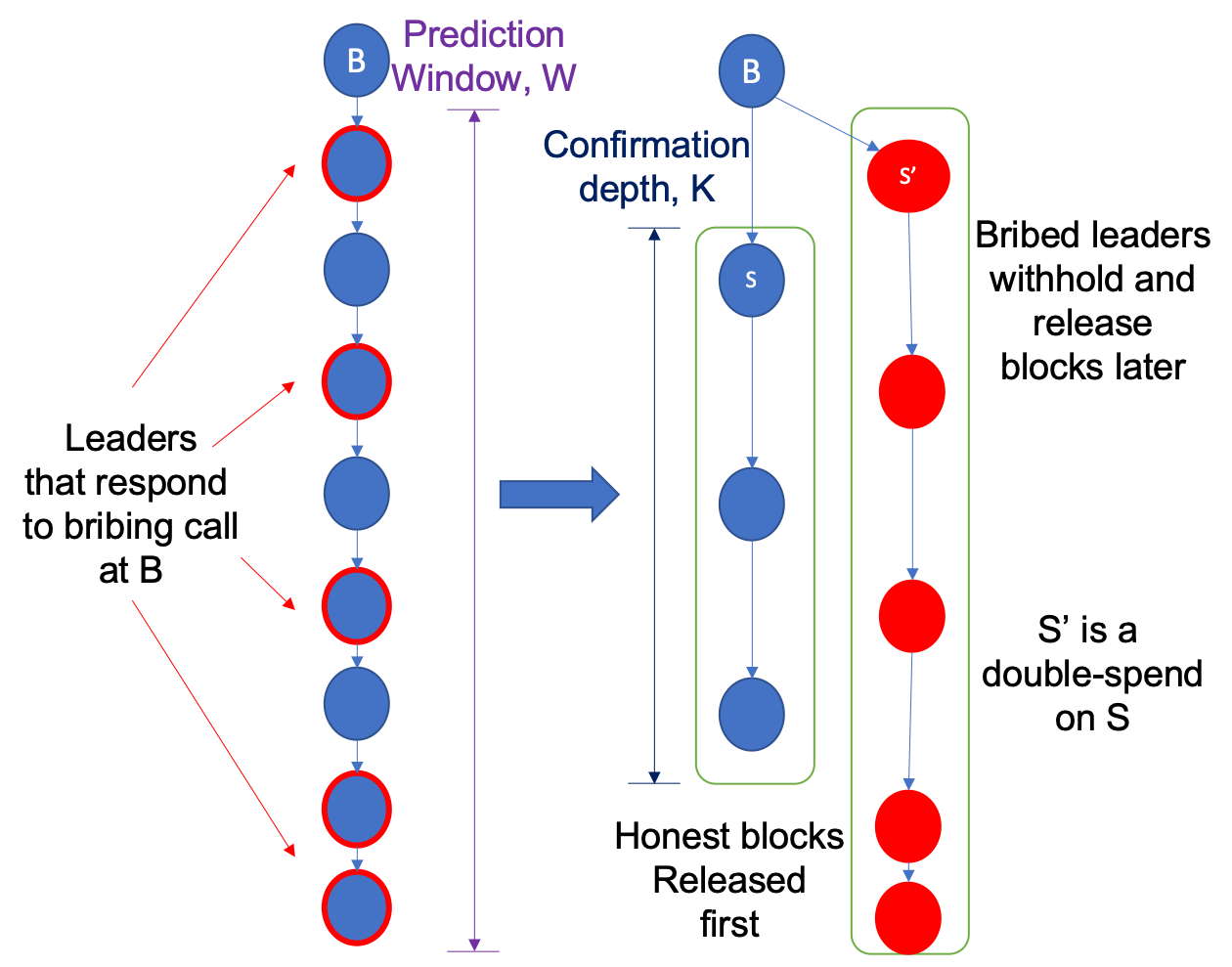}
\caption{Prediction attack on Ouroboros Praos.} \label{fig:ouroboros}
\end{figure}
    Consider block B that has been mined at time $t$ (assume that B is $W$-predictable, since such a block exists by definition). The adversary launches a prediction attack by launching a website where it announces a reward for miners which possess a future block proposal slot. Some of the leaders respond to this call (these are shown with a red outer circle in Fig.~\ref{fig:ouroboros}). We note that while the adversary requires $\kappa+1$ leaders out of $2\kappa+1$ slots to respond to the bribe, the total stake represented by these bribed leaders can be a very tiny fraction of the total stake. The adversary bribes these leaders to sign a forked version of the blockchain that it hoards till the block $s$ is confirmed by a $\kappa$-deep honest chain. After that point, the adversary releases the hoarded blockchain to all the users thus switching the longest chain and confirming a block $s'$ (which contains a double spend) instead of $s$. We note that this attack is indistinguishable from network delay since none of the bribed leaders sign multiple blocks with a single leadership certificate - thus nodes that participate in the bribing attack have plausible deniability.

\subsubsection{Bribing can enlarge prediction window }
\begin{figure}[ht]
    \centering
\includegraphics[width=0.6\textwidth]{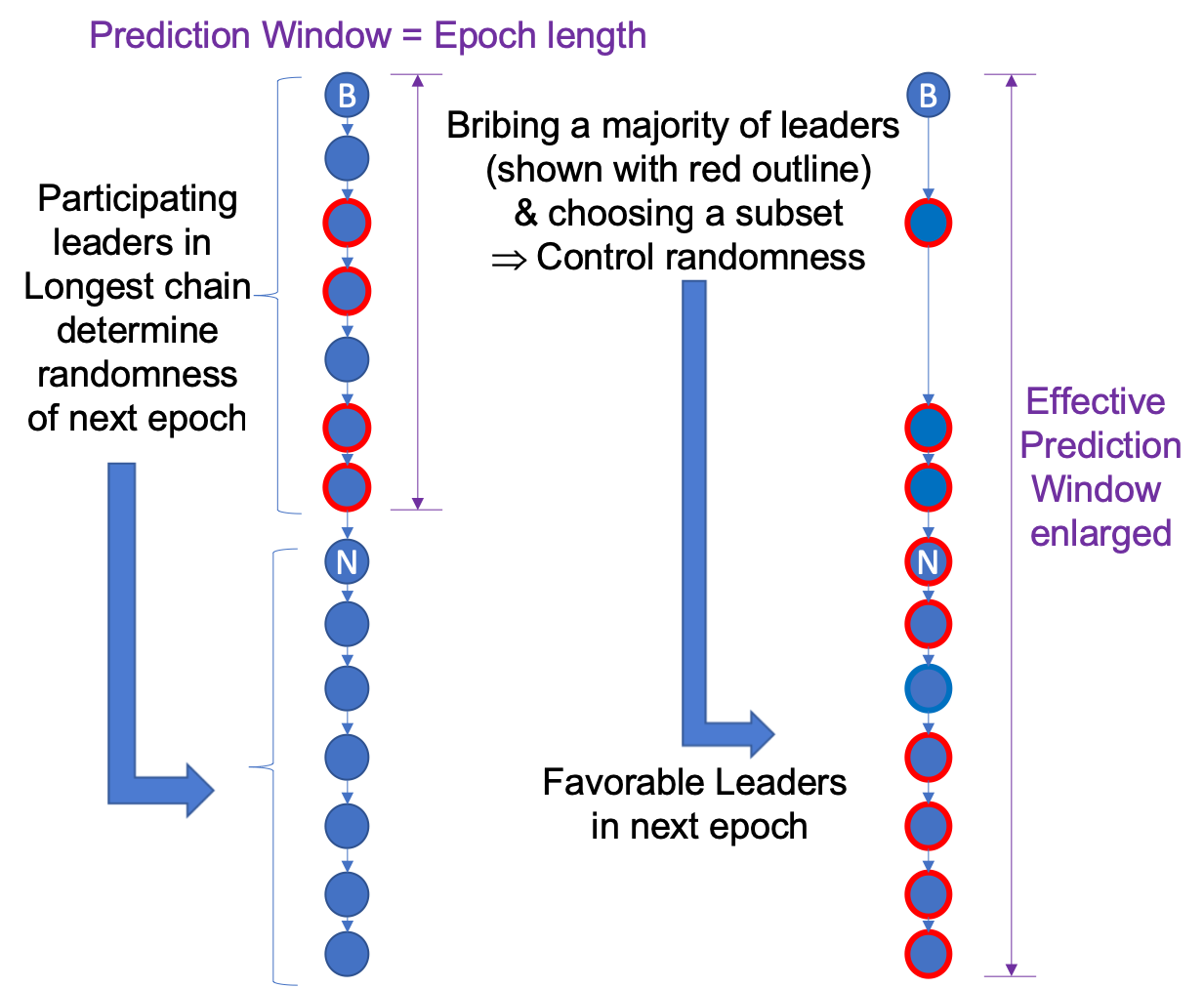}
\caption{Randomness grinding enlarges prediction ability.}
\end{figure}

We note that the previous attack applies whenever the prediction window $W$ is larger than the confirmation depth $\kappa$, this attack can be launched. However, here we will briefly note that even when $W$ is smaller than $\kappa$, the protocols still have a prediction problem. This is because in longest-chain PoS protocols such as Ouroboros, Praos, Snow White, the randomness is updated every epoch and a key assumption for the updated randomness to be unbiased is that a majority of the previous leaders were honest. However, by bribing the previous leaders, the adversary can bias the randomness (for example, by choosing a subset of proposers), thus leading itself to more favorable leadership slots in the upcoming epoch (and effectively enlarging the prediction window). These protocols offer no protection against these bribed randomness grinding attacks and hence their security parameter is limited to be of the same order as the epoch size. We note that while our proposed protocol has a similar structure as Ouroboros, our analysis proves security against adversarial randomness grinding (the so-called NaS attack).

\subsection{Prediction Attack on BFT-based PoS protocols}
\label{app:bftprediction}
While we have focused on longest chain PoS protocols in the previous section, here we consider the other large family of PoS protocols that are based on Byzantine Fault Tolerant consensus (BFT). Some PoS-based BFT protocols work with the same committee for many time-slots thus giving raise to prediction based attacks. Furthermore, in many BFT protocols, a single leader proposes blocks till evicted for wrong-doing \cite{yin2018hotstuff}, thus making the prediction problem worse. Among BFT-based PoS protocols, the one with the least prediction window is Algorand  \cite{chen2016algorand,gilad2017algorand}. We will demonstrate a fatal prediction attack on Algorand (other BFT protocols which are even more predictable are naturally attacked as well). 

\begin{figure}[ht]
    \centering
\includegraphics[width=0.8\textwidth]{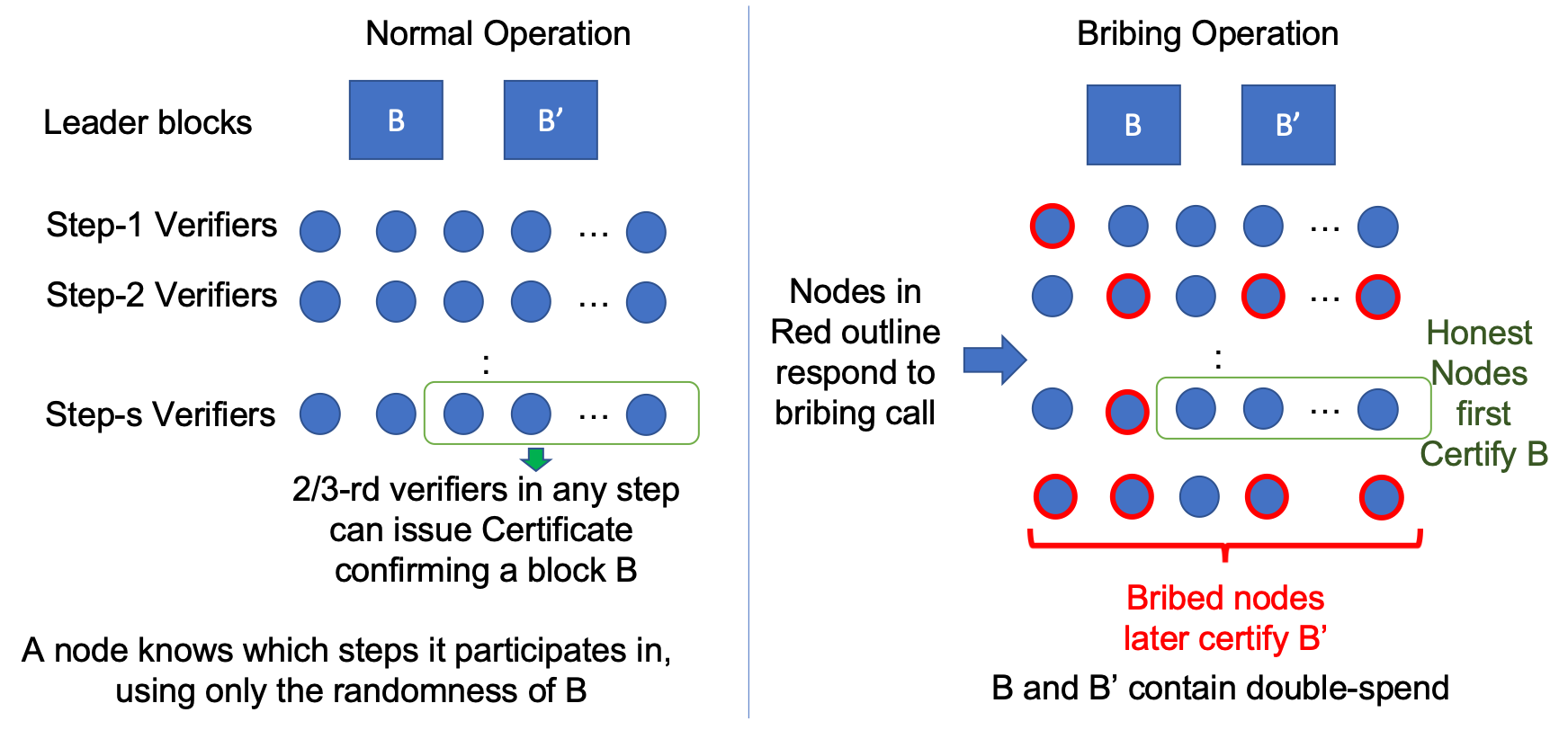}
\caption{Algorand Bribing Attack}
\end{figure}

In Algorand, at each round, a set of leaders and many sets of committees are elected using VRF from the previous finalized block. The leaders and committee members can construct their membership certificates from the previous finalized block's randomness, and others do not know their identities till they reveal themselves. The BFT consensus process proceeds in steps and each step is run by a different committee - a block is considered finalized if it is voted on by a $2/3$-majority of committee members at any step.

 This feature is used to prove that Algorand is secure against adaptive adversaries. However, we show that Algorand is not immune to the bribing-based prediction attack (similar to the one for longest-chain protocols). Since the protocol does not use a sequence of blocks for confirmation, the prediction window defined for longest-chain protocols is not the appropriate measure for prediction-attacks in this protocol. Rather, the appropriate measure is what fraction of the committee participants are locally known when a block is proposed. For Algorand, all of the committee participants are known to themselves. 

 Suppose in a given round, there are two leaders blocks $B$ and $B'$ (the latter block may contain a double-spend relative to $B$). The adversary solicits using a website the committee members to post their certificates. If the adversary is able to obtain a $2/3$-quorum in any step, then the adversary can use that quorum to sign a different certificate than the one the honest nodes signed. In particular, the adversary waits for a honest quorum to sign a certificate for $B$ in order to confirm the transaction and then signs a bribed quorum from a different step to certify $B'$. This enables the adversary to reverse a confirmation. We note that none of the bribed miners have to double sign a block since they would not have been elected in other steps of the quorum. 
 
 We note that establishing a stalling-attack is even easier in this model - the adversary needs a $1/3$-fraction of committee members in a step to remain silent in order to stall the progress of Algorand. Given that miners are randomly sampled, this quorum may hold a very small fraction of stake (which the adversary may compensate for using its bribe). This stalling attack can be used to launch extortion attacks demanding a large amount of money to un-stall the network.

\subsection{Summary of Prediction Attacks}
 We have demonstrated that both longest-chain and BFT based protocols are highly vulnerable to prediction-based security attacks (compromising the safety and liveness of the system). This motivates the design of a new PoS protocol that can only be predicted with a look ahead window, much shorter than the confirmation window.

\section{Proofs for \S5}
\label{app:proofs}
\subsection{Proof of Lemma \ref{lem:infinite_E}}
\label{app:proof_0_delay}

In this proof, we renormalize time such that $\lambda_h = 1$, and we set $\lambda_a = \lambda$ to simplify notations.

The random processes of interest start from time $0$. To look at the system in stationarity, let us extend them to $-\infty < t < \infty$. More specifically, define $\tau_{-1}, \tau_{-2}, \ldots$ such that together with $\tau_0, \tau_1, \ldots$ we have a double-sided infinite Poisson process of rate $1$. Also, for each $i < 0$, we define an  independent copy of a random adversarial tree $\tt_i$ with the same distribution as $\tt_0$. 

These extensions allow us to extend the definition of $E_{ij}$ to all $i,j$, $-\infty < i< j < \infty$, and define $E_j$ to be:
$$ E_j = \bigcap_{i< j} E_{ij}.$$

Note that $E_j \subset F_j$, so to prove that $F_j$ occurs for infinite many $j$'s with probability $1$, it suffices to prove that $E_j$ occurs for infinite many $j$'s with probability $1$. This is proved in the following.

Define $R_j = \tau_{j+1}- \tau_j$ and 
\begin{equation}
\zz_j = (R_j, \tt_j).    
\end{equation}
Consider the i.i.d. process $\{\zz_j\}_{-\infty < j < \infty}$.
Now,
\begin{eqnarray*}
E_{ij} &= & \mbox{event that $D_i(t - \tau_i) < A_h(t) - A_h(\tau_i)$ for all $t > \tau_j$}\\
&= & \mbox{event that $D_i(\tau_k - \tau_i) < A_h({\tau_k}^-) - A_h(\tau_i)$ for all $k > j$} \\
&= & \mbox{event that $D_i(\tau_k - \tau_i) < A_h(\tau_k) - A_h(\tau_i) - 1$ for all $k > j$}\\
&= & \mbox{event that $D_i(\tau_k - \tau_i) < k-i-1$ for all $k > j$}\\
&= & \mbox{event that $D_i(\sum_{m = i}^{k-1} R_m) < k - i - 1$ for all $k > j$}
\end{eqnarray*}
Hence $E_j = \cap_{i < j} E_{ij}$ has a time-invariant dependence on $\{\zz_i\}$. This means that $p = P(E_j)$ does not depend on $j$. Since $\{\zz_j\}$ is i.i.d. and in particularly ergodic,  with probability $1$, the long term fraction of $j$'s for which $E_j$ occurs is $p$, which is nonzero if $p \not = 0$. This is the last step to prove. 

Let
\begin{equation}
    E_0 = \mbox{event that $D_i(\sum_{m = i}^{k-1} R_m) < k - i-1$ for all $k > 0$ and $i < 0$}
\end{equation}
and 
\begin{equation}
\label{eq-Bik}
    B_{ik} = \mbox{event that $D_i(\sum_{m = i}^{k-1} R_m) \ge  k - i-1$}
\end{equation}
then
\begin{equation}
    E_0^c = \bigcup_{k >0, i < 0}  B_{ik}.
\end{equation}
Let us fix a particular $n >0$, and define:
\begin{equation}
    G_n = \mbox{event that $R_m < \frac{1}{n}$ for $m = -n, -n+1, \ldots, -1, 0, +1, \ldots, n-1$}
\end{equation}
Then 
\begin{eqnarray}
P(E_0) & \ge & P(E_0|G_n)P(G_n)\\
& = & \left ( 1 - P(\cup_{k>0,i<0}
B_{ik}|G_n) \right) P(G_n)\\
& \ge & \left ( 1 - \sum_{k>0,i<0} P(B_{ik}|G_n) \right) P(G_n)\\
& \ge &  ( 1 - a_n - 2 b_n - c_n) P(G_n)
\label{eq:e0_bound}
\end{eqnarray}
where
\begin{eqnarray}
a_n & := & \sum_{(i,k): -n \le i < 0 < k \le n} P(B_{ik}|G_n)\\
b_n & := & \sum_{(i,k):  -n \le i < 0 <  n < k}P(B_{ik}|G_n)\\
c_n & := & \sum_{(i,k):  i < -n <  n < k}P(B_{ik}|G_n).
\end{eqnarray}
The last inequality (\ref{eq:e0_bound}) comes from the fact $P(B_{ik}|G_n) = P(B_{-k,-i}|G_n)$.


Using (\ref{eq:Ofertheo}), we can bound $P(B_{ik}|G_n)$. Consider three cases:

\noindent
{\bf  Case 1:} $-n \le i <0 <  k \le n$:

For $k-i < \sqrt[4]{n}$, we have that
\begin{eqnarray}
P(B_{ik}|G_n) & = & P(D_i(\sum_{m = i}^{k-1} R_m) \ge  k - i - 1|G_n)  
\leq P(D_i(\frac{1}{\sqrt[4]{n^3}}) \geq 1)\nonumber\\
&= & P(\Gamma_\lambda \leq \frac{1}{\sqrt[4]{n^3}}) = 1-e^{-\frac{\lambda}{\sqrt[4]{n^3}}} <  \frac{\lambda}{\sqrt[4]{n^3}},
\end{eqnarray}
where $\Gamma_\lambda \sim \text{Exp}(\lambda)$.

For $k-i \geq \sqrt[4]{n}$, we have that
\begin{eqnarray}
P(B_{ik}|G_n) & = & P(D_i(\sum_{m = i}^{k-1} R_m) \ge  k - i - 1|G_n)  \nonumber \\
&\leq & P(D_i(\frac{k-i}{n}) \geq \sqrt[4]{n}-1) \nonumber \\
&\leq & P(D_i(2) \geq \sqrt[4]{n}-1) 
\leq  e^{-\sqrt[4]{n}+1 + 2e\lambda},
\end{eqnarray}
where the last inequality follows from (\ref{eq:Ofertheo}).  Summing these terms, we have:

\begin{eqnarray*}
a_n & = & \sum_{(i,k): -n \le i < 0 < k \le n} P(B_{ik}|G_n)\\
& \le & {\sum_{(i,k): -n \le i < 0 < k \le n, k-i < \sqrt[4]{n}} \frac{\lambda}{\sqrt[4]{n^3}}} +  {\sum_{(i,k): -n \le i < 0 < k \le n, k-i \geq \sqrt[4]{n}} e^{-\sqrt[4]{n}+1 + 2e\lambda}} \\
& \le & \frac{\lambda}{\sqrt[4]{n}} +  {\sum_{(i,k): -n \le i < 0 < k \le n, k-i \geq \sqrt[4]{n}} e^{-\sqrt[4]{n}+1 + 2e\lambda}} : = \bar{a}_n
\end{eqnarray*}
which is bounded and moreover $\bar{a}_n\rightarrow 0$ as $n \rightarrow \infty$.

\noindent
{\bf Case 2:} $-n \le i < 0 < n < k$:

    \begin{eqnarray*}
    &&P(B_{ik}|G_n)\\
    &\le & P\left (B_{ik}|G_n,2 + \sum_{m = n}^{k-1} R_m < (k-i-1) \frac{1+\lambda e}{2\lambda e} \right)\\
    && + P\left (2+  \sum_{m = n}^{k-1} R_m > (k-i-1)  \frac{1+\lambda e}{2\lambda e} \right) \\
    & \le &  e^{-\frac{1-\lambda e}{2}(k-i-1)} + A_1 e^{-\alpha (k-i-1)}
        \end{eqnarray*}
        for some positive  constants $A_1, \alpha$ independent of $n,k,i$. The first term in the last inequality follows from (\ref{eq:Ofertheo}), and the second term follows from the fact that $(1+\lambda e)/(2\lambda e) > 1$ and the $R_i$'s have mean $1$. 
Summing these terms, we have:
\begin{eqnarray*}
b_n & = & \sum_{(i,k):  -n \le i < 0 <  n < k}P(B_{ik}|G_n) \\
& \le  & \sum_{(i,k):  -n \le i < 0 <  n < k} \left [e^{-\frac{1-\lambda e}{2}(k-i-1)} + A_1 e^{-\alpha (k-i-1)}\right]: = \bar{b}_n 
\end{eqnarray*}
which is bounded and moreover $\bar{b}_n \rightarrow 0$ as $n \rightarrow \infty$.

\noindent
{\bf Case 3:} $i < -n <  n < k$:

    \begin{eqnarray*}
    &&P(B_{ik}|G_n)\\
    &\le & P\left (B_{ik}|G_n,2+ \sum_{m = -i}^{-n-1} R_m + \sum_{m = n}^{k-1} R_m < (k-i-1) \frac{1+\lambda e}{2\lambda e} \right)\\
    && + P\left (2+ \sum_{m = -i}^{-n-1} R_m + \sum_{m = n}^{k-1} R_m > (k-i-1)  \frac{1+\lambda e}{2\lambda e} \right) \\
    & \le &  e^{-\frac{1-\lambda e}{2}(k-i-1)} + A_2 e^{-\alpha (k-i-1)}
        \end{eqnarray*}
        for some positive  constants $A_2, \alpha$ independent of $n,k,i$. The first term in the last inequality follows from (\ref{eq:Ofertheo}), and the second term follows from the fact that $(1+\lambda e)/(2\lambda e) > 1$ and the $R_i$'s have mean $1$. 
        
Summing these terms, we have:
\begin{eqnarray*}
c_n & = & \sum_{(i,k):  i < -n <  n < k}P(B_{ik}|G_n) \\
& \le  & \sum_{ (i,k): i < -n <  n < k} \left [e^{-\frac{1-\lambda e}{2}(k-i-1)} + A_2 e^{-\alpha (k-i-1)}\right] := \bar{c}_n
\end{eqnarray*}
which is bounded and moreover $\bar{c}_n \rightarrow 0$ as $n \rightarrow \infty$.

Substituting these bounds in (\ref{eq:e0_bound}) we finally get:
\begin{equation}
    P(E_0) > [1- (\bar{a}_n+ 2\bar{b}_n+\bar{c}_n)]P(G_n)
\end{equation}
By setting $n$ sufficiently large such that $\bar{a}_n,\bar{b}_n$ and $\bar{c}_n$ are sufficiently small, we conclude that $P(E_0)> 0$.

\subsection{Proof of Lemma \ref{lem:infinite_many_F}}
\label{app:proof_delay}

In this proof, we fix $g = e^{-\lambda_h \Delta}$ and renormalize time such that $\lambda_h = 1$, and we set $\lambda_a = \lambda$ to simplify notations.

The random processes of interest start from time $0$. To look at the system in stationarity, let us extend them to $-\infty < t < \infty$. More specifically, define $\tau_{-1}, \tau_{-2}, \ldots$ such that together with $\tau_0, \tau_1, \ldots$ we have a double-sided infinite Poisson process of rate $1$. Also, for each $i < 0$, we define an  independent copy of a random adversarial tree $ \tt_i$ with the same distribution as $\tt_0$.

These extensions allow us to extend the definition of $\hat{E}_{ij}$ to all $i,j$, $-\infty < i< j < \infty$, and define $\hat{E}_j$ and $\hat{V}_j$ to be:
$$ \hat{E}_j = \bigcap_{i < j} \hat{E}_{ij}$$ and 
$$ \hat{V}_j = \hat{E}_j \cap U_j.$$

Note that $\hat{V}_j \subset \hat{U}_j$, so to prove that $\hat{U}_j$ occurs for infinite many $j$'s with probability $1$, it suffices to prove that $\hat{V}_j$ occurs for infinite many $j$'s with probability $1$. This is proved in the following.



Define $R_j = \tau_{j+1}- \tau_j$ and 
\begin{equation*}
\zz_j = (R_j, \tt_j).    
\end{equation*}
Consider the i.i.d. process $\{\zz_j\}_{-\infty < j < \infty}$.
Now,
\begin{eqnarray*}
U_j \cap \hat{E}_{ij} &= & U_j \cap \mbox{event that $D_i(t - \tau_i) < H_h(t-\Delta) - H_h(\tau_i)$ for all $t > \tau_j + \Delta$}\\
&= & U_j \cap \mbox{event that $D_i(t +\Delta - \tau_i) < H_h(t) - H_h(\tau_i)$ for all $t > \tau_j$}\\
&= & U_j \cap \mbox{event that $D_i({\tau_k}^- +\Delta - \tau_i) < H_h({\tau_k}^-) - H_h(\tau_i)$ for all $k > j$}\\
&= & U_j \cap \mbox{event that $D_i(\sum_{m = i}^{k-1} R_m + \Delta) < H_h(\tau_{k-1} ) - H_h(\tau_i)$ for all $k > j$}
\end{eqnarray*}
Hence $\hat{E}_j \cap U_j = \bigcap_{i < j} \hat{E}_{ij} \cap U_j$ has a time-invariant dependence on $\{\zz_i\}$, which means that $p = P(\hat{V}_j)$ does not depend on $j$. Since $\{\zz_j\}$ is i.i.d. and in particularly ergodic,  with probability $1$, the long term fraction of $j$'s for which $\hat{V}_j$ occurs is $p$, which is nonzero if $p \not = 0$. This is the last step to prove.


\begin{equation*}
    P(\hat{V}_0) = P(\hat{E}_0|U_0)P(U_0) = P(\hat{E}_0|U_0)P(R_0>\Delta)P(R_{-1}>\Delta) = g^2 P(\hat{E}_0|U_0).
\end{equation*}

It remains to show that $P(\hat{E}_0|U_0)>0$.

\begin{equation}
    \hat{E}_0 = \mbox{event that $D_i(\sum_{m = i}^{k-1} R_m + \Delta) < H_h(\tau_{k-1} ) - H_h(\tau_i)$ for all $k > 0$ and $i < 0$}
\end{equation}
Let 
\begin{equation}
    \hat{B}_{ik} = \mbox{event that $D_i(\sum_{m = i}^{k-1} R_m + \Delta) \ge  H_h(\tau_{k-1} ) - H_h(\tau_i)$}
\end{equation}
then
\begin{equation}
    \hat{E}_0^c = \bigcup_{k >0, i < 0}  \hat{B}_{ik}.
\end{equation}
Let us fix a particular $n > 2\Delta > 0$, and define:
\begin{equation}
    G_n = \mbox{event that $D_m(3n) = 0 $ for $m = -n, -n+1, \ldots, -1, 0, +1, \ldots, n-1,n$}
\end{equation}
Then 
\begin{eqnarray}
P(\hat{E}_0 | U_0) & \ge & P(\hat{E}_0|U_0,G_n)P(G_n|U_0)\\
& = & \left ( 1 - P(\cup_{k>0,i<0} \hat{B}_{ik}|U_0,G_n) \right) P(G_n|U_0)\\
& \ge & \left ( 1 - \sum_{k>0,i<0} P(\hat{B}_{ik}|U_0,G_n) \right) P(G_n|U_0)\\
& \ge &  ( 1 - a_n - b_n) P(G_n|U_0) \label{eqn:up_bound}
\end{eqnarray}
where
\begin{eqnarray}
a_n & := & \sum_{(i,k): -n \le i < 0 < k \le n} P(\hat{B}_{ik}|U_0,G_n)\\
b_n & := & \sum_{(i,k):  i < -n \text{~or~} k > n }P(\hat{B}_{ik}|U_0,G_n).
\end{eqnarray}

Using (\ref{eq:Ofertheo}), we can bound $P(\hat{B}_{ik}|U_0, G_n)$. Consider two cases:

\noindent
{\bf  Case 1:} $-n \le i <0 <  k \le n$:
\begin{eqnarray*}
P(\hat{B}_{ik}|U_0,G_n) &=& P(\hat{B}_{ik}|U_0, G_n,\sum_{m = i}^{k-1} R_m + \Delta \leq 3n) + P(\sum_{m = i}^{k-1} R_m + \Delta > 3n |U_0,G_n) \\
& \leq & P(\sum_{m = i}^{k-1} R_m + \Delta > 3n |U_0,G_n) \\
& \leq & P(\sum_{m = i}^{k-1} R_m > 5n/2 |U_0) \\
& \leq & P(\sum_{m = i}^{k-1} R_m  > 5n/2)/P(U_0)  \\
& \leq & A_1 e^{-\alpha_1 n} 
\end{eqnarray*}
for some positive  constants $A_1, \alpha_1$ independent of $n,k,i$. The last inequality follows from the fact that $R_i$'s are iid 
exponential random variables of  mean $1$. Summing these terms, we have:

\begin{eqnarray*}
a_n & = & \sum_{(i,k): -n \le i < 0 < k \le n} P(B_{ik}|U_0, G_n)
\leq \sum_{(i,k): -n \le i < 0 < k \le n} A_1 e^{-\alpha_1 n} : = \bar{a}_n,
\end{eqnarray*}
which is bounded and moreover $\bar{a}_n\rightarrow 0$ as $n \rightarrow \infty$.

\noindent {\bf  Case 2:} $k>n \text{~or~} i<-n$:

For $0<\varepsilon<1$, let us define event $W^{\varepsilon}_{ik}$ to be:
\begin{equation*}
    W^{\varepsilon}_{ik} = \mbox{event that $H_h(\tau_{k-1}) - H_h(\tau_i) \geq (1-\varepsilon)g(k-i-1)$}.
\end{equation*}
Then we have
\begin{equation*}
    P(\hat{B}_{ik}|U_0,G_n) \leq P(\hat{B}_{ik}|U_0, G_n,W^{\varepsilon}_{ik}) + P({W^{\varepsilon}_{ik}}^c |U_0,G_n).
\end{equation*}

Let $X_j$ be a Bernoulli random variable such that $X_j = 1$ if and only if $R_{j-1} > \Delta$, i.e., the $j$-th honest block is a non-tailgater. Since $R_j$'s are i.i.d exponential random variables with mean $1$, we have that $X_j$'s are also i.i.d and $P(X_j = 1)=g$.
By the definition of $H_h(\cdot)$, we have $H_h(\tau_{k-1}) - H_h(\tau_i) = \sum_{j=i+1}^{k-1} X_j$, then
\begin{eqnarray}
    P({W^{\varepsilon}_{ik}}^c |U_0,G_n) &=& P(\sum_{j=i+1}^{k-1} X_j < (1-\varepsilon)g(k-i-1)|X_0 = 1, X_{1} = 1 )  \nonumber \\
    & = &P(\sum_{j=i+1}^{-1} X_j + \sum_{j=2}^{k-1} X_j < (1-\varepsilon)g(k-i-1) - 2 )  \nonumber \\
    & \leq & P(\sum_{j=i+1}^{-1} X_j + \sum_{j=2}^{k-1} X_j < (1-\varepsilon)g(k-i-3)  )  \nonumber \\
    & \leq & A_2 e^{-\alpha_2 (k-i-3)}
    \label{eqn:prob_event}
\end{eqnarray}
for some positive  constants $A_2, \alpha_2$ independent of $n,k,i$. The last inequality follows from the Chernoff bound.

Meanwhile, we have 
\begin{eqnarray*}
&~&P(\hat{B}_{ik}|U_0, G_n,W^{\varepsilon}_{ik})  \\
&\leq& P(D_i(\sum_{m = i}^{k-1}R_m + \Delta) \geq (1-\varepsilon)g(k-i-1)) | U_0,G_n,W^{\varepsilon}_{ik})  \\
& \leq & P(D_i(\sum_{m = i}^{k-1}R_m + \Delta) \geq (1-\varepsilon)g(k-i-1))\\
&&\qquad\qquad \qquad \qquad | \; U_0,G_n,W^{\varepsilon}_{ik},\sum_{m = i}^{k-1}R_m + \Delta \leq (k-i-1)\frac{g+\lambda e}{2\lambda e}) \\
&+& P(\sum_{m = i}^{k-1}R_m + \Delta > (k-i-1)\frac{g+\lambda e}{2\lambda e}|U_0,G_n,W^{\varepsilon}_{ik}) \\
&\leq& e^{-\frac{(1-2\varepsilon)g-\lambda e}{2} (k-i-1)} + P(\sum_{m = i}^{k-1}R_m + \Delta > (k-i-1)\frac{g+\lambda e}{2\lambda e}|U_0,G_n,W^{\varepsilon}_{ik})
\end{eqnarray*}
where the first term in the last inequality follows from (\ref{eq:Ofertheo}), and the second term can also be bounded:

\begin{eqnarray*}
&~&P(\sum_{m = i}^{k-1}R_m + \Delta > (k-i-1)\frac{g+\lambda e}{2\lambda e}|U_0,G_n,W^{\varepsilon}_{ik})\\
&=& P(\sum_{m = i}^{k-1}R_m + \Delta > (k-i-1)\frac{g+\lambda e}{2\lambda e}|U_0, W^{\varepsilon}_{ik}) \\
&\leq& P(\sum_{m = i}^{k-1}R_m + \Delta > (k-i-1)\frac{g+\lambda e}{2\lambda e})/P(U_0, W^{\varepsilon}_{ik}) \\
&\leq& A_3 e^{-\alpha_3(k-i-1)}
\end{eqnarray*}
for some positive  constants $A_3, \alpha_3$ independent of $n,k,i$. The last inequality follows from the fact that $(g+\lambda e)/(2\lambda e) > 1$ and the $R_i$'s have mean $1$, while $P(U_0, W^{\varepsilon}_{ik})$ is a event with high probability as we showed in (\ref{eqn:prob_event}).

Then we have 
\begin{equation*}
     P(\hat{B}_{ik}|U_0,G_n) \leq A_2 e^{-\alpha_2 (k-i-3)} + e^{-\frac{(1-2\varepsilon)g-\lambda e}{2} (k-i-1)} + A_3 e^{-\alpha_3(k-i-1)}.
\end{equation*}
Summing these terms, we have:
\begin{eqnarray*}
b_n & = & \sum_{(i,k):  i<-n \text{~or~} k > n}P(\hat{B}_{ik}|U_0,G_n) \\
& \le  & \!\!\!\!
\sum_{ (i,k): i<-n \text{~or~} k > n} \left [A_2 e^{-\alpha_2 (k-i-3)} + e^{-\frac{(1-2\varepsilon)g-\lambda e}{2} (k-i-1)} + A_3 e^{-\alpha_3(k-i-1)}\right] := \bar{b}_n
\end{eqnarray*}
which is bounded and moreover $\bar{b}_n \rightarrow 0$ as $n \rightarrow \infty$ when we set $\varepsilon$ to be small enough such that $(1-2\varepsilon)g-\lambda e > 0$.

Substituting these bounds in (\ref{eqn:up_bound}) we finally get:
\begin{equation}
    P(\hat{E}_0|U_0) > [1- (\bar{a}_n+ \bar{b}_n)]P(G_n|U_0)
\end{equation}
By setting $n$ sufficiently large such that $\bar{a}_n$ and $\bar{b}_n$ are sufficiently small, we conclude that $P(\hat{V}_0)> 0$.

{
}

\section{Growth rate of $c$-correlated adversary tree}
\label{app:growth_c}


We set $\lambda_a = \lambda$ in this section to simplify notations.

The adversary is growing a private chain over the genesis block, under the $c$-correlation.
As illustrated in Fig.~\ref{fig:ccorrelation}, the common source of randomness at a block
is only updated when the depth is a multiple of $c$
(Algorithm~\ref{alg:PoS} line:\ref{algo:c-correlation}).
We refer to such a block $b$ with ${\rm depth}(b)\%c = 0$ as
a {\em godfather-block}.
\begin{align*}
&{\rm RandSource}(b):= \nonumber\\ 
&\begin{cases}
\text{\sc VRFprove}\big({\rm RandSource}(\text{parent}(b)), {\rm time}, {\sf VRF}.sk(b)\big) , &\quad \text{if } {\rm depth}\big(b\big)\%c=0,\\
{\rm RandSource}(\text{parent}(b)), &\quad \text{otherwise}
\end{cases}
\end{align*}
The randomness of a block changes only at \textit{godfather-blocks}. In other words, for $n\in\mathbb{Z}$, blocks along a chain at depths $\big\{nc, nc+1, nc+2,\cdots,nc+c-1\big\}$ share a common random number.
Two blocks are called {\em siblings-blocks} if they have the same parent block.
Given this shared randomness, the adversary now has a freedom to choose where to place the newly generated blocks. The next theorem provides a dominant strategy, that creates the fastest growing private tree.

\begin{lemma}
\label{lemma:optimal_strategy_ccorrelation}
Under $c$-correlation,
the optimal adversarial strategy to grow the tree fast is to only fork at the parents of godfather-blocks. 
\end{lemma}
\begin{proof}
Note that under the security model,
several types of grinding attacks are plausible.
First, at depth multiple of $c$, the adversary can grind on the header of the parent of the godfather block, and run an independent election in every round.
Secondly, for blocks sharing the same source of randomness, once an adversary is elected a leader, it can generate multiple blocks of the same header but
appending on different blocks.
However, adding multiple blocks with the same header cannot
make the tree any higher than the optimal scheme.

\begin{figure}[h]
    \centering
    \includegraphics[width=\textwidth]{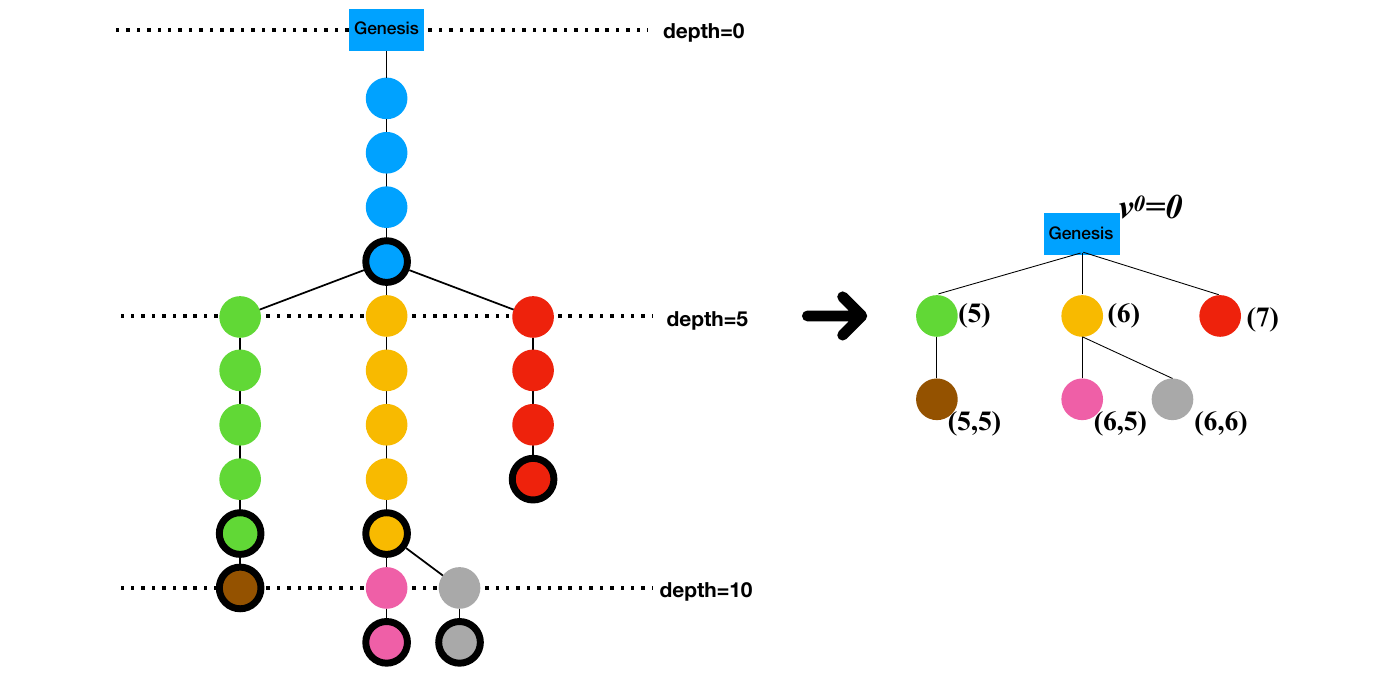}
    \caption{An example of $T^{(a)}(t)$ with $c=5$ under the optimal strategy to grow the private NaS tree. Blocks forking from the same godfather-block share the same common source of randomness, as shown by the colors. To grow the tree fast, it is optimal to grow a single chain until the next godfather block. Circles with black outlines indicate blocks that are currently mined on.}
    \label{fig:ccorrelation}
\end{figure}

Sibling non-godfather blocks share a common source of randomness and thus ‘‘mining events'' on these blocks are completely dependent.
Specifically, for sibling non-godfather blocks, there is only one leader election in each round.
As a result, for a particular non-godfather block $b$, the child-block of $b$'s any younger sibling (a sibling block mined after block $b$) should share the same header and identical source of randomness with one of $b$'s child-block.
Thus it is not necessary for a non-godfather block to have sibling blocks, that is, mining a sibling to a non-godfather block does not increase the growth rate of the longest chain in the adversarial block tree.
However, sibling god-father blocks have independent sources of randomness and thus mining multiple such block increase the growth rate of the longest chain.
\end{proof}

Here we use the representation in Fig.~\ref{fig:ccorrelation}, whose growth rate is the same as the full NaS adversarial tree as shown in Lemma~\ref{lemma:optimal_strategy_ccorrelation}. We can transform the tree $T^{(a)}(t)$ in Fig.~\ref{fig:ccorrelation} into a new random tree $T^0(t)$.  Every $c$ generations in $T^{(a)}(t)$ we can view as a single generation in $T^0(t)$; in the example of Fig.~\ref{fig:ccorrelation}, we have depth $0$, $5$, $10$ etc (i.e., all the godfather depths) as the generations in $T^0(t)$. $T^0(t)$ corresponds to a branching random walk. 
For example, the genesis block $B_0$ is the root of $T^0(t)$ at depth 0 with arrival time $0$. The children blocks of $B_0$ in $T^0(t)$ are the descendant blocks at depth $5$ in $T^{(a)}(t)$. We can order these children blocks in their arrival times. 
Consider block $B_1$ to be the first such block, then the arrival time of block $B_1$ is
\begin{equation*}
    S_1 = X_1 + X_2 + \ldots + X_{c},
\end{equation*}
where $X_i$ is the inter-arrival time between block at depth $i-1$ and block at depth $i$ in $T^{(a)}(t)$. Note that all the $X_i$'s are exponential with parameter $\lambda$, and they are all independent. Similarly, the arrival time $S_i$ of the $i$-th child of root in $T^0(t)$ is a sum of $i + c - 1$ i.i.d exponential random variables with parameter $\lambda$. 

Let the depth of the tree $T^a(t)$ and $T^0(t)$ be denoted by $D^a(t)$ and $D(t)$ respectively defined as the maximum depth of the blocks in the tree.

Similar to \S\ref{sec:analysis}, each vertices at generation $k$ can be labelled as a $k$ tuple of positive integers $(i_1,\ldots,i_k)$ with $i_j \geq c$ for $1\leq j \leq k$: the vertex $v = (i_1,\ldots,i_k)\in \I_k$ 
is the $(i_k-c+1)$-th child of vertex $(i_1,\ldots,i_{k-1})$
at level $k-1$.
Let $\I_k =\{(i_1,\ldots,i_k): i_j \geq c \rm{~for~} 1\leq j \leq k\}$, and  set $\I=\cup_{k>0} 
\I_k$. 
For such $v$ we also let $v^j=(i_1,\ldots,i_j)$, 
$j=1,\ldots,k$, denote the 
ancestor of $v$ at level $j$, with $v^k=v$. 
For notation convenience, we set $v^0=0$ as the root of the tree.

Next, let $\{\E_v\}_{v\in \I}$ be an i.i.d. family of exponential 
random variables
of parameter $\lambda$. For $v=(i_1,\ldots,i_k)\in \I_k$,
let $\W_v=\sum_{j\leq i_k} \E_{(i_1,\ldots,i_{k-1},j)}$ and
let $S_v=\sum_{j\leq k} \W_{v^j}$. This creates a labelled tree, with the following
interpretation: for $v=(i_1,\ldots,i_j)$,  the $W_{v^j}$ are the waiting
for $v^j$ to appear, measured from the appearance of $v^{j-1}$, and 
$S_v$ is the appearance time of $v$.

Let $S^*_k=\min_{v\in \I_k} S_v$. Note that $S^*_k$ is the time of appearance
of a block at level $k$ and therefore we have
\begin{equation}
  \label{eq:c_grow}
  \{D^a(t)\geq ck\}=\{D(t)\geq k\}=\{S^*_k\leq  t\}.
\end{equation}

$S^*_k$ is the minimum of a standard BRW. Introduce, for $\theta_c<0$,
the moment generating 
function
\begin{eqnarray*}
\Lambda_c(\theta_c)&=&\log \sum_{v\in \I_1} E(e^{\theta_c S_v})=
\log \sum_{j=c}^\infty E (e^{\sum_{i=1}^j \theta_c \E_i})\\
&=&\log \sum_{j=c}^\infty (E(e^{\theta_c \E_1}))^j  
=\log \frac{E^c(e^{\theta_c \E_1})}{1-E(e^{\theta_c\E_1})}.
\end{eqnarray*}
Due to the exponential law of $\E_1$,
$E(e^{\theta_c \E_1})= \frac{\lambda}{\lambda-\theta_c}$ and therefore
$\Lambda_c(\theta_c)=\log(-\lambda^c/\theta_c(\lambda-\theta_c)^{c-1})$. 

An important role is played by $\theta_c^*$, which is the negative solution to the equation $\Lambda_c(\theta_c) = \theta_c {\dot \Lambda_c(\theta_c)}$ and let $\eta_c$ satisfy that
$$\sup_{\theta_c<0} \left(\frac{\Lambda_c(\theta_c)}{\theta_c}\right)= \frac{\Lambda_c(\theta_c^*)}{\theta_c^*}=\frac{1}{\lambda \eta_c}.$$
Indeed, see \cite[Theorem 1.3]{shi}, we have the following.
\begin{proposition}
  \label{prop:prop-c}
$$\lim_{k\to\infty} \frac{S^*_k}{k}= \sup_{\theta_c<0} \left(\frac{\Lambda_c(\theta_c)}{\theta_c}\right)=\frac{1}{\lambda \eta_c}, \quad a.s.$$
\end{proposition}
In fact, much more is known, see e.g. \cite{hushi}. 
\begin{proposition}
  \label{prop:prop-c_tight}
  There exist explicit constants $c_1>c_2>0$ 
so that the sequence
$ S^*_k-k/\lambda \eta_c-c_1\log k$ is tight,
and
$$\liminf_{k\to \infty} S^*_k-k/\lambda \eta_c-c_2\log k=\infty, a.s.$$
\end{proposition}

Note that Propositions \ref{prop:prop-c},\ref{prop:prop-c_tight}
and \eqref{eq:c_grow} imply in particular
that  $D^a(t)\leq c\eta_c\lambda t$ for all large $t$, a.s., and also that
\begin{equation}
  \label{eq:c_converse}
\mbox{\rm  if
$c\eta_c\lambda>1$ then } D^a(t)>t\; \mbox{\rm for all large $t$, a.s.}.
\end{equation}

Let us define $\phi_c := c\eta_c$, then $\phi_c \lambda$ is the growth rate of private $c$-correlated NaS tree. One can check that $\phi_c$ is the solution to the same equation as in \cite{pos_revisited_arxiv}, where the same problem is solved with a differential equation approach. \cite{pos_revisited_arxiv} also proves the uniqueness of $\phi_c$ and provides an approximation for large $c$: $\phi_c  =  1+ \sqrt\frac{\ln c}{c}  + o\Big( \sqrt\frac{\ln c}{c} \Big)$.

We will need also tail estimates for the event
$D(t)>\eta_c\lambda t+x$. While such estimates can be read from 
\cite{shi}, we bring instead a quantitative statement suited 
for our
needs.
\begin{theorem}
  \label{thm:c-tail}
  For $x>0$ so that $\eta_c\lambda t+x$ is
  an integer,
  \begin{equation}
    \label{eq:c-tail}
    P(D^a(t)\geq \phi_c\lambda t+cx) = P(D(t)\geq \eta_c\lambda t+x)\leq  e^{\Lambda_c(\theta_c^*)x}.
  \end{equation}
\end{theorem}
\begin{proof}
  We use a simple upper bound. Write $m=\eta_c\lambda t+x$.
  Note that by \eqref{eq:c_grow},
  \begin{equation}
    \label{eq:c-bound}
    P(D(t)\geq m)=P(S^*_{ m}\leq t)
  \leq \sum_{v\in \I_{m}} P(S_v\leq t).
\end{equation}
  For $v=(i_1,\ldots,i_k)$, set $|v|=i_1+\cdots+i_k$. Then,
  we have that $S_v$ has the same law as 
  $\sum_{j=1}^{|v|} \E_j$. Thus, by Chebycheff's inequality,
  for $v\in \I_{m}$,
  \begin{equation}
    \label{eq:upper_1}
    P(S_v\leq  t)\leq Ee^{\theta_c^* S_v} e^{-\theta_c^*t}=\left(\frac{\lambda}
  {\lambda-\theta_c^*} \right)^{|v|} e^{-\theta_c^* t}.
\end{equation}
  And
  \begin{eqnarray}
    \label{eq:upper_2}
    \sum_{v\in \I_{m}}\left(\frac{\lambda}
  {\lambda-\theta_c^*} \right)^{|v|}
  &=&\sum_{i_1\geq c,\ldots,i_{m}\geq c} 
  \left(\frac{\lambda}
  {\lambda-\theta_c^*} \right)^{\sum_{j=1}^m i_j}\nonumber \\
  &=&\left(\sum_{i\geq c}\left(\frac{\lambda}
  {\lambda-\theta_c^*} \right)^i\right)^m= e^{\Lambda_c(\theta_c^*){m}}.
\end{eqnarray}
Combining \eqref{eq:upper_1}, \eqref{eq:upper_2} and \eqref{eq:c-bound}
yields \eqref{eq:c-tail}.
\end{proof}

\section{Nakamoto-PoS Protocol Pseudocode}
\label{app:pseudocode}
\begin{algorithm}
	{\fontsize{10pt}{10pt}\selectfont \caption{Nakamoto-PoS $(c, s, \delta)$}
	\label{alg:PoS}
\begin{algorithmic}[1]
\Procedure{Initialize}{ } 
\State ${\rm BlkTree} \gets {\rm genesis}$ \colorcomment{Blocktree}
\State ${\rm parentBk}$ $\gets {\rm genesis}$ \colorcomment{Block to mine on}
\State  unCnfTx $  \gets \phi$ \colorcomment{Blk content: Pool of unconfirmed {\sf tx}s}
\EndProcedure
\vspace{1mm}
\Procedure{PosMining}{coin} 
\While{True}
\State $\textsc{SleepUntil}({\rm SystemTime}\;\%\;\delta == 0)$ \colorcomment{System time is miner's machine time}
\State time  $\gets$  SystemTime
\State ({\sf KES}.$vk$,{\sf KES}.$sk$),({\sf VRF}.$pk$,{\sf VRF}.$sk$)
$\gets$ coin.\textsc{Keys()}
\State \maincolorcomment{Update the stake according to the stake distribution in the $s$-th last block in the main chain.}
\State stakeBk $\gets$  {\sc SearchChainUp}$({\rm parentBk},s$)
\State ${\rm stake} \gets {\rm coin}.\textsc{Stake}({\rm stakeBk})$ \label{algo:stake}
\State $\rho \gets ${\sc UpdateGrowthRate}(parentBk)
\State \maincolorcomment{Three sources of randomness: a common source (parentBk.content.RandSource),
a private source (coinSecretKey), and time.}
\State ${\rm header} \gets \langle$parentBk.content.RandSource ,\;time$\rangle$
\State $\langle {\rm hash}, {\rm proof} \rangle \gets
\text{\sc VRFprove}({\rm header},\text{{\sf VRF}.$sk$})$
    \label{algo:vrfproof}
	\colorcomment{Verifiable Random Function}
\If {hash  $< \rho \times {\rm stake}$}
        \label{algo:bias}
        \colorcomment{Block generated}
    \State \maincolorcomment{Update common source of randomness every $c$-th block in a chain as per $c$-correlation scheme}
    \If {parentBk.{\sc Height}() \;$\%\; c == c-1$} RandSource $\gets$ hash \label{algo:c-correlation}
    \Else
        { RandSource $\gets$ parentBk.content.RandSource}
    \EndIf
    \State {\rm state} $\gets$ {\sc Hash}({\rm parentBk})
    \State content $ \gets$ $\langle$ {\rm unCnfTx}, {\rm coin}, {\rm RandSource},  {\rm hash},  {\rm proof},
    state $\rangle$ \textbf{and} \textit{break}
    \label{algo:sign}
\EndIf
\EndWhile
\maincolorcomment{Return header along with signature on content}
\State \Return $\langle {\rm header}, {\rm content}, \textsc{Sign}({\rm content}, \text{{\sf KES}.$sk$})  \rangle$

\EndProcedure
\vspace{1mm}
\\
\maincolorcomment{Function to listen messages and update the blocktree}
\Procedure{ReceiveMessage}{\texttt{X}} \colorcomment{Receives messages from network}
\If {\texttt{X} is a valid {\sf tx}}
    \State undfTx $  \gets $ unCnfTx $\cup\; \{\texttt{X}\}$

	\ElsIf{ {\sc IsValidBlock}(\texttt{X})}
		\State $ \texttt{X}_{\rm fork} \gets$ the highest block shared by the main chain and the chain leading to \texttt{X}
		\State $L_{\rm fork} \gets$ min(parentBk.{\sc Height}(), \texttt{X}.{\sc Height()}) - ${\texttt{X}_{\rm fork}}$.{\sc Height()}
		\If{$L_{\rm fork}  < s$} \colorcomment{If the fork is less than $s$ blocks}
		    \If { parentBk.{\sc Height}() $< \texttt{X}$.{\sc Height}() }
			
            			\State \textsc{ChangeMainChain(\texttt{X})}
            	\colorcomment{If the new chain is longer}
    			\EndIf
    		
    		\Else
			\label{algo:lc}
			\colorcomment{check $s$-truncated longest chain rule}
			\State MainChainBk $\gets$   {\sc SearchChainDown}$({\rm parentBk},\texttt{X}_{\rm fork},s$)
			\colorcomment{find the $s$-th block down the main chain from fork}
			\State NewChainBk $\gets$  {\sc SearchChainDown}$(\texttt{X},\texttt{X}_{\rm fork},s$)
			\colorcomment{find the $s$-th block down the new chain from fork}
			\If { NewChainBk.{\rm header}.{\rm time} $<$ MainChainBk.{\rm header}.{\rm time} }
            			\State \textsc{ChangeMainChain(\texttt{X})}
            			\colorcomment{If the new chain is denser}
			\EndIf
    			\EndIf
		\EndIf

\EndProcedure

\vspace{1mm}
\Procedure{IsValidBlock}{\texttt{X}} \colorcomment{returns true if a block is valid}
	\If { {\bf not} {\sc IsUnspent}(\texttt{X}.content.coin)}
	 \Return False
	\EndIf
    	\If{  \texttt{X}.header.time $>$ SystemTime} \Return False
	\EndIf
	\If{{\sc VRFverify}(\texttt{X}.header,
	\texttt{X}.content.hash,\texttt{X}.content.proof,\texttt{X}.content.coin.\text{{\sf VRF}.$pk$}) } \label{algo:vrfverify}
        		\State \Return True
   	\Else
        		\State \Return False
	\EndIf
\EndProcedure


\vspace{1mm}
\Procedure{Main}{ }
    \State \textsc{Initialize}()
    \State \textsc{StartThread(ReceiveMessage)} 
    \While{True}
        \State block =  \textsc{PosMining}(coin)
        \State\textsc{SendMessage}(block)  \colorcomment{Broadcast to the whole network}

    \EndWhile
\EndProcedure

\end{algorithmic}
}
\end{algorithm}

\end{document}